\documentclass{amsart}

\usepackage{amsfonts, amssymb, amsmath, amsthm}                               
\usepackage{mathrsfs}                                                         
\usepackage{setspace}                                                         
\usepackage{geometry}                                                         
\usepackage{xcolor}                                                           
\usepackage{graphicx}                                                         
\usepackage{slashed}                                                          
\usepackage{comment}                                                          

\newtheorem{theorem}{Theorem} 	      	      	                              
\newtheorem{corollary}[theorem]{Corollary}     	      	      	      	      
\newtheorem{lemma}[theorem]{Lemma}     	       	      	      	      	      
\newtheorem{proposition}[theorem]{Proposition} 	      	      	      	      
\newtheorem{definition}[theorem]{Definition} 	      	      	                
\newtheorem{question}[theorem]{Question}     	      	      	      	        
\newtheorem{remark}[theorem]{Remark}                                          

\numberwithin{equation}{section}                                              
\numberwithin{theorem}{section}                                               

\newcommand{\mf}[1]{\mathfrak{#1}}                                            
\newcommand{\mi}[1]{\mathscr{#1}}                                             
\newcommand{\mc}[1]{\mathcal{#1}}                                             
\newcommand{\ms}[1]{\mathsf{#1}}                                              
\newcommand{\mr}[1]{\mathrm{#1}}                                              
\newcommand{\mb}[1]{\mathbb{#1}}                                              

\newcommand{\N}{\mathbb{N}}                                                   
\newcommand{\R}{\mathbb{R}}                                                   

\newcommand{\trace}[1]{\mr{tr}_{#1}}                                          

\newcommand{\gv}{\mathsf{g}}                                                  
\newcommand{\hv}{\mathsf{h}}                                                  
\newcommand{\Rv}{\mathsf{R}}                                                  
\newcommand{\Dv}{\mathsf{D}}                                                  

\newcommand{\gm}{\mf{g}}                                                      
\newcommand{\hm}{\mf{h}}                                                      
\newcommand{\gs}{\bar{\mf{g}}}                                                
\newcommand{\Dm}{\mf{D}}                                                      
\newcommand{\Rm}{\mf{R}}                                                      

\newcommand{\nablam}{\bar{\nabla}}                                            
\newcommand{\Dvm}{\bar{\Dv}}                                                  
\newcommand{\Boxm}{\bar{\Box}}                                                
\newcommand{\hvm}{\bar{\hv}}                                                  
\newcommand{\hvd}{\hv^\ast}                                                   

\allowdisplaybreaks
\def\comp{1}

\begin{document}

\title[Null Geodesics, Improved Unique Continuation for Waves in aAdS Spacetimes]{Null Geodesics and Improved Unique Continuation for Waves in Asymptotically Anti-de Sitter Spacetimes}
\author{Alex McGill, Arick Shao}

\address{School of Mathematical Sciences\\
Queen Mary University of London\\
London E1 4NS\\ United Kingdom}
\email{a.mcgill@qmul.ac.uk}

\address{School of Mathematical Sciences\\
Queen Mary University of London\\
London E1 4NS\\ United Kingdom}
\email{a.shao@qmul.ac.uk}

\begin{abstract}
We consider the question of whether solutions of Klein--Gordon equations on asymptotically Anti-de Sitter spacetimes can be uniquely continued from the conformal boundary.
Positive answers were first given in \cite{hol_shao:uc_ads, hol_shao:uc_ads_ns}, under suitable assumptions on the boundary geometry and with boundary data imposed over a sufficiently long timespan.
The key step was to establish Carleman estimates for Klein--Gordon operators near the conformal boundary.

In this article, we further improve upon the above-mentioned results.
First, we establish new Carleman estimates---and hence new unique continuation results---for Klein--Gordon equations on a larger class of spacetimes than in \cite{hol_shao:uc_ads, hol_shao:uc_ads_ns}, in particular with more general boundary geometries.
Second, we argue for the optimality, in many respects, of our assumptions by connecting them to trajectories of null geodesics near the conformal boundary; these geodesics play a crucial role in the construction of counterexamples to unique continuation.
Finally, we develop a new covariant formalism that will be useful---both presently and more generally beyond this article---for treating tensorial objects with asymptotic limits at the conformal boundary.
\end{abstract}

\maketitle

\section{Introduction} \label{sec.intro}

In this paper, we revisit the problem of unique continuation for Klein--Gordon equations from the conformal boundary of asymptotically Anti-de Sitter spacetimes.
Our main objectives are:
\begin{enumerate}
\item To prove new Carleman estimates that extend the existing estimates and unique continuation results of \cite{hol_shao:uc_ads, hol_shao:uc_ads_ns} to a wider class of spacetimes.

\item To argue that the assumptions we impose in (1) are in many ways optimal, by connecting them to the construction of counterexamples to unique continuation.

\item To develop and then apply a new formalism for covariantly treating tensorial objects and their asymptotic limits at the conformal boundary.
\end{enumerate}

In recent years, asymptotically anti-de Sitter spacetimes have featured prominently in theoretical physics due to their connection to the AdS/CFT correspondence \cite{malda:ads_cft}.
This article is part of a larger program to establish mathematically rigorous statements related to AdS/CFT.

\subsection{Background}

The main settings of our discussions are portions of \emph{asymptotically Anti-de Sitter} (which we henceforth abbreviate as \emph{aAdS}) spacetimes, near their conformal boundaries.

To be more precise, we consider, as our spacetime manifold,
\begin{equation}
\label{eq.intro_manifold} \mi{M}^{ n + 1 } := ( 0, \rho_0 ) \times \mi{I}^n \text{,} \qquad \rho_0 > 0 \text{,}
\end{equation}
and we consider the following Lorentzian metric on $\mi{M}$,
\begin{equation}
\label{eq.intro_metric} g := \rho^{-2} [ d \rho^2 + \gv ( \rho ) ] \text{,}
\end{equation}
where $\rho$ denotes the coordinate for the component $( 0, \rho_0 )$, and where $\gv$ is a $\rho$-parametrized family of metrics on $\mi{I}$.
In particular, we refer to $( \mi{M}, g )$ as our \emph{aAdS spacetime}, and we refer to the specific form \eqref{eq.intro_metric} of the metric as a \emph{Fefferman--Graham gauge}.

Furthermore, we assume $\gv$ in \eqref{eq.intro_metric} has the asymptotic behavior
\begin{equation}
\label{eq.intro_asymp} \gv ( \rho ) = \gm + \rho^2 \gs + \mc{O} ( \rho^3 ) \text{,}
\end{equation}
where $\gm$ is a Lorentzian metric on $\mi{I}$, and where $\gs$ is another symmetric $( 0, 2 )$-tensor field on $\mi{I}$.
By convention, we refer to $( \mi{I}, \gm )$ as the \emph{conformal boundary} of our aAdS spacetime.

If one identifies $\mi{I}$ with $\{ 0 \} \times \mi{I}$, then $( \mi{I}, \gm )$ can be regarded as the (timelike) boundary $\{ \rho = 0 \}$ of $( \mi{M}, \rho^2 g )$.
As a consequence of this, we will, throughout this section, view the conformal boundary as a boundary manifold that is formally attached to $\mi{M}$ at $\rho \searrow 0$.

\begin{remark}
Note that we do not impose any restrictions on the topology or geometry of the conformal boundary.
Moreover, we do not assume $( \mi{M}, g )$ is Einstein-vacuum.
\end{remark}

Our notion of aAdS includes Anti-de Sitter (AdS) spacetime itself, as well as the Schwarzschild-AdS and Kerr-AdS families.
In particular, via an appropriate change of coordinates, these metrics can be expressed in the form \eqref{eq.intro_asymp} near the conformal boundary.
\footnote{See \cite[Section 3.3]{shao:aads_fg} for more detailed computations in the AdS and Schwarzschild-AdS settings.}
Also included are variants of the above families with different boundary topologies (e.g.~planar or toric Schwarzschild-AdS).

In addition, we assume the boundary $( \mi{I}, \gm )$ is foliated by a \emph{global time function} $t$ that splits $\mi{I}$ into the form $( t_-, t_+ ) \times \mc{S}^{ n - 1 }$.
This $t$ naturally extends to all of $\mi{M}$ by the assumption that it is constant along the $\rho$-direction.
Examples include the usual Schwarzschild-AdS time coordinate.

Consider now the following family of Klein--Gordon equations on $( \mi{M}, g )$,
\begin{equation}
\label{eq.intro_kg} ( \Box_g + \sigma ) \phi = \mc{G} ( \phi, \nabla \phi ) \text{,} \qquad \sigma \in \R \text{,}
\end{equation}
where $\mc{G} ( \phi, \nabla \phi )$ represents some appropriate lower-order linear or nonlinear terms that depend on the unknown $\phi$, which can be scalar or tensorial.
Our main problem for \eqref{eq.intro_kg} is the following:

\begin{question}[Unique continuation from the conformal boundary] \label{q.intro_uc}
Suppose $\phi$ satisfies \eqref{eq.intro_kg}, and suppose $\phi \rightarrow 0$, in an appropriate sense, as $\rho \searrow 0$ along some portion $\mi{D} \subseteq \mi{I}$.
Then, must $\phi$ also vanish in some neighborhood in $\mi{M}$ of the conformal boundary?
\footnote{Again, we view the conformal boundary as the manifold $\{ \rho = 0 \}$.}
\end{question}

In terms of wave equations, the conformal boundary serves as a timelike hypersurface ``at infinity" on which one imposes boundary data for \eqref{eq.intro_kg}; see \cite{breit_freedm:stability_sgrav, vasy:wave_aads, warn:wave_aads}, as well as the discussions in the introduction of \cite{hol_shao:uc_ads}.
In this context, \emph{the statement ``$\phi \rightarrow 0$ in an appropriate sense" can be roughly interpreted as both the Dirichlet and Neumann branches of $\phi$ vanishing as $\rho \searrow 0$}.
\footnote{In other words, $\phi$ has ``vanishing Cauchy data" on the conformal boundary.}

Question \ref{q.intro_uc} is a variant of the \emph{unique continuation} problem for PDEs, for which there is an extensive list of classical literature; see, for instance, \cite{cald:unique_cauchy, carl:uc_strong, hor:lpdo4, hor:uc_interp, robb_zuil:uc_interp, tat:uc_hh}.
However, these classical results fail to apply to Question \ref{q.intro_uc}, due to the following reasons:
\begin{itemize}
\item The crucial condition needed for the usual unique continuation results is \emph{pseudoconvexity}.
In our present context, the conformal boundary (barely) fails to be pseudoconvex with respect to $\Box_g$.
Thus, a suitable analysis of Question \ref{q.intro_uc} must take into account the degenerate situation in which one uniquely continues from a ``zero pseudoconvex" hypersurface.
See the introduction of \cite{hol_shao:uc_ads} for more extensive discussions in this direction.

\item It is often useful to express \eqref{eq.intro_kg} in terms of the conformal metric $\rho^2 g$, for which $( \mi{I}, \gm )$ now acts as a \emph{finite} boundary.
However, in this setting, the equation that is equivalent to \eqref{eq.intro_kg} contains a \emph{critically singular potential} whenever $4 \sigma \neq n^2 - 1$:
\[
( \Box_{ \rho^2 g } + C_\sigma \rho^{-2} ) \psi = \tilde{\mc{G}} ( \psi, \nabla \psi ) \text{,} \qquad C_\sigma \neq 0 \text{.}
\]
This drastically alters the analysis near the conformal boundary.
In particular, the Dirichlet and Neumann branches of $\psi$ (and hence $\phi$) behave like distinct nonzero powers of $\rho$ at $( \mi{I}, \gm )$; see \cite{vasy:wave_aads, warn:wave_aads} for well-posedness results, and see \cite{enc_shao_verga:obs_swave} for further discussions.
\end{itemize}

The first positive results for Question \ref{q.intro_uc} were established in \cite{hol_shao:uc_ads} for aAdS spacetimes with static conformal boundaries.
This was then extended to non-static boundaries in \cite{hol_shao:uc_ads_ns}.
An informal summary of these results is provided in the subsequent theorem:

\begin{theorem}[\cite{hol_shao:uc_ads, hol_shao:uc_ads_ns}] \label{thm.intro_uc}
Assume the following:
\begin{itemize}
\item $t$ is \emph{unit geodesic} on the conformal boundary.
\footnote{More specifically, the $\gm$-gradient of $t$ is a $\gm$-unit and $\gm$-geodesic vector field.}

\item There is some $p > 0$ such that $\mc{G}$, from \eqref{eq.intro_kg}, satisfies
\footnote{Here, the size of the first derivative, $| \nabla \phi |$, is measured with respect to $\rho$ and to coordinates on $\mi{I}$.}
\begin{equation}
\label{eq.intro_uc_G} | \mc{G} ( \phi, \nabla \phi ) |^2 \lesssim \rho^{ 4 + p } | \nabla \phi |^2 + \rho^{ 3 p } | \phi |^2 \text{.}
\end{equation}
\end{itemize}
In addition, assume that $( \mi{M}, g )$ satisfies the following \emph{pseudoconvexity criterion}:
\begin{itemize}
\item There exists $\zeta \in C^\infty ( \mi{I} )$ such that the following lower bound holds on $\mi{I}$:
\footnote{The left-hand sides of \eqref{eq.intro_uc_psc_1} and \eqref{eq.intro_uc_psc_2} are measured using a positive-definite metric on $\mi{I}$; see \cite[Definition 3.2]{hol_shao:uc_ads_ns}.}
\begin{equation}
\label{eq.intro_uc_psc_1} - \gs - \zeta \cdot \gm \geq K > 0 \text{.}
\end{equation}

\item The ``non-stationarity" of $\gm$ is sufficiently small with respect to $K$:
\footnote{Here, $\xi_K > 0$ is a constant whose value depends only on $K$.}
\begin{equation}
\label{eq.intro_uc_psc_2} | \mi{L}_t \gm | \leq \xi_K \text{.}
\end{equation}
\end{itemize}
If $\phi$ is a (scalar or tensorial) solution of \eqref{eq.intro_kg} such that
\begin{itemize}
\item $\phi$ is compactly supported on each level set of $( \rho, t )$, and

\item $\rho^{ - \kappa_0 } \phi \rightarrow 0$ in $C^1$ as $\rho \searrow 0$ over a \emph{sufficiently large timespan} $\{ t_0 \leq t \leq t_1 \}$ of the conformal boundary $\mi{I}$, with $\kappa_0$ depending on $\sigma$, and with $t_1 - t_0$ depending on $K$ and $\xi_K$,
\end{itemize}
then $\phi \equiv 0$ on some neighborhood in $\mi{M}$ of the boundary region $\{ t_0 \leq t \leq t_1 \} \subseteq \mi{I}$.
\end{theorem}

Of particular interest in Theorem \ref{thm.intro_uc} is the condition that $\phi$ vanishes as $\rho \searrow 0$ \emph{along a sufficiently large timespan of the conformal boundary}, which was first observed in \cite{hol_shao:uc_ads}.
That such an assumption is necessary can be seen in pure AdS spacetime.
Indeed, one can find null geodesics $\Lambda$ such that:
\begin{itemize}
\item $\Lambda$ starts from the conformal boundary time $t = 0$.

\item $\Lambda$ remains arbitrarily close to the conformal boundary.

\item $\Lambda$ returns to the conformal boundary at $t = \pi$.
\end{itemize}
As a result, one can (at least in the conformal case $4 \sigma = n^2 - 1$) apply the techniques of \cite{alin_baou:non_unique}---involving geometric optics constructions along the above geodesics $\Lambda$---to construct nontrivial solutions $\phi$ of equations in the form \eqref{eq.intro_kg} such that $\phi \rightarrow 0$ along a timespan $\delta \leq t \leq \pi - \delta$.
Therefore, on AdS spacetime, one generally requires vanishing on a timespan of at least $\pi$ for unique continuation.

\begin{remark}
From finite speed of propagation, it is clear that one requires vanishing on a timespan $t_1 - t_0 \geq \pi$ for a \emph{global} unique continuation result to the full AdS interior.
However, Theorem \ref{thm.intro_uc} suggests, more surprisingly, that this condition remains necessary even for \emph{local} unique continuation results, on an arbitrarily small neighborhood of the conformal boundary.
\end{remark}

\begin{remark}
It is, however, important to note that all the known counterexamples to unique continuation arising from \cite{alin_baou:non_unique} satisfy equations with \emph{complex-valued} potentials.
Whether counterexamples exist for purely real-valued wave equations remains an open question.
\end{remark}

\begin{remark}
On the other hand, when $\mc{G}$ is linear, and when both $g$ and $\mc{G}$ are analytic, the ``large enough timespan" condition in Theorem \ref{thm.intro_uc} is not needed due to Holmgren's theorem; see \cite{hor:lpdo2}.
\end{remark}

\begin{remark}
Schwarzschild-AdS spacetimes satisfy the pseudoconvexity criterion \eqref{eq.intro_uc_psc_1}--\eqref{eq.intro_uc_psc_2}, while planar and toric Schwarzschild-AdS fail this criterion; see \cite[Appendix B]{hol_shao:uc_ads_ns} for details.
\end{remark}

\begin{figure}[ht]
\includegraphics{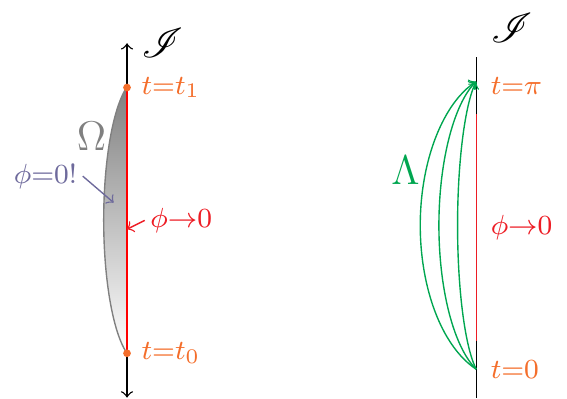}
\caption{The left graphic depicts the setting of Theorem \ref{thm.intro_uc}---if $\phi$ vanishes on a large enough timespan of $\mi{I}$ (in red), then $\phi$ also vanishes nearby in the interior ($\Omega$, in grey).
The right graphic depicts AdS spacetime---if $\phi \rightarrow 0$ on a small timespan (in red), then one can (in the conformal case $4 \sigma = n^2 - 1$) construct waves concentrated along null geodesics $\Lambda$ (in green) arbitrarily close to $\mi{I}$.}
\label{fig.intro_uc}
\end{figure}

The key ingredient for the proof of Theorem \ref{thm.intro_uc} is a novel geometric \emph{Carleman estimate} for the Klein--Gordon operator $\Box_g + \sigma$ near the conformal boundary.
\footnote{Since \cite{carl:uc_strong}, Carleman estimates have been a staple of unique continuation theory in the absence of analyticity.}
Informally speaking, this Carleman estimate is a weighted spacetime $L^2$-inequality of the form
\begin{equation}
\label{eq.intro_carleman} \int_\Omega e^{ - \lambda F } w_\Box | ( \Box_g + \sigma ) \phi |^2 \geq \mc{C} \lambda \int_{ \Omega } e^{ - \lambda F } w_1 | \nabla \phi |^2 + \mc{C} \lambda \int_{ \Omega } e^{ - \lambda F } w_0 | \phi |^2 \text{,}
\end{equation}
where $\Omega$ is an appropriate spacetime region (see Figure \ref{fig.intro_uc}) near the conformal boundary (in particular, compatible with the ``sufficiently large timespan" condition), where $\lambda > 0$ is a sufficiently large free parameter, and where $F$, $w_\Box$, $w_1$, $w_0$ are appropriate weights on $\Omega$.
The desired unique continuation result then follows from \eqref{eq.intro_carleman} via standard arguments.

For precise statements of the above results, the reader is referred to \cite[Definition 3.2]{hol_shao:uc_ads_ns} for the pseudoconvexity criterion, \cite[Theorem 3.7]{hol_shao:uc_ads_ns} and \cite[Theorem C.1]{hol_shao:uc_ads_ns} for the Carleman estimates, and \cite[Theorem 3.11]{hol_shao:uc_ads_ns} for the unique continuation result itself.

\subsection{Connections to Holography}

One perspective of the wave equations \eqref{eq.intro_kg}---and by extension Question \ref{q.intro_uc}---is as a model problem for the Einstein-vacuum equations.
In particular, a key motivation for the present article, as well as for \cite{hol_shao:uc_ads, hol_shao:uc_ads_ns}, is to build toward resolving the following unique continuation question, which is directly inspired by the AdS/CFT conjecture:

\begin{question} \label{q.intro_correspondence}
For a vacuum aAdS spacetime $( \mi{M}, g )$, does its ``boundary data" at the conformal boundary uniquely determine $g$?
\footnote{The precise description of ``boundary data" for the Einstein-vacuum equations is based on partial Fefferman--Graham expansions from the conformal boundary.
In physics terminology, this data corresponds to both the boundary metric and the boundary stress-energy tensor.
See \cite{shao:aads_fg} for further discussions on this point.}
In other words, is there a one-to-one correspondence between aAdS solutions of the Einstein-vacuum equations and an appropriate space of boundary data?
\end{question}

Elliptic analogues of Question \ref{q.intro_correspondence} have been proved---\cite{and_herz:uc_ricci, and_herz:uc_ricci_err, biq:uc_einstein} showed that the conformal boundary data for an asymptotically hyperbolic (Riemannian) Einstein metric uniquely determines the metric itself.
Moreover, \cite{chru_delay:uc_killing} extended these results to \emph{stationary} aAdS Einstein-vacuum spacetimes.

On the other hand, Question \ref{q.intro_correspondence} remains open in non-stationary settings.
The results of this article will be applied, in \cite{hol_shao:uc_ads_cft}, as a crucial step toward resolving Question \ref{q.intro_correspondence} in general.

Another interesting problem, closely related to Question \ref{q.intro_correspondence}, is that of \emph{symmetry extension}:

\begin{question} \label{q.intro_symmetry}
Assume $( \mi{M}, g )$ is a vacuum aAdS spacetime, and suppose its conformal boundary data has a symmetry.
\footnote{Again, ``boundary data" refers to both the boundary metric $\gm$ and the boundary stress-energy tensor $\mf{t}$.
Thus, by a symmetry of the boundary data, we mean a Killing vector field $\mf{X}$ on $( \mi{I}, \gm )$ such that $\mi{L}_{ \mf{X} } \mf{t} = 0$.}
Then, must this symmetry necessarily extend into $( \mi{M}, g )$?
\end{question}

A positive answer to Question \ref{q.intro_symmetry} was provided in \cite{chru_delay:uc_killing}, for the special case of \emph{stationary} spacetimes.
Again, the Carleman estimates of this article will play a prominent role, in the upcoming \cite{hol_shao:uc_ads_sym}, toward resolving Question \ref{q.intro_symmetry} in \emph{non-stationary} spacetimes.

\subsection{Novel Features}

As mentioned before, our main objective is to establish Carleman estimates near the conformal boundary that extend the estimates of \cite{hol_shao:uc_ads, hol_shao:uc_ads_ns}, and hence the results of Theorem \ref{thm.intro_uc}.
The main new features of our results are as follows:
\begin{enumerate}
\item Our estimates now apply for \emph{general time foliations} of $( \mi{I}, \gm )$.

\item We replace the assumptions \eqref{eq.intro_uc_psc_1} and \eqref{eq.intro_uc_psc_2} by a weaker \emph{null convexity criterion}.

\item We connect the null convexity criterion to \emph{geodesic trajectories} near the conformal boundary.

\item We only assume \emph{finite regularity} for our geometric quantities.

\item We develop a general formalism of \emph{vertical tensor fields} to treat the relevant tensorial quantities in our spacetime that have asymptotic limits at the conformal boundary.
\end{enumerate}
Below, we discuss each of these features in further detail:

\subsubsection{General Time Foliations}

First, recall that the results of \cite{hol_shao:uc_ads, hol_shao:uc_ads_ns} apply only when $( \mi{I}, \gm )$ is described in terms of a time function $t$ that is unit geodesic.
In particular, this is a special function for which one can find coordinates $( t, x^1, \dots, x^{ n - 1 } )$ on $\mi{I}$ satisfying the conditions
\begin{equation}
\label{eq.intro_geodesic_time} \gm_{ t t } \equiv -1 \text{,} \qquad \gm_{ t x^A } \equiv 0 \text{,} \qquad 1 \leq A < n \text{.}
\end{equation}

In contrast, \emph{the Carleman estimates in the present paper apply even when the conformal boundary is described in terms of a general time function $t$ that is not unit geodesic}.
As a result, the present results can be applied more flexibly to a larger class of settings.

\begin{remark}
For instance, the Boyer-Lindquist time coordinate of an Kerr-AdS spacetime defines a time $t$ on its conformal boundary for which \eqref{eq.intro_geodesic_time} fails to hold.
\end{remark}

\begin{remark}
Note that one can always find \emph{local} time functions $t$ on $( \mi{I}, \gm )$ for which \eqref{eq.intro_geodesic_time} is satisfied.
Thus, the assumption that $t$ is unit geodesic imposes an implicit restriction on the \emph{global} geometry of $( \mi{I}, \gm )$---namely, that it can be ruled by a family of timelike geodesics that do not contain any focal or cut locus points.
From this point of view, our main Carleman estimates extend the results of \cite{hol_shao:uc_ads, hol_shao:uc_ads_ns} by removing this extra technical requirement for $( \mi{I}, \gm )$.
\end{remark}

In \cite{hol_shao:uc_ads_ns}, the Lie derivative $\mi{L}_t \gm$ was potentially problematic and had to be treated carefully---in particular, one assumed in \eqref{eq.intro_uc_psc_2} that this was not too large.
Here, the corresponding quantity is $\Dm^2 t$, where $\Dm$ denotes the Levi-Civita connection with respect to the boundary metric $\gm$.
\footnote{Observe that $\Dm^2 t$ is equivalent to the Lie derivative of $\gm$ along the $\gm$-gradient of $t$.}

\subsubsection{The Null Convexity Criterion}

In \cite{hol_shao:uc_ads_ns}, the key assumption for unique continuation was the so-called \emph{pseudoconvexity criterion}, informally stated in \eqref{eq.intro_uc_psc_1} and \eqref{eq.intro_uc_psc_2}.
In particular, these bounds determined the timespan along which one must assume vanishing of the solution $\phi$ at the conformal boundary.
For our main estimates, we will further simplify and improve upon \eqref{eq.intro_uc_psc_1} and \eqref{eq.intro_uc_psc_2}.

First, rather than requiring that the quantity $- \gs - \zeta \gm$, for some function $\zeta$, is positive-definite in all directions, we instead assume that \emph{$- \gs$ itself is positive, but only along null directions}:
\begin{equation}
\label{eq.intro_uc_nc_1} - \gs ( Z, Z ) \geq C ( Z t )^2 \text{,} \qquad \gm ( Z, Z ) = 0 \text{.}
\end{equation}
Notice that \eqref{eq.intro_uc_psc_1} clearly implies \eqref{eq.intro_uc_nc_1}.

In fact, with some effort, one can also show that \eqref{eq.intro_uc_psc_1} and \eqref{eq.intro_uc_nc_1} are equivalent.
\footnote{This will be established within the proof of Proposition \ref{thm.psc_pointwise}.}
As a result, while \eqref{eq.intro_uc_nc_1} has a more natural interpretation than \eqref{eq.intro_uc_psc_1} and is simpler to check, it does not, in principle, enlarge the class of spacetimes on which our results apply.

On the other hand, we obtain a genuine improvement in the second portion of the pseudoconvexity criterion---we only require that the bound in \eqref{eq.intro_uc_psc_2} holds along null directions:
\begin{equation}
\label{eq.intro_uc_nc_2} | \Dm^2_{ Z Z } t | \leq B ( Z t )^2 \text{,} \qquad \gm ( Z, Z ) = 0 \text{.}
\end{equation}

We refer to \eqref{eq.intro_uc_nc_1} and \eqref{eq.intro_uc_nc_2} together as the \emph{null convexity criterion}; see Definition \ref{def.aads_nc} for the precise statements.
In Theorem \ref{thm.psc_nc}, we will show that the null convexity criterion implies a slightly weaker version of the pseudoconvexity criterion from \cite{hol_shao:uc_ads_ns}, which will then be used to establish our main Carleman estimates.
Moreover, in Theorem \ref{thm.geodesic}, we will connect the null convexity criterion to the absence of counterexamples to unique continuation.

\subsubsection{Geodesic Trajectories}

We had previously argued that the ``sufficiently large timespan" assumption in Theorem \ref{thm.intro_uc} was necessary for unique continuation in pure AdS spacetime.
One may then ask whether this observation also extends to general aAdS spacetimes.

In Section \ref{sec.geodesic}, we will provide an affirmative answer to this question.
More specifically, similar to the AdS setting, we construct a family of null geodesics $\Lambda$ such that:
\begin{itemize}
\item $\Lambda$ starts from the conformal boundary time at $t = \ell_-$.

\item $\Lambda$ remains arbitrarily close to the conformal boundary.

\item $\Lambda$ returns to the conformal boundary at $t = \ell_+$, with $\ell_+ - \ell_- \simeq 1$.
\end{itemize}
Thus, using the techniques of \cite{alin_baou:non_unique}, one can construct counterexamples to unique continuation (at least when $4 \sigma = n^2 - 1$) if $\phi$ vanishes along a timespan that is not sufficiently large.

In fact, we will prove a more specific result---\emph{we connect the timespan $\ell_+ - \ell_-$ of the geodesics $\Lambda$ to the null convexity criterion}.
In particular, we will establish the following:
\begin{itemize}
\item Assuming the inequalities \eqref{eq.intro_uc_nc_1} and \eqref{eq.intro_uc_nc_2}, we establish an upper bound $\mc{T}_+ ( B, C )$ on the timespan $\ell_+ - \ell_-$ of the near-boundary null geodesics $\Lambda$.
This gives a threshold timespan, beyond which the aforementioned counterexamples to unique continuation no longer exist.

\item Furthermore, we obtain a lower bound $\mc{T}_- ( B, C )$ on the timespan $\ell_+ - \ell_-$, implying that one cannot expect unique continuation results for timespans smaller than this bound.
\end{itemize}
The precise statements of these results are given in Theorem \ref{thm.geodesic}; moreover, the precise formulas for the threshold timespans $\mc{T}_\pm ( B, C )$ are given in Definition \ref{def.geodesic_timebound}.

The above-mentioned upper bound on $\ell_+ - \ell_-$ also coincides with the timespan that is needed for our main Carleman estimates.
As a result, our analysis will lead to the following statement of optimality:~\emph{the timespan beyond which the known counterexamples to unique continuation are no longer viable is the same as the timespan past which our unique continuation results hold}.

\subsubsection{Finite Regularity}

In \cite{hol_shao:uc_ads, hol_shao:uc_ads_ns}, the main results assumed the metric was smooth and satisfied asymptotic bounds at all orders near the conformal boundary.
In terms of current notations \eqref{eq.intro_metric} and \eqref{eq.intro_asymp}, the assumption was roughly that ``$\ms{E} := \mc{O} ( \rho^3 )$" in the right-hand side of \eqref{eq.intro_asymp} satisfied
\begin{equation}
\label{eq.intro_asymp_inf} | \partial_\rho^k \partial_{ x^{ a_1 } } \dots \partial_{ x^{ a_m } } \ms{E} | \lesssim \rho^{ 3 - k } \text{,} \qquad k \geq 0 \text{,}
\end{equation}
where $x^{ a_1 }, \dots, x^{ a_m }$ denote coordinate derivatives in directions tangent to $\mi{I}$.
\footnote{See \cite[Definition 2.4, Definition 2.5, Definition 2.8]{hol_shao:uc_ads_ns} for the precise asymptotic bounds.}

In contrast, \emph{the results here will require only finite regularity for the metric}.
More specifically, \emph{we will assume control of up to only three derivatives of $\gv$}; see Definition \ref{def.aads_strong} for precise statements.
This improvement in regularity does not rely on any new insights, as this could already be attained in the results of \cite{hol_shao:uc_ads, hol_shao:uc_ads_ns} by carefully tracking how many derivatives of the metric were used throughout.
In this paper, we choose to make explicit this more precise accounting.

A key motivation for stating our results in terms of finitely regular metrics is that the asymptotics at the conformal boundary for Einstein-vacuum metrics can be quite complicated.
In particular, when the boundary dimension $n$ is even, the expansion for $\gv$ from the conformal boundary generally become polyhomogeneous (i.e.~containing terms logarithmic in $\rho$) starting from the $n$-th order term.
This is a consequence of the Fefferman--Graham expansion \cite{fef_gra:conf_inv, fef_gra:amb_met} adapted to aAdS settings.
(See the companion paper \cite{shao:aads_fg} for rigorous results in finite regularity and for further discussions.)

\begin{remark}
In particular, whenever $n > 2$, our main results will not involve any derivatives of $\gv$ for which this polyhomogeneity begins to play a role.
\end{remark}

\subsubsection{Vertical Tensor Fields}

For future applications, it will be important that the Carleman estimates we prove also apply to tensorial waves.
In this regard, the results of \cite{hol_shao:uc_ads, hol_shao:uc_ads_ns} are stated for Klein--Gordon equations for which the unknown $\phi$ is a \emph{horizontal tensor field}---roughly, a tensor field on $\mi{M}$ which at each point is tangent to the corresponding level set of $( t, \rho )$.

Here, we instead state our results in terms of \emph{vertical tensor fields}---those tangent to the level sets of just $\rho$.
The main reason is that vertical fields represent the natural tensorial quantities for which one can make sense of asymptotic limits at the conformal boundary.
In particular, the Fefferman--Graham partial expansions \cite{shao:aads_fg} (which will play key roles in upcoming results \cite{hol_shao:uc_ads_cft, hol_shao:uc_ads_sym} in Einstein-vacuum settings) are defined in terms of vertical quantities.
Moreover, the Einstein equations can be far more easily derived in terms of vertical rather than horizontal fields.
\footnote{One obtains far fewer connection terms when decomposing spacetime tensorial quantities along one coordinate ($\rho$) rather than two coordinates ($\rho$ and $t$).
Also, because we assume the Fefferman--Graham gauge condition \eqref{eq.intro_metric}, tensorial decompositions in the $\rho$-coordinate tend to be relatively simple.}

Another aspect of our treatment of tensor fields is our adoption of a novel covariant formalism.
In particular, we view tensors as ``mixed", containing both ``vertical" and ``spacetime" components:
\begin{itemize}
\item The vertical components, which mainly describe the solution $\phi$, are treated using the conformally rescaled metrics $\gv$ that are finite on the conformal boundary $\mi{I}$.
This is convenient for more directly capturing the asymptotic behavior of $\phi$ near $\mi{I}$.

\item On the other hand, the spacetime components are treated using the physical metric $g$.
These components are primarily used to handle the wave operator $\Box_g$ and its structure.
\end{itemize}
In particular, using such a mixed formalism, we can make sense of a (spacetime) covariant wave operator acting on vertical tensor fields; see Definition \ref{def.aads_mixed_wave}.
\footnote{This is analogous to the wave operator acting on horizontal tensor fields in \cite{hol_shao:uc_ads, hol_shao:uc_ads_ns}.}
Furthermore, our system is naturally compatible with standard Leibniz rule and integration by parts formulas.
This observation will play a crucial role in streamlining the derivation of our main Carleman estimate.

\subsubsection{Order of Vanishing}

The choice to base our analysis around vertical tensor fields also comes with an inconvenience---the metrics $\gv$ are Lorentzian, and hence not positive-definite, on the level sets of $\rho$.
As a result of this, we construct a Riemannian metric $\hv$ (using $\gv$ and our time function $t$) in order to measure the sizes of vertical tensors; see Definition \ref{def.aads_riemann} for precise formulas.
However, $\hv$ is no longer compatible with the $\gv$-Levi-Civita connection, hence one obtains additional error terms (depending on derivatives of $t$) in our Carleman estimates when integrating by parts.

A consequence of this is that we must assume additional vanishing for $\phi$.
Whereas in \cite{hol_shao:uc_ads, hol_shao:uc_ads_ns}, the order of vanishing ($\kappa_0$ in Theorem \ref{thm.intro_uc}) depended only on $\sigma$, here we require
\begin{equation}
\label{eq.intro_vanish_higher} \rho^{ - \kappa } \phi \rightarrow 0 \text{,} \qquad \rho \searrow 0
\end{equation}
in $C^1$, \emph{for $\kappa$ large enough, depending on the rank of $\phi$ and on the derivatives of $t$}.
In this particular regard, the present results seem weaker than those found in \cite{hol_shao:uc_ads, hol_shao:uc_ads_ns}.

On the other hand, this deficiency is, in practice, only cosmetic; for our upcoming applications in correspondence and holography, one would not obtain stronger results using \cite{hol_shao:uc_ads, hol_shao:uc_ads_ns}.
The reason is that the equations for the relevant horizontal tensor fields contain lower-order coefficients that do not decay fast enough at $( \mi{I}, \gm )$ for the main Carleman estimates, \cite[Theorem 3.7]{hol_shao:uc_ads_ns}.
\footnote{In fact, the error terms arising from $\hv$ can be directly connected to these lower-order terms.}
To treat these lower-order terms, one must instead prove and apply a modified Carleman estimate---in particular, \cite[Theorem C.1]{hol_shao:uc_ads_ns}---which also requires higher-order vanishing of the form \eqref{eq.intro_vanish_higher}.

\begin{remark}
In fact, if $\phi$ is scalar, or if the boundary metric $\gm$ is stationary, then our results still hold while assuming the same vanishing as in \cite{hol_shao:uc_ads, hol_shao:uc_ads_ns}; see Remark \ref{rmk.intro_vanishing} below.
\end{remark}

\subsection{The Main Results}

We now provide informal statements of the main results of this article.
The first theorem is \emph{a Carleman estimate for Klein--Gordon operators, assuming the null convexity criterion \eqref{eq.intro_uc_nc_1} and \eqref{eq.intro_uc_nc_2}}; this directly leads to unique continuation results for \eqref{eq.intro_kg}.

\begin{theorem} \label{thm.intro_carleman}
Assume the following:
\begin{itemize}
\item There is some $p > 0$ such that the bound \eqref{eq.intro_uc_G} holds.

\item $( \mi{M}, g )$ satisfies the \emph{null convexity criterion} \eqref{eq.intro_uc_nc_1}--\eqref{eq.intro_uc_nc_2}.
\end{itemize}
Then, if $\phi$ is a (possibly tensorial) solution of \eqref{eq.intro_kg} such that:
\begin{itemize}
\item $\phi$ is compactly supported on each level set of $( \rho, t )$.

\item The limit \eqref{eq.intro_vanish_higher} holds over a \emph{large enough timespan} $t_0 \leq t \leq t_1$ of $\mi{I}$, where $t_1 - t_0$ depends on $B, C$ in \eqref{eq.intro_uc_nc_1}--\eqref{eq.intro_uc_nc_2}, and where $\kappa$ is sufficiently large depending on $t$ and $\phi$.
\footnote{Our Carleman estimates are also capable of handling the slowly decaying lower-order coefficients that were treated in \cite[Theorem C.1]{hol_shao:uc_ads_ns}; see Theorem \ref{thm.carleman}. In this case, $\kappa$ would also depend on these ``borderline" coefficients.}
\end{itemize}
Then, $\phi$ satisfies a Carleman estimate of the form \eqref{eq.intro_carleman}, where $\lambda > 0$ is sufficiently large, and where $\Omega$ is a spacetime region that is sufficiently close to the portion of $\mi{I}$ along which $\phi$ vanishes.

In particular, a corollary of the Carleman estimate is that $\phi \equiv 0$ on a portion of $\mi{M}$ that is sufficiently close to the region $\{ t_0 \leq t \leq t_1 \}$ of $\mi{I}$ along which $\phi$ vanishes.
\end{theorem}

The second main result concerns counterexamples to unique continuation.
More specifically, it states that the null convexity criterion (the same as in Theorem \ref{thm.intro_carleman}) also governs the trajectories of null geodesics near the conformal boundary and hence determines whether geometric optics counterexamples to unique continuation can be constructed as in \cite{alin_baou:non_unique}:

\begin{theorem} \label{thm.intro_geodesic}
Assume $( \mi{M}, g )$ satisfies the \emph{null convexity criterion} \eqref{eq.intro_uc_nc_1}--\eqref{eq.intro_uc_nc_2}.
Then, there exists a family of future null geodesics $\Lambda$ of $( \mi{M}, g )$ satisfying the following:
\begin{itemize}
\item $\Lambda$ starts from the conformal boundary at some time $t = t_0$.

\item $\Lambda$ remains arbitrarily close the conformal boundary.

\item $\Lambda$ returns to the conformal boundary before some fixed time $t = t_1$, where the timespan $t_1 - t_0$ is precisely the value from the statement of Theorem \ref{thm.intro_carleman}.
\end{itemize}
Furthermore, if one also has an upper bound
\begin{equation}
\label{eq.intro_uc_rnc} - \gs ( Z, Z ) \leq C' ( Z t )^2 \text{,} \qquad \gm ( Z, Z ) = 0 \text{,}
\end{equation}
then the null geodesics $\Lambda$ described above must also return to the conformal boundary after time $t = t_1'$, where $t_1' - t_0 > 0$ depends only on the constants $B, C'$.

As a consequence of the above, we have that:
\begin{itemize}
\item The usual geometric optics counterexamples to unique continuation (at least for $4 \sigma = n^2 - 1$) are no longer valid when $\phi$ vanishes on a timespan of more than $t_1 - t_0$.

\item The usual geometric optics counterexamples to unique continuation (at least for $4 \sigma = n^2 - 1$) are necessarily valid when $\phi$ vanishes on a timespan of less than $t_1' - t_0$.
\end{itemize}
\end{theorem}

The precise (and slightly stronger) versions of Theorems \ref{thm.intro_carleman} and \ref{thm.intro_geodesic} are given in Theorems \ref{thm.carleman} and \ref{thm.geodesic}, respectively.
A formal unique continuation result is stated in Corollary \ref{thm.carleman_uc}.

\begin{remark} \label{rmk.intro_vanishing}
A closer look at the proof of Theorem \ref{thm.intro_carleman} shows that the order of vanishing $\kappa$ in \eqref{eq.intro_vanish_higher} depends on the rank of $\phi$.
In particular, if $\phi$ is scalar, then the Riemannian metric $\ms{h}$ is not present in the Carleman estimates.
Consequently, in this case, Theorem \ref{thm.intro_carleman} can be shown to hold for the same optimal order of vanishing $\kappa = \kappa_0$ as in \cite{hol_shao:uc_ads, hol_shao:uc_ads_ns}.

Similarly, if the boundary metric is stationary with respect to $t$ (or, more precisely, if $\Dm^2 t \equiv 0$), then one can again show that Theorem \ref{thm.intro_carleman} holds for the optimal $\kappa = \kappa_0$.

Both of the above improvements are treated in the precise versions of our main results; see \eqref{eq.carleman_kappa}.
On the other hand, it is not known whether Theorem \ref{thm.intro_carleman} holds in general for $\kappa = \kappa_0$.
\end{remark}

\begin{remark}
As mentioned before, the known counterexamples to unique continuation are only directly applicable, in the present setting, to the conformal case $4 \sigma = n^2 - 1$.
Moreover, it is not yet known whether the timespan in Theorem \ref{thm.intro_carleman} can be further improved for other values of $\sigma$.
\end{remark}

\begin{remark}
In Theorem \ref{thm.intro_carleman}, one assumes $\phi$ vanishes on an entire time slab $\{ t_0 \leq t \leq t_1 \}$ of the conformal boundary.
It is not yet clear whether similar Carleman estimates or unique continuation results hold with vanishing on regions $\mi{D} \subseteq \mi{I}$ that are also spatially localized.

A more ambitious question, which may require microlocal methods, would be to characterize all such $\mi{D}$ directly by the behavior of null geodesics near the conformal boundary.
\end{remark}

We conclude with a brief discussion of the proofs of Theorems \ref{thm.intro_carleman} and \ref{thm.intro_geodesic}.

\subsubsection{Proof of Theorem \ref{thm.intro_carleman}}

Our argument is mostly based on the proofs of the Carleman estimates of \cite[Theorems 3.7 and C.1]{hol_shao:uc_ads_ns}.
However, our proof does contain various novel elements, due to our more general setting and our use of vertical tensor fields.
As a result, we focus on aspects that are exclusive to this article, and we refer the reader to \cite{hol_shao:uc_ads_ns} for further details.

A key part of the proof is to connect the null convexity criterion \eqref{eq.intro_uc_nc_1}, \eqref{eq.intro_uc_nc_2} to pseudoconvexity properties of $\Box_g$ near $\mi{I}$ that are crucial for unique continuation.
As previously mentioned, this is accomplished by relating the null convexity criterion to a variant of the pseudoconvexity criterion \eqref{eq.intro_uc_psc_1}, \eqref{eq.intro_uc_psc_2} of \cite{hol_shao:uc_ads_ns}.
(The is based on a projective geometric argument in the unique continuation literature for wave operators; see \cite{tat:notes_uc}.)
From this point, arguments similar to those in \cite{hol_shao:uc_ads_ns} allow us to construct a family of hypersurfaces that are pseudoconvex near the conformal boundary.

Another key aspect is the use of a Riemannian metric $\hv$ to measure the sizes of vertical tensor fields.
Since $\hv$ fails to be compatible with $g$-covariant derivatives, one encounters additional error terms containing derivatives of $\hv$.
Furthermore, some of these new terms become ``leading-order", in that they affect the order of vanishing required for $\phi$ in the Carleman estimate.

In fact, considerable care is needed in order to minimize the impact of these terms in the Carleman estimate.
While this requires a number of technical alterations to the arguments in \cite{hol_shao:uc_ads_ns}, the main idea, at a basic level, is to treat these dangerous terms at the same level as the first-order terms with ``borderline" decay in \cite[Theorem C.1]{hol_shao:uc_ads_ns} (which also affect the requisite order of vanishing).

Finally, we mention that our estimates for $\hv$, along with various tensor computations in the proof, make extensive use of our covariant formalism for ``mixed" and vertical tensor fields.

\subsubsection{Proof of Theorem \ref{thm.intro_geodesic}}

This argument revolves around the geodesic equations for $( \mi{M}, g )$---in particular, the equation corresponding to the $\rho$-component, which measures how close a geodesic is to the conformal boundary.
The key observation is that for small values of $\rho$ (that is, near $\mi{I}$), the leading terms of this equation describe a damped harmonic oscillator,
\footnote{This same damped harmonic oscillator also plays a crucial role in Theorem \ref{thm.intro_carleman}, as well as in the Carleman estimates of \cite{hol_shao:uc_ads_ns}.
In particular, \eqref{eq.intro_harmonic} is closely tied to the pseudoconvexity properties of $\Box_g$.}
\begin{equation}
\label{eq.intro_harmonic} \rho'' + A \rho' + B \rho \approx 0 \text{.}
\end{equation}
As a consequence of this, one expects that any ($g$-)null geodesic starting from and remaining sufficiently close to $\mi{I}$ will return to $\mi{I}$ after a finite timespan.

Another important observation is that the coefficients $A$ and $B$ in \eqref{eq.intro_harmonic} are directly connected to $\Dm^2 t$ and $\gs$, respectively.
One can thus apply the classical Sturm comparison theorem to \eqref{eq.intro_harmonic} in order to estimate the return time of null geodesics, from above and below, in terms of $\Dm^2 t$ and $\gs$.
Moreover, a closer analysis reveals that $A$ and $B$ depend only on the components of $\Dm^2 t$ and $\gs$ that are arbitrarily close to null.
As a result of this, the above-mentioned bounds on the geodesic return time can, in fact, be captured by the null convexity criterion.

What significantly complicates this analysis, however, is that the actual geodesic equations contain many nonlinear lower-order terms.
As a result of this, one must couple this Sturm comparison process with a carefully constructed continuity argument to ensure that the nonlinear terms remain negligible throughout the entire trajectory of the geodesic (and hence the Sturm comparison remains valid).
This coupling is the main novelty and technical difficulty of this proof.

\subsection{Organization of the Paper}

The remainder of the paper is organized as follows:
\begin{itemize}
\item In Section \ref{sec.aads}, we describe more precisely the aAdS spacetimes that we will study, as well as the quantities (e.g., vertical tensor fields) that we will analyze on these spacetimes.

\item In Section \ref{sec.psc}, we formally state (in Definition \ref{def.aads_nc}) the null convexity criterion.
We then connect (in Theorem \ref{thm.psc_nc}) the null convexity criterion to the pseudoconvexity criterion of \cite{hol_shao:uc_ads_ns}.
This will be used in the proof of our main Carleman estimate.

\item In Section \ref{sec.geodesic}, we discuss the trajectories of near-boundary null geodesics and its connection to the null convexity criterion.
In particular, we prove Theorem \ref{thm.intro_geodesic}.

\item In Section \ref{sec.carleman}, we describe the detailed setup for our main Carleman estimate, Theorem \ref{thm.intro_carleman}.
We then proceed to formally state (in Theorem \ref{thm.carleman}) and prove this estimate.
\end{itemize}
\if\comp1
Finally, Appendix \ref{sec.comp} contains proofs of some technical propositions from the main text.
\fi

\subsection{Acknowledgments}

The authors thank Gustav Holzegel for numerous discussions and advice.
A.S.~is supported by EPSRC grant EP/R011982/1.

\section{Asymptotically AdS Spacetimes} \label{sec.aads}

In this section, we define precisely the aAdS spacetimes that we will consider in this article.
\footnote{We note that a portion of this material was covered in \cite[Section 2.1]{shao:aads_fg}.
However, in order to keep the present article self-contained, we briefly review here the relevant parts of \cite{shao:aads_fg}.}

\subsection{Asymptotically AdS Manifolds} \label{sec.aads_mfld}

Our first task is to provide a precise description of the manifolds that we will study, as well as of the objects on them that we will analyze.

\begin{definition} \label{def.aads_manifold}
We define an \emph{aAdS region} to be a manifold with boundary of the form 
\begin{equation}
\label{eq.aads_manifold} \mi{M} := ( 0, \rho_0 ] \times \mi{I} \text{,} \qquad \rho_0 > 0 \text{,}
\end{equation}
where $\mi{I}$ is a smooth $n$-dimensional manifold, and where $n \in \N$.
\footnote{While we refer to $\mi{M}$ as the aAdS region, we always also implicitly assume the associated quantities $n$, $\mi{I}$, $\rho_0$.}

Furthermore, given an aAdS region $\mi{M}$ as in \eqref{eq.aads_manifold}:
\begin{itemize}
\item We let $\rho$ denote the function on $\mi{M}$ that projects onto its $( 0, \rho_0 ]$-component.

\item We let $\partial_\rho$ denote the lift to $\mi{M}$ of the canonical vector field $d_\rho$ on $( 0, \rho_0 ]$.
\end{itemize}
\end{definition}

\begin{definition} \label{def.aads_vertical}
We define the \emph{vertical bundle} $\ms{V}^k_l \mi{M}$ of rank $( k, l )$ over $\mi{M}$ to be the manifold of all tensors of rank $( k, l )$ on each level set of $\rho$ in $\mi{M}$:
\footnote{We use the standard notation $T^k_l$ to denote the usual tensor bundle of rank $( k, l )$ over a manifold.
As usual, $k$ refers to the contravariant rank, while $l$ refers to the covariant rank.}
\begin{equation}
\label{eq.aads_vertical} \ms{V}^k_l \mi{M} = \bigcup_{ \sigma \in ( 0, \rho_0 ] } T^k_l \{ \rho = \sigma \} \text{.}
\end{equation}
Moreover, we refer to sections of $\ms{V}^k_l \mi{M}$ as \emph{vertical tensor fields} of rank $( k, l )$.
\end{definition}

Observe that a vertical tensor field of rank $( k, l )$ on an aAdS region $\mi{M}$ can also be interpreted as a one-parameter family, indexed by $\rho \in ( 0, \rho_0 ]$, of rank $( k, l )$ tensor fields on $\mi{I}$.

\begin{definition} \label{def.aads_tensor}
Similar to \cite{shao:aads_fg}, we adopt the following conventions regarding tensor fields:
\begin{itemize}
\item We use italicized font (such as $g$) for tensor fields on $\mi{M}$.

\item We use serif font (such as $\gv$) for vertical tensor fields.

\item We use Fraktur font (such as $\mf{g}$) for tensor fields on $\mi{I}$.
\end{itemize}
Also, unless otherwise stated, we assume that a given tensor field is smooth.
\end{definition}

\begin{definition} \label{def.aads_tensor_id}
Throughout, we will adopt the following natural identifications of tensor fields:
\begin{itemize}
\item Given a tensor field $\mf{A}$ on $\mi{I}$, we will also use $\mf{A}$ to denote the vertical tensor field on $\mi{M}$ obtained by extending $\mf{A}$ as a $\rho$-independent family of tensor fields on $\mi{I}$.

\item In particular, a scalar function on $\mi{I}$ also defines a $\rho$-independent function on $\mi{M}$.

\item In addition, any vertical tensor field $\ms{A}$ can be uniquely identified with a tensor field on $\mi{M}$ (of the same rank) via the following rule:~the contraction of any component of $\ms{A}$ with $\partial_\rho$ or $d \rho$ (whichever is appropriate) is defined to vanish identically.
\end{itemize}
\end{definition}

\begin{definition} \label{def.aads_vertical_lie}
Let $\mi{M}$ be an aAdS region, and let $\ms{A}$ be a vertical tensor field.
\begin{itemize}
\item Given any $\sigma \in ( 0, \rho_0 ]$, we let $\ms{A} |_\sigma$ denote the tensor field on $\mi{I}$ obtained from restricting $\ms{A}$ to the level set $\{ \rho = \sigma \}$ (and then naturally identifying $\{ \rho = \sigma \}$ with $\mi{I}$).

\item We define the $\rho$-\emph{Lie derivative} of $\ms{A}$, denoted $\mi{L}_\rho \ms{A}$, to be the vertical tensor field satisfying
\begin{equation}
\label{eq.aads_vertical_lie} \mi{L}_\rho \ms{A} |_\sigma = \lim_{ \sigma' \rightarrow \sigma } ( \sigma' - \sigma )^{-1} ( \ms{A} |_{ \sigma' } - \ms{A} |_\sigma ) \text{,} \qquad \sigma \in ( 0, \rho_0 ] \text{.}
\end{equation}
\end{itemize}
\end{definition}

Next, we establish our conventions for coordinate systems on $\mi{I}$ and $\mi{M}$:

\begin{definition} \label{def.aads_index}
Let $\mi{M}$ be an aAdS region, and let $( U, \varphi )$ be a coordinate system on $\mi{I}$:
\begin{itemize}
\item Let $\varphi_\rho := ( \rho, \varphi )$ denote the corresponding lifted coordinates on $( 0, \rho_0 ] \times U$.

\item We use lower-case Latin indices $a, b, c, \dots$ to denote $\varphi$-coordinate components, and we use the symbols $x^a, x^b, x^c, \dots$ to denote $\varphi$-coordinate functions.

\item We use lower-case Greek indices $\alpha, \beta, \mu, \nu, \dots$ to denote $\varphi_\rho$-coordinate components.

\item Repeated indices will indicate summations over the appropriate components.
\end{itemize}
\end{definition}

\begin{definition} \label{def.aads_coord}
Let $\mi{M}$ be an aAdS region.
A coordinate system $( U, \varphi )$ on $\mi{I}$ is called \emph{compact} iff:
\begin{itemize}
\item $\bar{U}$ is a compact subset of $\mi{I}$.

\item $\varphi$ extends smoothly to (an open neighborhood of) $\bar{U}$.
\end{itemize}
\end{definition}

We now recall the notions from \cite{shao:aads_fg} of local size and convergence for vertical tensor fields:

\begin{definition} \label{def.aads_limit}
Let $\mi{M}$ be an aAdS region, and fix $M \geq 0$.
In addition, let $\ms{A}$ and $\mf{A}$ be a vertical tensor field and a tensor field on $\mi{I}$, respectively, both of the same rank $(k, l)$.
\begin{itemize}
\item Given a compact coordinate system $( U, \varphi )$ on $\mi{I}$, we define (with respect to $\varphi$-coordinates)
\begin{equation}
\label{eq.aads_norm} | \ms{A} |_{ M, \varphi } := \sum_{ m = 0 }^M \sum_{ \substack{ a_1, \dots, a_m \\ b_1, \dots, b_k \\ c_1, \dots, c_l } } | \partial^m_{ a_1 \dots a_m } \ms{A}^{ b_1 \dots b_k }_{ c_1 \dots c_l } | \text{.}
\end{equation}

\item $\ms{A}$ is \emph{locally bounded in $C^M$} iff for any compact coordinate system $( U, \varphi )$ on $\mi{I}$,
\begin{equation}
\label{eq.aads_bounded} \sup_{ ( 0, \rho_0 ] \times U } | \ms{A} |_{ M, \varphi } < \infty \text{.}
\end{equation}

\item We write $\ms{A} \rightarrow^M \mf{A}$ iff for any compact coordinate system $( U, \varphi )$ on $\mi{I}$,
\begin{equation}
\label{eq.aads_limit} \lim_{ \sigma \searrow 0 } \sup_{ \{ \sigma \} \times U } | \ms{A} - \mf{A} |_{ M, \varphi } = 0 \text{.}
\end{equation}
\end{itemize}
\end{definition}

\subsection{Asymptotically AdS Metrics} \label{sec.aads_fg}

We now recall the notion of ``FG-aAdS segments", as well as the Fefferman--Graham gauge condition, from \cite{shao:aads_fg}.
This represents the minimal conditions needed for a spacetime to reasonably be considered ``asymptotically anti-de Sitter".

\begin{definition} \label{def.aads_metric}
$( \mi{M}, g )$ is called an \emph{FG-aAdS segment} iff the following hold:
\begin{itemize}
\item $\mi{M}$ is an aAdS region, and $g$ is a Lorentzian metric on $\mi{M}$.

\item There exists a vertical tensor field $\gv$ of rank $( 0, 2 )$ such that
\footnote{In \eqref{eq.aads_metric}, we identified $\gv$ with a $\rho$-trivial spacetime tensor field via Definition \ref{def.aads_tensor_id}.}
\begin{equation}
\label{eq.aads_metric} g := \rho^{-2} ( d \rho^2 + \gv ) \text{.}
\end{equation}

\item There exists a Lorentzian metric $\gm$ on $\mi{I}$ such that
\begin{equation}
\label{eq.aads_metric_limit} \gv \rightarrow^0 \gm \text{.}
\end{equation}
\end{itemize}

Furthermore, given such an FG-aAdS segment $( \mi{M}, g )$:
\footnote{While we refer to $( \mi{M}, g )$ as the FG-aAdS segment, this also implicitly includes $\gv$ and $\gm$ from \eqref{eq.aads_metric} and \eqref{eq.aads_metric_limit}.}
\begin{itemize}
\item We refer to the form \eqref{eq.aads_metric} for $g$ as the \emph{Fefferman--Graham} (or \emph{FG}) \emph{gauge condition}.

\item $( \mi{I}, \gm )$ is called the \emph{conformal boundary} associated with $( \mi{M}, g, \rho )$.
\end{itemize}
\end{definition}

\begin{remark}
We note that the conformal boundary $( \mi{I}, \gm )$ is ``gauge-dependent" and depends on the choice of boundary-defining function $\rho$.
See  \cite{deharo_sken_solod:holog_adscft, imbim_schwim_theis_yanki:diffeo_holog} for further discussions of this point.
\end{remark}

The following definitions establish notations for some standard geometric objects:

\begin{definition} \label{def.aads_covar}
Given an FG-aAdS segment $( \mi{M}, g )$:
\begin{itemize}
\item Let $g^{-1}$, $\nabla$, and $R$ denote the metric dual, the Levi-Civita connection, and the Riemann curvature tensor (respectively) associated with the spacetime metric $g$.

\item Let $\gm^{-1}$, $\Dm$, and $\Rm$ denote the metric dual, the Levi-Civita connection, and the Riemann curvature tensor (respectively) associated with the boundary metric $\gm$.

\item Similar to \cite[Section 2.1]{shao:aads_fg}, we let $\gv^{-1}$, $\Dv$, and $\Rv$ denote the metric dual, the Levi-Civita connection, and the Riemann curvature (respectively) for the vertical metric $\gv$.
\footnote{In particular, $\gv^{-1}$ is the rank $( 2, 0 )$ vertical tensor field satisfying $\gv^{-1} |_\sigma := ( \gv |_\sigma )^{-1}$ for each $\sigma \in ( 0, \rho_0 ]$, while $\Dv$ acts like the Levi-Civita connection associated with $\gv |_\sigma$ on any level set $\{ \rho = \sigma \}$.}
\end{itemize}
In addition, we omit the superscript ``${}^{-1}$" when expressing a metric dual in index notion.
\end{definition}

\begin{definition} \label{def.aads_gradient}
Let $( \mi{M}, g )$ be an FG-aAdS segment.
We then define the following:
\begin{itemize}
\item Let $\nabla^\sharp$ denote the $g$-gradient operator (that is, the $g$-dual of $\nabla$).

\item Let $\Dm^\sharp$ and $\trace{\gm}$ denote the $\gm$-gradient and $\gm$-trace operators, respectively.

\item Let $\Dv^\sharp$ and $\trace{\gv}$ denote the $\gv$-gradient and $\gv$-trace operators, respectively.
\footnote{In particular, on each level set $\{ \rho = \sigma \}$, $\sigma \in ( 0, \rho_0 ]$, the operator $\Dv^\sharp$ behaves like the $( \gv |_\sigma )$-gradient operator.}
\end{itemize}
\end{definition}

For our upcoming results, we will require a stronger notion of ``FG-aAdS segments" which also involves boundary limits for derivatives of the metric.

\begin{definition} \label{def.aads_strong}
$( \mi{M}, g )$ is called a \emph{strongly FG-aAdS segment} iff the following hold:
\begin{itemize}
\item $( \mi{M}, g )$ is an FG-aAdS segment.

\item There exists a symmetric rank $( 0, 2 )$ tensor field $\gs$ on $\mi{I}$ such that the following hold:
\begin{equation}
\label{eq.aads_strong_limits} \gv \rightarrow^3 \gm \text{,} \qquad \mi{L}_\rho \gv \rightarrow^2 0 \text{,} \qquad \mi{L}_\rho^2 \gv \rightarrow^1 \gs \text{.}
\end{equation}

\item $\mi{L}_\rho^3 \gv$ is locally bounded in $C^0$.
\end{itemize}
\end{definition}

\begin{remark}
Definition \ref{def.aads_strong} can also be connected to the main result of \cite{shao:aads_fg}.
More specifically, \cite[Theorem 3.3]{shao:aads_fg} implies that if $( \mi{M}, g )$ is an FG-aAdS segment, and if
\begin{itemize}
\item $( \mi{M}, g )$ is Einstein-vacuum, and

\item $\gv$ and $\mi{L}_\rho \gv$ are locally bounded in $C^{ n + 4 }$ and $C^0$, respectively,
\end{itemize}
then $( \mi{M}, g )$ is also a strongly FG-aAdS segment.
\end{remark}

The estimates of \cite[Proposition 2.36, Proposition 2.37]{shao:aads_fg} immediately imply the following:

\begin{proposition} \label{thm.aads_geom_limit}
Let $( \mi{M}, g )$ be a strongly FG-aAdS segment.
\begin{itemize}
\item $\gv^{-1}$ and $\Rv$ are locally bounded in $C^3$ and $C^1$, respectively, and
\begin{equation}
\label{eq.aads_geom_limit} \gv^{-1} \rightarrow^3 \gm^{-1} \text{,} \qquad \Rv \rightarrow^1 \Rm \text{.}
\end{equation}

\item Let $\ms{A}$ be a vertical tensor field, and let $\mf{A}$ be a tensor field on $\mi{I}$ of the same rank, and assume that $\ms{A} \rightarrow^M \mf{A}$ for some $M > 0$.
Then, the following limits hold:
\begin{equation}
\label{eq.aads_geom_limit_deriv} \Dv^m \ms{A} \rightarrow^0 \Dm^m \mf{A} \text{,} \qquad 0 < m \leq \min ( M, 3 ) \text{.}
\end{equation}
\end{itemize}
\end{proposition}

\subsection{Time Foliations} \label{sec.aads_nc}

For our main results, we will also need an appropriate global measure of time for our setting.
In particular, this can be viewed as a partial gauge choice, and it plays a similar role as in \cite{hol_shao:uc_ads, hol_shao:uc_ads_ns}, except we allow for a much more general class of time functions here.

\begin{definition} \label{def.aads_time}
Let $( \mi{M}, g )$ be an FG-aAdS segment.
A smooth function $t: \mi{I} \rightarrow \R$ is called a \emph{global time} for $( \mi{M}, g )$ iff the following conditions hold:
\begin{itemize}
\item The nonempty level sets of $t$ are Cauchy hypersurfaces of $( \mi{I}, \gm )$.

\item $\Dv^\sharp t$ is uniformly timelike---there exists $C_t > 1$ such that
\footnote{Here, $t$ is extended in a $\rho$-independent manner to $\mi{M}$ as in Definition \ref{def.aads_tensor_id}.}
\begin{equation}
\label{eq.aads_time} C_t^{-1} \leq - \gv ( \Dv^\sharp t, \Dv^\sharp t ) \leq C_t \text{.}
\end{equation}
\end{itemize}

Furthermore, whenever $t$ is such a global time for $( \mi{M}, g )$, we define the shorthands
\footnote{We note that standard results in Lorentzian geometry (see \cite{on:srg, wald:gr}) imply $\mi{I}$ is diffeomorphic to $( t_-, t_+ ) \times \mc{S}$ for some $( n - 1 )$-dimensional manifold $\mc{S}$, with $t$ being the projection onto the first component $( t_-, t_+ )$.}
\begin{equation}
\label{eq.aads_timespan} t_+ := \sup_{ \mi{I} } t \text{,} \qquad t_- := \inf_{ \mi{I} } t \text{.}
\end{equation}
\end{definition}

\begin{remark} \label{rmk.aads_smooth}
Though we implicitly assume in our development that geometric quantities are smooth, this condition can be considerably weakened.
In particular, all our results still hold when the vertical metric $\gv$ is only $C^3$, and when the global time $t$ is $C^4$.
\footnote{In all the upcoming proofs, one takes at most three derivatives of $\gv$ and four derivatives of $t$.}
\end{remark}

Compactness will be a crucial ingredient for many of our main results.
As a result, we will often restrict our attention to the following class of boundary domains:

\begin{definition} \label{def.aads_scompact}
Let $( \mi{M}, g )$ be an FG-aAdS segment, let $t$ be a global time, and let $\mi{D}$ be an open subset of $\mi{I}$.
We say that $\bar{\mi{D}}$ has \emph{compact cross-sections} iff the following sets are compact:
\begin{equation}
\label{eq.aads_scompact} \bar{\mi{D}} \cap \{ t = \tau \} \text{,} \qquad t_- < \tau < t_+ \text{.}
\end{equation}
\end{definition}

Next, we note that a global time naturally induces Riemannian metrics on our setting.
These provide ways to measure, in a coordinate-independent manner, the sizes of tensor fields.

\begin{definition} \label{def.aads_riemann}
Let $( \mi{M}, g )$ be an FG-aAdS segment, and let $t$ be a global time.
\begin{itemize}
\item We define the vertical Riemannian metric associated with $( \gv, t )$ by
\begin{equation}
\label{eq.aads_riemann_vertical} \hv := \gv - \frac{2}{ \gv ( \Dv^\sharp t, \Dv^\sharp t ) } \, dt^2 \text{.}
\end{equation}
\item We define the boundary Riemannian metric associated with $( \gm, t )$ by
\begin{equation}
\label{eq.aads_riemann} \hm := \gm - \frac{2}{ \gm ( \Dm^\sharp t, \Dm^\sharp t ) } \, dt^2 \text{.}
\end{equation}
\end{itemize}
\end{definition}

Finally, we define some additional tensorial notations that will be particularly useful in the statement and proof of our main Carleman estimates.

\begin{definition} \label{def.aads_tensor_norm}
Let $( \mi{M}, g )$ be an FG-aAdS segment, and let $t$ be a global time.
Furthermore, let $\hv$ denote the associated vertical Riemannian metric from \eqref{eq.aads_riemann_vertical}.
\begin{itemize}
\item Given vertical tensor fields $\ms{A}$ and $\ms{C}$ of dual ranks $( k, l )$ and $( l, k )$, respectively, we let $\langle \ms{A}, \ms{C} \rangle$ denote the \emph{full contraction} of $\ms{A}$ and $\ms{C}$, that is, the scalar quantity obtained by contracting all the corresponding components of $\ms{A}$ and $\ms{C}$.

\item For a vertical tensor field $\ms{A}$ of rank $( k, l )$, we let $\hvd \ms{A}$ denote the \emph{full $\hv$-dual} of $\ms{A}$---the rank $( l, k )$ vertical tensor field obtained by raising and lowering all indices of $\ms{A}$ using $\hv$.

\item We define the (positive-definite) bundle metric $\bar{\hv}$ on the vertical bundle $\ms{V}^k_l \mi{M}$ as follows:~for any vertical tensor fields $\ms{A}, \ms{B}$ of rank $( k, l )$, we define
\begin{equation}
\label{eq.aads_tensor_inner} \bar{\hv} ( \ms{A}, \ms{B} ) := \langle \hvd \ms{A}, \ms{B} \rangle = \langle \ms{A}, \hvd \ms{B} \rangle \text{.}
\end{equation}

\item Furthermore, for any vertical tensor field $\ms{A}$, we define the shorthand
\begin{equation}
\label{eq.aads_tensor_norm} | \ms{A} |_{ \hv }^2 := \bar{\hv} ( \ms{A}, \ms{A} ) \text{.}
\end{equation}
\end{itemize}
\end{definition}

\begin{proposition} \label{thm.aads_tensor_norm}
Let $( \mi{M}, g )$ be an FG-aAdS segment, and let $t$ be a global time.
Then:
\begin{itemize}
\item For any vertical tensor fields $\ms{A}$ and $\ms{B}$,
\begin{equation}
\label{eq.aads_tensor_norm_product} | \hvd \ms{A} |_\hv = | \ms{A} |_\hv \text{,} \qquad | \ms{A} \otimes \ms{B} |_\hv \leq | \ms{A} |_\hv | \ms{B} |_\hv \text{.}
\end{equation}

\item For any vertical tensor fields $\ms{A}$, $\ms{B}$, and $\ms{C}$ of ranks $( k, l )$, $( k, l )$, and $( l, k )$, respectively:
\begin{equation}
\label{eq.aads_tensor_norm_contract} | \langle \ms{A}, \ms{C} \rangle | \leq | \ms{A} |_\hv | \ms{C} |_\hv \text{,} \qquad | \hvm ( \ms{A}, \ms{B} ) | \leq | \ms{A} |_\hv | \ms{B} |_\hv \text{.}
\end{equation}
\end{itemize}
\end{proposition}

\begin{proof}
These follow immediately from Definition \ref{def.aads_tensor_norm}, once one expands each of the left-hand sides of \eqref{eq.aads_tensor_norm_product} and \eqref{eq.aads_tensor_norm_contract} using an $\hv$-orthonormal frame and coframe.
\end{proof}

\subsection{The Mixed Tensor Calculus} \label{sec.aads_mixed}

In order for our upcoming Carleman estimates to apply to general vertical tensor fields, we must first make sense of a $g$-wave operator acting on a vertical tensor field.
Moreover, we wish to achieve this in a manner that is compatible with the standard covariant operations.
For this, we adopt an approach similar in nature to that of \cite{hol_shao:uc_ads, hol_shao:uc_ads_ns} (except that we now work with vertical, rather than horizontal, tensor fields).

The first step is to construct natural connections on the vertical bundles.
Perhaps the most explicit and concise method for doing this is through coordinates and index notation:

\begin{definition} \label{def.aads_vertical_connection}
Let $( \mi{M}, g )$ be an FG-aAdS segment, and let $( U, \varphi )$ denote a coordinate system on $\mi{I}$.
We make the following preliminary definitions with respect to $\varphi$ and $\varphi_\rho$-coordinates:
\begin{itemize}
\item For any indices $a, b, c$, we define the coefficients
\begin{equation}
\label{eq.aads_vertical_christoffel} \bar{\Gamma}^a_{ c b } := \frac{1}{2} \gv^{ a d } ( \partial_c \gv_{ d b } + \partial_b \gv_{ d c } - \partial_d \gv_{ b c } ) \text{,} \qquad \bar{\Gamma}^a_{ \rho b } := \frac{1}{2} \gv^{ a c } \mi{L}_\rho \gv_{ c b } \text{.}
\end{equation}

\item For any vertical tensor field $\ms{A}$ of rank $( k, l )$, we define
\begin{align}
\label{eq.aads_vertical_connection} \bar{\Dv}_c \ms{A}^{ a_1 \dots a_k }_{ b_1 \dots b_l } &= \partial_c ( \ms{A}^{ a_1 \dots a_k }_{ b_1 \dots b_l } ) + \sum_{ i = 1 }^k \bar{\Gamma}^{ a_i }_{ c d } \, \ms{A}^{ a_1 \hat{d}_i a_k }_{ b_1 \dots b_l } - \sum_{ j = 1 }^l \bar{\Gamma}^d_{ c b_j } \, \ms{A}^{ a_1 \dots a_k }_{ b_1 \hat{d}_j b_l } \text{,} \\
\notag \bar{\Dv}_\rho \ms{A}^{ a_1 \dots a_k }_{ b_1 \dots b_l } &= \mi{L}_\rho \ms{A}^{ a_1 \dots a_k }_{ b_1 \dots b_l } + \sum_{ i = 1 }^k \bar{\Gamma}^{ a_i }_{ \rho d } \, \ms{A}^{ a_1 \hat{d}_i a_k }_{ b_1 \dots b_l } - \sum_{ j = 1 }^l \bar{\Gamma}^d_{ \rho b_j } \, \ms{A}^{ a_1 \dots a_k }_{ b_1 \hat{d}_j b_l } \text{,}
\end{align}
where the symbols $a_1 \hat{d}_i a_k$ and $b_1 \hat{d}_j b_l$ denote the sequences $a_1 \dots a_k$ and $b_1 \dots b_l$ of indices, respectively, except with with $a_i$ and $b_j$ replaced by $d$.
\end{itemize}
\end{definition}

\begin{proposition} \label{thm.aads_vertical_connection}
Let $( \mi{M}, g )$ be an FG-aAdS segment.
Then, the relations \eqref{eq.aads_vertical_connection} define a (unique) family of connections $\bar{\Dv}$ on the vertical bundles $\ms{V}^k_l \mi{M}$, for all ranks $( k, l )$.
Furthermore:
\begin{itemize}
\item For any vertical vector field $\ms{Y}$ (i.e., having rank $( 1, 0 )$) and any vertical tensor field $\ms{A}$,
\begin{equation}
\label{eq.aads_vertical_connection_extend} \bar{\Dv}_{ \ms{Y} } \ms{A} = \Dv_{ \ms{Y} } \ms{A} \text{.}
\end{equation}

\item For any smooth $\ms{a}: \mi{M} \rightarrow \R$ and any vector field $X$ on $\mi{M}$,
\footnote{Note $\ms{a}$ is a vertical tensor field of rank $( 0, 0 )$.}
\begin{equation}
\label{eq.aads_vertical_connection_scalar} \bar{\Dv}_X \ms{a} = X \ms{a} \text{,}
\end{equation}

\item For any vector field $X$ on $\mi{M}$ and any vertical tensor fields $\ms{A}$ and $\ms{B}$,
\begin{equation}
\label{eq.aads_vertical_connection_leibniz} \bar{\Dv}_X ( \ms{A} \otimes \ms{B} ) = \bar{\Dv}_X \ms{A} \otimes \ms{B} + \ms{A} \otimes \bar{\Dv}_X \ms{B} \text{.}
\end{equation}

\item For any vector field $X$ on $\mi{M}$, vertical tensor field $\ms{A}$, and tensor contraction $\mc{C}$,
\footnote{$\mc{C}$ is an operation mapping vertical tensors of rank $( k + 1, l + 1 )$ to those of rank $( k, l )$ via index summation.}
\begin{equation}
\label{eq.aads_vertical_connection_contract} \bar{\Dv}_X ( \mc{C} \ms{A} ) = \mc{C} ( \bar{\Dv}_X \ms{A} ) \text{.}
\end{equation}

\item For any vector field $X$ on $\mi{M}$,
\begin{equation}
\label{eq.aads_vertical_connection_compat} \bar{\Dv}_X \gv = 0 \text{,} \qquad \bar{\Dv}_X \gv^{-1} = 0 \text{.}
\end{equation}
\end{itemize}
\end{proposition}

\if\comp1

\begin{proof}
See Appendix \ref{sec.aads_vertical_connection}.
\end{proof}

\fi

In short, the connections $\bar{\Dv}$ in Proposition \ref{thm.aads_vertical_connection} extend the vertical Levi-Civita connections $\Dv$ to allow covariant derivatives of vertical fields in all directions along $\mi{M}$.
In particular, \eqref{eq.aads_vertical_connection_scalar}--\eqref{eq.aads_vertical_connection_contract} imply that these extended derivatives are tensor derivations, in the sense of \cite[Definition 2.11]{on:srg}.
Moreover, \eqref{eq.aads_vertical_connection_compat} ensures that these connections remain compatible with the vertical metric.

\begin{remark}
One difference between the present approach and those of \cite{hol_shao:uc_ads, hol_shao:uc_ads_ns} is that we define our vertical connections with respect to $\gv$.
In contrast, \cite{hol_shao:uc_ads, hol_shao:uc_ads_ns} defined horizontal connections with respect to the metric induced by $g$ (which, in the present setting, corresponds to $\rho^{-2} \gv$).

While our choice to define $\bar{\Dv}$ in terms of $\gv$ is less standard, the main reason for doing so is that we are interested in limits of vertical tensor fields as $\rho \searrow 0$, and $\gv$ itself has such a boundary limit.
\end{remark}

To properly construct the $g$-wave operator for vertical tensor fields, however, we will need a brief detour involving more complex tensorial quantities on $\mi{M}$.
\footnote{Again, the ideas are analogous to those presented in \cite{hol_shao:uc_ads, hol_shao:uc_ads_ns}.}

\begin{definition} \label{def.aads_mixed}
Let $( \mi{M}, g )$ be an FG-aAdS segment.
We then define the \emph{mixed bundle} of rank $( \kappa, \lambda; k, l )$ over $\mi{M}$ to be the tensor product bundle given by
\footnote{More explicitly, the fiber of $T^\kappa_\lambda \ms{V}^k_l \mi{M}$ at each $P \in \mi{M}$ is the tensor product of the fibers of $T^\mu_\lambda \mi{M}$ and $\ms{V}^k_l \mi{M}$ at $P$.}
\begin{equation}
\label{eq.aads_mixed} T^\kappa_\lambda \ms{V}^k_l \mi{M} := T^\kappa_\lambda \mi{M} \otimes \ms{V}^k_l \mi{M} \text{.}
\end{equation}
Moreover, we refer to sections of $T^\kappa_\lambda \ms{V}^k_l \mi{M}$ as \emph{mixed tensor fields} of rank $( \kappa, \lambda; k, l )$.
\end{definition}

Roughly, one can view mixed tensor fields as containing both spacetime and vertical components.

\begin{remark}
Note that any tensor field of rank $( \kappa, \lambda )$ on $\mi{M}$ can be viewed as a mixed tensor field, with rank $( \kappa, \lambda; 0, 0 )$.
Similarly, any vertical tensor field is also a mixed tensor field.
\end{remark}

\begin{definition} \label{def.aads_mixed_connection}
Let $( \mi{M}, g )$ be an FG-aAdS segment.
Then, for all ranks $( \kappa, \lambda; k, l )$, we define the bundle connection $\bar{\nabla}$ on the mixed bundle $T^\kappa_\lambda \ms{V}^k_l \mi{M}$ to be the tensor product connection of the spacetime Levi-Civita connection $\nabla$ on $T^\kappa_\lambda \mi{M}$ and the vertical connection $\bar{\Dv}$ on $\ms{V}^k_l \mi{M}$.

In other words, the $\bar{\nabla}$'s are defined to be the unique family of connections on the mixed bundles such that for any vector field $X$ on $\mi{M}$, tensor field $G$ on $\mi{M}$, and vertical tensor field $\ms{B}$,
\begin{equation}
\label{eq.aads_mixed_connection} \nablam_X ( G \otimes \ms{B} ) = \nabla_X G \otimes \ms{B} + G \otimes \bar{\Dv}_X \ms{B} \text{.}
\end{equation}
\end{definition}

\begin{proposition} \label{thm.aads_mixed_connection}
Let $( \mi{M}, g )$ be an FG-aAdS segment.
Then:
\begin{itemize}
\item For any vector field $X$ on $\mi{M}$ and mixed tensor fields $\mathbf{A}$ and $\mathbf{B}$,
\footnote{$\mathbf{A} \otimes \mathbf{B}$ can be defined componentwise, as usual, by multiplying the components of $\mathbf{A}$ and $\mathbf{B}$.}
\begin{equation}
\label{eq.aads_mixed_connection_leibniz} \nablam_X ( \mathbf{A} \otimes \mathbf{B} ) = \nablam_X \mathbf{A} \otimes \mathbf{B} + \mathbf{A} \otimes \nablam_X \mathbf{B} \text{.}
\end{equation}

\item For any vector field $X$ on $\mi{M}$, we have
\begin{equation}
\label{eq.aads_mixed_connection_compat} \nablam_X g = 0 \text{,} \qquad \nablam_X g^{-1} = 0 \text{,} \qquad \nablam_X \gv = 0 \text{,} \qquad \nablam_X \gv^{-1} = 0 \text{.}
\end{equation}
\end{itemize}
\end{proposition}

\if\comp1

\begin{proof}
See Appendix \ref{sec.aads_mixed_connection}.
\end{proof}

\fi

Roughly, the mixed connections $\nablam$ from Definition \ref{def.aads_mixed_connection} are characterized by the condition that they behave like $\nabla$ on spacelike components and like $\bar{\Dv}$ on vertical components.
Most importantly, the properties in Proposition \ref{thm.aads_mixed_connection} are analogous to the properties of covariant derivatives that enable the standard integration by parts formulas.
Thus, in practice, Proposition \ref{thm.aads_mixed_connection} ensures that the usual integration by parts formulas extend directly to mixed tensor fields.

We can now, in the context of mixed bundles, make sense of higher covariant derivatives:

\begin{definition} \label{def.aads_mixed_wave}
Let $( \mi{M}, g )$ be an FG-aAdS segment, and let $\mathbf{A}$ be a mixed tensor field of rank $( \kappa, \lambda; k, l )$.
We then define the following quantities from $\mathbf{A}$:
\begin{itemize}
\item The \emph{mixed covariant differential} of $\mathbf{A}$ is the mixed tensor field $\bar{\nabla} \mathbf{A}$, of rank $( \kappa, \lambda + 1; k, l )$, that maps each vector field $X$ on $\mi{M}$ (in the extra covariant slot) to $\bar{\nabla}_X \mathbf{A}$.

\item We can then make sense of higher-order covariant differentials of $\mathbf{A}$.
For instance, the \emph{mixed Hessian} $\bar{\nabla}^2 \mathbf{A}$ is defined to be the mixed covariant differential of $\bar{\nabla} \mathbf{A}$.

\item In particular, we can make sense of \emph{wave operator} applied to $\mathbf{A}$---we define $\bar{\Box} \mathbf{A}$ to be the $g$-trace of $\bar{\nabla}^2 \mathbf{A}$, with this trace being applied to the two derivative components.

\item Moreover, the \emph{mixed curvature} applied to $\mathbf{A}$ is defined to be the mixed tensor field $\bar{R} [ \mathbf{A} ]$, of rank $( \kappa, \lambda + 2; k, l )$, that maps vector fields $X, Y$ on $\mi{M}$ (in the extra two slots) to
\begin{equation}
\label{eq.aads_mixed_curv} \bar{R}_{XY} [ \mathbf{A} ] := \bar{\nabla}^2_{XY} \mathbf{A} - \bar{\nabla}^2_{YX} \mathbf{A} \text{.}
\end{equation}
\end{itemize}
\end{definition}

\begin{remark}
Although we systematically define the operators $\bar{\nabla}$, $\bar{\Box}$, $\bar{R}$ in Definition \ref{def.aads_mixed_wave} on all mixed tensor fields, we will, in practice, only apply them to vertical tensor fields.
However, even in this simplified setting, there are some subtleties regarding second covariant differentials of vertical vector fields---namely, the second derivative acts as a spacetime derivative $\nabla$ on the first derivative slot and as a vertical derivative $\bar{\Dv}$ on the vertical tensor field itself.
\end{remark}

\begin{remark}
Mixed tensor fields and wave operators, in their current form, were first defined in \cite{shao:ksp}.
A similar formulation of mixed wave operators was recently, and independently, used in \cite{keir:weak_null}.
\end{remark}

\begin{proposition} \label{thm.aads_vertical_curv}
Let $( \mi{M}, g )$ be an FG-aAdS segment, and let $( U, \varphi )$ be a coordinate system on $\mi{I}$.
Then, for any vertical tensor field $\ms{A}$ of rank $( k, l )$, we have the identities
\begin{align}
\label{eq.aads_vertical_curv} \bar{R}_{ a b } [ \ms{A} ]^{ c_1 \dots c_k }_{ d_1 \dots d_l } &= \sum_{ i = 1 }^k \ms{R}^{ c_i }{}_{ e a b } \ms{A}^{ c_1 \hat{e}_i c_k }_{ d_1 \dots d_l } - \sum_{ j = 1 }^l \ms{R}^e{}_{ d_j a b } \ms{A}^{ c_1 \dots c_k }_{ d_1 \hat{e}_j d_l } \text{,} \\
\notag \bar{R}_{ \rho a } [ \ms{A} ]^{ c_1 \dots c_k }_{ d_1 \dots d_l } &= \frac{1}{2} \sum_{ i = 1 }^k \gv^{ b c_i } ( \Dv_e \mi{L}_\rho \gv_{ a b } - \Dv_b \mi{L}_\rho \gv_{ a e } ) \ms{A}^{ c_1 \hat{e}_i c_k }_{ d_1 \dots d_l } \\
\notag &\qquad - \frac{1}{2} \sum_{ j = 1 }^l \gv^{ b e } ( \Dv_{ d_j } \mi{L}_\rho \gv_{ a b } - \Dv_b \mi{L}_\rho \gv_{ a d_j } ) \ms{A}^{ c_1 \dots c_k }_{ d_1 \hat{e}_j c_k } \text{.}
\end{align}
where we have indexed with respect to $\varphi$ and $\varphi_\rho$-coordinates, and where the symbols $c_1 \hat{e}_i c_k$ and $d_1 \hat{e}_j d_l$ are defined in the same manner as in Definition \ref{def.aads_vertical_connection}.
\end{proposition}

\if\comp1

\begin{proof}
See Appendix \ref{sec.aads_vertical_curv}.
\end{proof}

\fi

\section{The Null Convexity Criterion} \label{sec.psc}

In this section, we give a precise statement of the null convexity criterion.
We then demonstrate that it implies a uniform positivity property that is analogous to the pseudoconvexity condition in \cite{hol_shao:uc_ads_ns}.
This is the crucial property that is required for our main Carleman estimates to hold.

\subsection{The Pseudoconvexity Theorem} \label{sec.psc_thm}

We begin by stating the null convexity criterion.
In contrast to the developments in \cite{hol_shao:uc_ads_ns}, we formulate this condition locally on subsets of $\mi{I}$.

\begin{definition} \label{def.aads_nc}
Let $( \mi{M}, g )$ be a strongly FG-aAdS segment, let $t$ be a global time, and let $\mi{D} \subseteq \mi{I}$ be open.
We say that $( \mi{M}, g, t )$ satisfies the \emph{null convexity criterion} on $\mi{D}$, with constants $0 \leq B < C$, iff the following inequalities hold for any $\gm$-null vector field $\mf{Z}$ on $\mi{D}$:
\begin{equation}
\label{eq.aads_nc} - \gs ( \mf{Z}, \mf{Z} ) \geq C^2 \cdot ( \mf{Z} t )^2 \text{,} \qquad | \Dm^2 t ( \mf{Z}, \mf{Z} ) | \leq 2 B \cdot ( \mf{Z} t )^2 \text{.}
\end{equation}
\end{definition}

The constants $B$ and $C$ in Definition \ref{def.aads_nc} are closely connected to our upcoming main results:
\begin{enumerate}
\item They determine the timespan required for our Carleman estimates to hold.

\item They also determine the time needed for a null geodesic starting from and remaining near the conformal boundary to return to the boundary.
\end{enumerate}
In particular, (1) is the crucial ingredient for establishing unique continuation results, while (2) is critical for constructing counterexamples to unique continuation results.

The following theorem, which is the main result of the present section, connects the null convexity criterion with the pseudoconvexity condition from \cite[Definition 3.2]{hol_shao:uc_ads_ns}:

\begin{theorem} \label{thm.psc_nc}
Let $( \mi{M}, g )$ be a strongly FG-aAdS segment, let $t$ be a global time, and let $\mi{D} \subseteq \mi{I}$ be open.
In addition, assume $( \mi{M}, g, t )$ satisfies the null convexity criterion on $\mi{D}$, with constants $0 \leq B < C$.
Then, given any $b, c \in \R$ such that $B < b < c < C$, the following hold:
\begin{itemize}
\item There exists $\zeta_0 \in C^\infty ( \mi{I} )$ such that for any vector field $\mf{X}$ on $\mi{D}$,
\begin{equation}
\label{eq.psc_nc_g} ( - \gs - c^2 \cdot dt^2 - \zeta_0 \cdot \gm ) ( \mf{X}, \mf{X} ) \gtrsim_{ C^2 - c^2, t } \mf{h} ( \mf{X}, \mf{X} ) \text{.}
\end{equation}

\item There exists $\zeta_\pm \in C^\infty ( \mi{I} )$ such that for any vector field $\mf{X}$ on $\mi{D}$,
\begin{equation}
\label{eq.psc_nc_t} ( \mp \Dm^2 t + 2 b \cdot dt^2 - \zeta_\pm \cdot \gm ) ( \mf{X}, \mf{X} ) \gtrsim_{ b - B, t } \mf{h} ( \mf{X}, \mf{X} ) \text{.}
\end{equation}
\end{itemize}
\end{theorem}

\begin{remark}
In fact, the conclusions of Theorem \ref{thm.psc_nc} are a generalization of \cite[Definition 3.2]{hol_shao:uc_ads_ns}:
\begin{itemize}
\item The constants $B$ and $C$ in Theorem \ref{thm.psc_nc} correspond to $\mu$ and $\xi$ in \cite{hol_shao:uc_ads_ns} via the formulas
\[
B := \frac{ \xi }{4} \text{,} \qquad C^2 := \mu^2 + \frac{ \xi^2 }{ 16 } \text{.}
\]

\item In contrast to \cite{hol_shao:uc_ads_ns}, we do not assume that $t$ is a unit geodesic parameter (i.e., $\Dm^\sharp t$ is both unit and geodesic).
Under this additional assumption, $2 \, \Dm^2 t$ coincides with the Lie derivative of $\gm$ in the $t$-direction.
Thus, \eqref{eq.psc_nc_t} serves as the analogue of \cite[Equation (3.4)]{hol_shao:uc_ads_ns}.

\item Unlike in \cite{hol_shao:uc_ads_ns}, we need only measure the size of $\Dm^2 t$ along the $\gm$-null directions, rather than in all directions.
Thus, \eqref{eq.psc_nc_t} is a strictly weaker assumption than \cite[Equation (3.4)]{hol_shao:uc_ads_ns}.
\end{itemize}
\end{remark}

The proof of Theorem \ref{thm.psc_nc} is given in Section \ref{sec.psc_proof}.
The key ingredient of the proof is the following pointwise analogue, stated for bilinear forms on a vector space:

\begin{proposition} \label{thm.psc_pointwise}
Let $n > 1$, and let $\mf{V}$ be an $n$-dimensional real vector space.
In addition:
\begin{itemize}
\item Let $\mf{m}$ be a nondegenerate symmetric bilinear form on $\mf{V}$.

\item Let $\mf{p}$ be another symmetric bilinear form on $\mf{V}$.
\end{itemize}
Then, the following statements are equivalent:
\begin{enumerate}
\item $\mf{p} ( v, v ) > 0$ for any $v \in \mf{V} \setminus \{ 0 \}$ such that $\mf{m} ( v, v ) = 0$.

\item There exists $\lambda \in \R$ such that
\begin{equation}
\label{eq.psc_pointwise} \mf{p} ( v, v ) - \lambda \cdot \mf{m} ( v, v ) > 0 \text{,} \qquad v \in \mf{V} \setminus \{ 0 \} \text{.}
\end{equation}
\end{enumerate}
\end{proposition}

The proof of Proposition \ref{thm.psc_pointwise} is given in Section \ref{sec.psc_pointwise}.
While this is, in large part, an elaboration of a similar argument found in \cite[Lemma 4.3]{tat:notes_uc}, we give the details here for completeness.

Finally, the following reformulation of Proposition \ref{thm.psc_pointwise} will be applied toward Theorem \ref{thm.psc_nc}:

\begin{corollary} \label{thm.psc_pointwise_ex}
Let $n$, $\mf{V}$, $\mf{m}$, $\mf{p}$ be as in the statement of Proposition \ref{thm.psc_pointwise}.
In addition, let $c > 0$, and let $\mf{n}$ be a (positive) inner product on $\mf{V}$.
Then, the following statements are equivalent:
\begin{enumerate}
\item $\mf{p} ( v, v ) > c \cdot \mf{n} ( v, v )$ for any $v \in \mf{V} \setminus \{ 0 \}$ such that $\mf{m} ( v, v ) = 0$.

\item There exists $\lambda \in \R$ such that
\begin{equation}
\label{eq.psc_pointwise_ex} \mf{p} ( v, v ) - \lambda \cdot \mf{m} ( v, v ) > c \cdot \mf{n} ( v, v ) \text{,} \qquad v \in \mf{V} \setminus \{ 0 \} \text{.}
\end{equation}
\end{enumerate}
\end{corollary}

\begin{proof}
This follows by applying Proposition \ref{thm.psc_pointwise} with $\mf{p}$ replaced by $\mf{p} - c \mf{n}$.
\end{proof}

\subsection{Proof of Proposition \ref{thm.psc_pointwise}} \label{sec.psc_pointwise}

First, if $\mf{m}$ is positive-definite, then the result is trivial.
Therefore, we assume from now on that $\mf{m}$ is sign-indefinite.
Moreover, notice that (2) trivially implies (1) in general.
Thus, we need only show (1) implies \eqref{eq.psc_pointwise} for some $\lambda$, and hence (2).

Let $\mb{P} ( \mf{V} )$ denote the projective space over $\mf{V}$, and let $\Pi: \mf{V} \rightarrow \mb{P} ( \mf{V} )$ be the corresponding natural projection.
For convenience, we also treat the bilinear forms mentioned above as quadratic forms:
\begin{equation}
\label{eql.psc_pointwise_0} \mf{m} (v) := \mf{m} ( v, v ) \text{,} \qquad \mf{p} ( v ) := \mf{p} ( v, v ) \text{.}
\end{equation}
In addition, we define the following subsets of $\mf{V}$,
\begin{align}
\label{eql.psc_pointwise_G} G_+ &:= \{ v \in \mf{V} \mid \mf{m} ( v ) > 0 \} \text{,} \\
\notag G_- &:= \{ v \in \mf{V} \mid \mf{m} ( v ) < 0 \} \text{,} \\
\notag G_0 &:= \{ v \in \mf{V} \mid \mf{m} ( v ) = 0 \} \text{,}
\end{align}
and we define, for any $\lambda \in \R$, the set
\begin{equation}
\label{eql.psc_pointwise_Z} Z_\lambda := \{ v \in \mf{V} \setminus \{ 0 \} \mid \mf{p} ( v ) - \lambda \cdot \mf{m} ( v ) = 0 \} \text{.}
\end{equation}

\begin{lemma} \label{thm.psc_pointwise_Z}
For any $\lambda \in \R$, exactly one of the following is true:
\begin{itemize}
\item $Z_\lambda$ is entirely contained in $G_+$.

\item $Z_\lambda$ is entirely contained in $G_-$.

\item $Z_\lambda$ is empty.
\end{itemize}
\end{lemma}

\begin{proof}
First, note that the assumed condition (1) immediately implies that $Z_\lambda \cap G_0 = \emptyset$.
Moreover since $Z_\lambda$ and $G_\pm$ are scaling-invariant, it suffices to work in the projective setting and show that exactly one of the following is true: (a) $\Pi Z_\lambda \subseteq \Pi G_+$; (b) $\Pi Z_\lambda \subseteq \Pi G_-$; or (c) $\Pi Z_\lambda = \emptyset$.

Next, recall that $\mb{P} ( \mf{V} ) \setminus \Pi G_0$ is disconnected and has two connected components, namely, $\Pi G_\pm$.
Furthermore, $Z_\lambda$, being the zero set of a quadratic form, must be connected, hence $\Pi Z_\lambda$ is also connected.
Since $\Pi Z_\lambda$ and $\Pi G_0$ are disjoint, the desired result follows.
\end{proof}

\begin{lemma} \label{thm.psc_pointwise_empty}
There exists $\lambda \in \R$ such that $Z_\lambda = \emptyset$.
\end{lemma}

\begin{proof}
First, observe that $\mf{q} := \mf{p} / \mf{m}$ is a well-defined and continuous function on $G_+ \cup G_-$.
Since $\mf{q}$ is dilation-invariant, it also induces a continuous function at the projective level,
\[
\mf{q}: \Pi G_+ \cup \Pi G_- \rightarrow \R \text{.}
\]
Furthermore, since $\mf{p}$ is positive near $G_0$, it follows that:
\begin{itemize}
\item $\mf{q} ( \bar{v} ) \rightarrow + \infty$ as $\bar{v} \rightarrow \Pi G_0$ along $\Pi G_+$.

\item $\mf{q} ( \bar{v} ) \rightarrow - \infty$ as $\bar{v} \rightarrow \Pi G_0$ along $\Pi G_-$.
\end{itemize}

Consider now the function
\begin{equation}
\label{eql.psc_pointwise_empty_1p} \mf{q}_+: \Pi G_+ \cup \Pi G_0 \rightarrow ( -\infty, +\infty ] \text{,} \qquad \mf{q}_+ ( \bar{v} ) = \begin{cases} + \infty & \bar{v} \in \Pi G_0 \text{,} \\ \mf{q} ( \bar{v} ) & \bar{v} \in \Pi G_+ \text{.} \end{cases}
\end{equation}
By the above, we know that $\mf{q}_+$ is continuous.
Moreover, since $\Pi G_+ \cup \Pi G_0$ is compact and connected, it follows that the image of $\mf{q}_+$ must be of the form $[ c_+, + \infty ]$ for some $c_+ \in \R$.

By an analogous argument, the function
\begin{equation}
\label{eql.psc_pointwise_empty_1n} \mf{q}_-: \Pi G_- \cup \Pi G_0 \rightarrow [ -\infty, +\infty ) \text{,} \qquad \mf{q}_- ( \bar{v} ) = \begin{cases} - \infty & \bar{v} \in \Pi G_0 \text{,} \\ \mf{q} ( \bar{v} ) & \bar{v} \in \Pi G_- \text{,} \end{cases}
\end{equation}
is also continuous, and its image must be an interval of the form $[ - \infty, c_- ]$ for some $c_- \in \R$.
In particular, it follows that the image of $\mf{q}$ must be $( - \infty, c_- ] \cup [ c_+, + \infty )$.

Suppose, for a contradiction, that $c_- \geq c_+$.
Fix now some $\alpha \in [ c_+, c_- ]$.
Then, by \eqref{eql.psc_pointwise_empty_1p}, \eqref{eql.psc_pointwise_empty_1n}, and the above discussions, we obtain that $\mf{q} ( \bar{v}_\pm ) = \alpha$ for some pair of elements $\bar{v}_\pm \in \Pi G_\pm$.
Lifting back up to $\mf{V}$, we deduce from \eqref{eql.psc_pointwise_Z} that there exist $v_\pm \in G_\pm$ satisfying
\[
\mf{p} ( v_\pm ) = \alpha \cdot \mf{m} ( v_\pm ) \text{,} \qquad v_\pm \in Z_\alpha \text{.}
\]
In particular, $Z_\alpha$ contains an element from both $G_+$ and $G_-$, which contradicts Lemma \ref{thm.psc_pointwise_Z}.

As a result, we conclude that $c_- < c_+$, and hence there exists some $\lambda \in \R$ that does not lie in the image of $\mf{q}$.
The definition of $\mf{q}$ then yields that $Z_\lambda$ must be empty.
\end{proof}

We can now complete the proof of Proposition \ref{thm.psc_pointwise}.
By Lemma \ref{thm.psc_pointwise_empty}, there is some $\lambda \in \R$ such that $Z_\lambda = \emptyset$.
As a result, we have, for any $v \in \mf{V} \setminus \{ 0 \}$, that
\[
\mf{p} ( v ) - \lambda \cdot \mf{m} ( v ) \neq 0 \text{.}
\]
Moreover, since the range of $\mf{p} - \lambda \cdot \mf{m}$ (viewed as a quadratic form) is connected by continuity, and since $\mf{p} - \lambda \cdot \mf{m}$ is positive on $G_0$, it follows that $\mf{p} - \lambda \cdot \mf{m}$ must be everywhere positive.
Therefore, \eqref{eq.psc_pointwise} holds for this particular $\lambda$, and the proof of Proposition \ref{thm.psc_pointwise} is complete.

\subsection{Proof of Theorem \ref{thm.psc_nc}} \label{sec.psc_proof}

First, for any $p \in \bar{\mi{D}}$, we apply Corollary \ref{thm.psc_pointwise_ex}, with
\begin{equation}
\label{eql.psc_nc_1} \mf{V} := T_p \mi{I} \text{,} \qquad \mf{m} := \gm |_p \text{,} \qquad \mf{n} := \mf{h} |_p \text{,} \qquad \mf{p} := ( - \gs - c^2 \cdot dt^2 ) |_p \text{.}
\end{equation}
For any $\mf{g}$-null $v \in T_p \mi{I}$, so that $\mf{m} ( v, v ) = 0$, we use \eqref{eq.aads_time}, \eqref{eq.aads_riemann}, \eqref{eq.aads_nc}, and \eqref{eql.psc_nc_1} to deduce
\begin{align*}
\mf{p} ( v, v ) &= - \gs |_p ( v, v ) - C^2 \cdot ( v t )^2 + ( C^2 - c^2 ) \cdot dt^2 |_p ( v, v ) \\
&\geq ( C^2 - c^2 ) \cdot dt^2 |_p ( v, v ) \\
&\simeq \mf{h} ( v, v ) \text{.}
\end{align*}
Thus, by Corollary \ref{thm.psc_pointwise_ex}, we can find some $\lambda_p \in \R$ such that for any $u \in T_p \mi{I}$,
\begin{equation}
\label{eql.psc_nc_2} ( - \gs - c^2 \cdot dt^2 - \lambda_p \cdot \gm ) ( u, u ) \gtrsim_{ C^2 - c^2, t } \mf{h} ( u, u ) \text{.}
\end{equation}
The bound \eqref{eq.psc_nc_g} now follows by combining \eqref{eql.psc_nc_2}, for all $p \in \bar{\mi{D}}$, with a partition of unity argument.

For \eqref{eq.psc_nc_t}, we rewrite the second part of \eqref{eq.aads_nc} as
\begin{equation}
\label{eql.psc_nc_3} \mp \mf{D}^2 t ( \mf{Z}, \mf{Z} ) + 2 B \cdot ( \mf{Z} t )^2 \geq 0 \text{,}
\end{equation}
where $\mf{Z}$ is any $\gm$-null vector field.
Then, \eqref{eq.psc_nc_t} follows by applying Corollary \ref{thm.psc_pointwise_ex} in the same manner as in the preceding paragraph, with $\mf{V}$, $\mf{m}$, $\mf{n}$ as in \eqref{eql.psc_nc_1}, and with
\begin{equation}
\label{eql.psc_nc_4} \mf{p} := ( \mp \Dm^2 t + 2 b \cdot dt^2 ) |_p \text{.}
\end{equation}
In particular, $\mf{p}$ is positive-definite on null vectors due to \eqref{eql.psc_nc_3}.

\section{Null Geodesics} \label{sec.geodesic}

The objective of this section is to connect the null convexity criterion of Definition \ref{def.aads_nc}, defined on the conformal boundary, to trajectories of null geodesics in the corresponding FG-aAdS segment.
In particular, we consider null geodesics that begin at and remain near the conformal boundary, and we control the amount of time needed for the geodesics to return to the boundary.

\subsection{Description of Null Geodesics} \label{sec.geodesic_desc}

The first step is make precise sense of null geodesics starting from the conformal boundary.
Here, we expand our definition slightly, as it will be more convenient to work instead with certain reparametrizations of geodesics.

\begin{definition} \label{def.geodesic_curve}
Let $( \mi{M}, g )$ be a strongly FG-aAdS segment, and let $t$ be a global time.
Then, given any curve $L: ( \ell_-, \ell_+ ) \rightarrow \mi{M}$, we define the following shorthands:
\begin{itemize}
\item We abbreviate the $\rho$- and $t$-components of $L$ as
\begin{equation}
\label{eq.geodesic_curve_comp} \rho_L := \rho \circ L \text{,} \qquad t_L := t \circ L \text{.}
\end{equation}

\item We let $\lambda_L$ denote the natural projection onto $\mi{I}$ of $L$.
\end{itemize}
\end{definition}

\begin{definition} \label{def.geodesic}
Let $( \mi{M}, g )$ be a strongly FG-aAdS segment, and let $t$ be a global time.
A curve
\begin{equation}
\label{eq.geodesic_curve} \Lambda: ( \ell_-, \ell_+ ) \rightarrow \mi{M} \text{,}
\end{equation}
is called a \emph{$t$-parametrized null curve of $( \mi{M}, g )$} iff the following hold:
\begin{itemize}
\item $\Lambda$ is a reparametrization of an inextendible null geodesic of $( \mi{M}^\circ, g )$.
\footnote{Here, $\mi{M}^\circ := ( 0, \rho_0 ) \times \mi{I}$ denotes the interior of $\mi{M}$.}

\item $\Lambda$ is parametrized by $t$:
\begin{equation}
\label{eq.geodesic_t} t_\Lambda ( \tau ) = \tau \text{,} \qquad \tau \in ( \ell_-, \ell_+ ) \text{.}
\end{equation}
\end{itemize}

Moreover, for such a $t$-parametrized null curve $\Lambda$, we say that $\Lambda$ lies \emph{over $\mi{D} \subseteq \mi{I}$} iff:
\begin{itemize}
\item The image of $\Lambda$ is contained in $( 0, \rho_0 ) \times \mi{D}$.

\item The following initial conditions hold:
\footnote{Here, we abuse notation slightly and write $\rho_\Lambda' ( \ell_- )$ rather than $\lim_{ \tau \searrow \ell_- } \rho_\Lambda' ( \tau )$ for brevity.}
\begin{equation}
\label{eq.geodesic_init} \lim_{ \tau \searrow \ell_- } \rho_\Lambda ( \tau ) = 0 \text{,} \qquad \rho_\Lambda' ( \ell_- ) := \lim_{ \tau \searrow \ell_- } \rho_\Lambda' ( \tau ) > 0 \text{.}
\end{equation}
\end{itemize}
\end{definition}

Next, we derive the equation governing the $\rho$-values of $t$-parametrized null curves.

\begin{proposition} \label{thm.geodesic_eq}
Let $( \mi{M}, g )$ be a strongly FG-aAdS segment, let $t$ be a global time, and let $\Lambda$ be a $t$-parametrized null curve.
Then, $\rho_\Lambda$ satisfies the following equation:
\begin{equation}
\label{eq.geodesic_eq} \rho_\Lambda'' + \left[ \Dv^2 t ( \lambda_\Lambda', \lambda_\Lambda' ) - \mi{L}_\rho \gv ( \Dv^\sharp t, \lambda_\Lambda' ) \cdot \rho_\Lambda' \right] \rho_\Lambda' - \frac{1}{2} \mi{L}_\rho \gv ( \lambda_\Lambda', \lambda_\Lambda' ) = 0 \text{.}
\end{equation}
\end{proposition}

\begin{proof}
We consider the conformal metric $\rho^2 g$, for which $\rho = 0$ can be treated as a finite boundary.
Fix any coordinate system $( U, \varphi )$ on $\mi{I}$; by \eqref{eq.aads_time}, we can assume $t$ is one of the coordinates in $\varphi$.
Then, the Christoffel symbols $\smash{ \Gamma^\gamma_{ \alpha \beta } }$ for $\rho^2 g$, with respect to $\varphi_\rho$-coordinates, satisfy
\begin{align}
\label{eql.geodesic_eq_1} \Gamma^\rho_{ \rho \rho } = 0 \text{,} \qquad \Gamma^\rho_{ \rho a } &= 0 \text{,} \qquad \Gamma^a_{ \rho \rho } = 0 \text{,} \\
\notag \Gamma^\rho_{ a b } = - \frac{1}{2} \mi{L}_\rho \gv_{ a b } \text{,} \qquad \Gamma^a_{ \rho b } &= \frac{1}{2} \gv^{ a c } \mi{L}_\rho \gv_{ c b } \text{,} \qquad \Gamma^t_{ a b } = - \Dv_{ a b } t \text{.}
\end{align}

Since $\Lambda$ is null with respect to $\rho^2 g$, then $\Lambda$ has a reparametrization $L$ that is a null geodesic with respect to $\rho^2 g$.
As a result, the geodesic equations, in the above $\varphi_\rho$-coordinates, yield
\begin{align}
\label{eql.geodesic_eq_2} 0 &= ( \rho \circ L )'' + \Gamma^\rho_{ a b } ( x^a \circ L )' ( x^b \circ L )' \text{,} \\
\notag 0 &= ( t \circ L )'' + 2 \Gamma^t_{ \rho a } ( \rho \circ L )' ( x^a \circ L )' + \Gamma^t_{ a b } ( x^a \circ L )' ( x^b \circ L )' \text{.}
\end{align}
Then, combining \eqref{eql.geodesic_eq_1} and \eqref{eql.geodesic_eq_2}, while recalling Definition \ref{def.geodesic_curve}, we obtain
\begin{equation}
\label{eql.geodesic_eq_10} \rho_L'' - \frac{1}{2} \mi{L}_\rho \gv ( \lambda_L', \lambda_L' ) = 0 \text{,} \qquad t_L'' + \mi{L}_\rho \gv ( \Dv^\sharp t, \lambda_L' ) \cdot \rho_L' - \Dv^2 t ( \lambda_L', \lambda_L' ) = 0 \text{.}
\end{equation}

Next, fix $\tau \in ( \ell_-, \ell_+ )$, and let $s$ be the value of the affine parameter for $L$ such that $L (s) = \Lambda ( \tau )$.
Then, applying the chain rule, we obtain the following:
\begin{equation}
\label{eql.geodesic_eq_11} \rho_L' ( s ) = t_L' ( s ) \cdot \rho_\Lambda' ( \tau ) \text{,} \qquad \lambda_L' ( s ) = t_L' ( s ) \cdot \lambda_\Lambda' ( \tau ) \text{.}
\end{equation}
Moreover, differentiating the first part of \eqref{eql.geodesic_eq_11} yields
\begin{align}
\label{eql.geodesic_eq_12} \rho_L'' ( s ) &= [ t_L' ( s ) ]^2 \rho_\Lambda'' ( \tau ) + t_L'' ( s ) \rho_\Lambda' ( \tau ) \\
\notag &= [ t_L' ( s ) ]^2 \rho_\Lambda'' ( \tau ) + \left[ \Dv^2 t ( \lambda_L' (s), \lambda_L' (s) ) - \mi{L}_\rho \gv ( \Dv^\sharp t, \lambda_L' (s) ) \cdot \rho_L' (s) \right] \rho_\Lambda' ( \tau ) \\
\notag &= [ t_L' ( s ) ]^2 \left\{ \rho_\Lambda'' ( \tau ) + \left[ \Dv^2 t ( \lambda_\Lambda' ( \tau ), \lambda_\Lambda' ( \tau ) ) - \mi{L}_\rho \gv ( \Dv^\sharp t, \lambda_\Lambda' ( \tau ) ) \cdot \rho_\Lambda' ( \tau ) \right] \rho_\Lambda' ( \tau ) \right\} \text{.}
\end{align}
where we also used the second part of \eqref{eql.geodesic_eq_10} and \eqref{eql.geodesic_eq_11}.
Similarly, by the first part of \eqref{eql.geodesic_eq_10} and \eqref{eql.geodesic_eq_11},
\begin{equation}
\label{eql.geodesic_eq_13} \rho_L'' ( s ) = \frac{1}{2} \mi{L}_\rho \gv ( \lambda_L' (s), \lambda_L' (s) ) = [ t_L' (s) ]^2 \cdot \frac{1}{2} \mi{L}_\rho \gv ( \lambda_\Lambda' ( \tau ), \lambda_\Lambda' ( \tau ) ) \text{.}
\end{equation}
Finally, combining \eqref{eql.geodesic_eq_12} with \eqref{eql.geodesic_eq_13} results in \eqref{eq.geodesic_eq}.
\footnote{Note $t_L' > 0$ everywhere, since if $t_\Lambda' (s)$ vanishes, then $L' ( s )$ is tangent to a level set of $t$ and cannot be null.}
\end{proof}

\subsection{The Geodesic Return Theorem}

The next task is to bound the timespans of $t$-parametrized null curves near the conformal boundary using the null convexity criterion of Definition \ref{def.aads_nc}.

First, the formulas for $\mc{T}_\pm ( B, C )$ in the following definition represent the upper and lower bounds on the above-mentioned timespans of $t$-parametrized null curves.

\begin{definition} \label{def.geodesic_timebound}
Given constants $0 \leq B < C$, we define
\begin{equation}
\label{eq.geodesic_nc_T} \mc{T}_\pm ( B, C ) := \frac{ 2 \cdot \mc{W}_\pm ( B, C ) }{ \sqrt{ C^2 - B^2 } } \text{,} \qquad \mc{W}_\pm ( B, C ) := \tan^{-1} \left( \mp \frac{ \sqrt{ C^2 - B^2 } }{ B } \right) \text{,}
\end{equation}
where the value of $\tan^{-1}$ is chosen such that
\begin{equation}
\label{eq.geodesic_nc_W} \mc{W}_+ ( B, C ) \in \left[ \frac{ \pi }{2}, \pi \right) \text{,} \qquad \mc{W}_- ( B, C ) \in \left[ 0, \frac{ \pi }{2} \right) \text{.}
\end{equation}
\end{definition}

We now state our main theorem concerning the trajectories of null curves near $\mi{I}$:

\begin{theorem} \label{thm.geodesic}
Let $( \mi{M}, g )$ be a strongly FG-aAdS segment, and let $t$ be a global time.
Moreover:
\begin{itemize}
\item Let $\mi{D} \subseteq \mi{I}$ be open, and suppose $\bar{\mi{D}}$ has compact cross-sections.

\item Assume the null convexity criterion holds on $\mi{D}$, with constants $0 \leq B < C_+$.

\item Assume there exists $C_- > C_+$ such that for any $\mathfrak{g}$-null vector field $\mf{Z}$ on $\mi{D}$,
\begin{equation}
\label{eq.geodesic_nc_reverse} - \gs ( \mf{Z}, \mf{Z} ) \leq C_-^2 \cdot ( \mf{Z} t )^2 \text{.}
\end{equation}
\end{itemize}
Then, given any $\eta > 0$, there exists $\varepsilon > 0$---with $\varepsilon$ depending only on $\mi{D}$, $\gv$, and $t$---such that
\begin{itemize}
\item if $\Lambda: ( \ell_-, \ell_+ ) \rightarrow \mi{M}$ is a $t$-parametrized null curve over $\mi{D}$, and

\item if the following conditions hold,
\begin{equation}
\label{eq.geodesic_ncT} t_- < \ell_- \text{,} \qquad \ell_- + \mc{T}_+ ( B, C_+ ) + \eta < t_+ \text{,} \qquad 0 < \rho'_\Lambda ( \ell_- ) < \varepsilon \text{,}
\end{equation}
\end{itemize}
then the following statements also hold:
\begin{itemize}
\item $\Lambda$ returns to $\mi{I}$:
\begin{equation}
\label{eq.geodesic_return} \lim_{ \tau \nearrow \ell_+ } \rho_\Lambda ( \tau ) = 0 \text{.}
\end{equation}

\item The time of return is bounded from above and below:
\begin{equation}
\label{eq.geodesic_timespan} \mc{T}_- ( B, C_- ) - \eta < \ell_+ - \ell_- < \mc{T}_+ ( B, C_+ ) + \eta \text{,}
\end{equation}

\item $\Lambda$ remains close to $\mi{I}$:
\footnote{See \eqref{eq.geodesic_init} for the definition of $\rho_\Lambda' ( \ell_- )$.}
\begin{equation}
\label{eq.geodesic_bound} 0 < \rho_\Lambda ( \tau ) \lesssim_{ B, C_+ } \rho'_\Lambda ( \ell_- ) \text{,} \qquad \tau \in ( \ell_-, \ell_+ ) \text{.}
\end{equation}
\end{itemize}
\end{theorem}

\begin{remark}
In particular, if $\mi{I}$ has compact cross-sections (such as for any Kerr-AdS spacetime), then we can take $\mi{D} := \mi{I}$ and obtain a global version of Theorem \ref{thm.geodesic}.
\end{remark}

\begin{remark}
With a straightforward adaptation of the proof of Theorem \ref{thm.geodesic}, one can also show that if the assumption \eqref{eq.geodesic_nc_reverse} is omitted, then the conclusions of Theorem \ref{thm.geodesic} still hold, except that one only obtains the upper bound in \eqref{eq.geodesic_timespan} for $\ell_+ - \ell_-$.
\footnote{In particular, one would compare $\rho_\Lambda$ only to $f_+$ in both Lemmas \ref{thm.geodesic_A1} and \ref{thm.geodesic_A2} below.}
\end{remark}

\subsection{Proof of Theorem \ref{thm.geodesic}} \label{sec.geodesic_proof}

Throughout, we assume the hypotheses of Theorem \ref{thm.geodesic}, in particular the curve $\Lambda: ( \ell_-, \ell_+ ) \rightarrow \mi{M}$.
Also, by replacing $t$ with $t - \ell_-$, we can assume that $\ell_- = 0$.

The first step is to better describe the behavior of $\rho_\Lambda$ near the conformal boundary.

\begin{lemma} \label{thm.geodesic_asymp}
$\rho_\Lambda$ satisfies the identity
\begin{equation}
\label{eq.geodesic_asymp} 0 = \rho_\Lambda'' + [ \Dm^2 t ( \lambda_\Lambda', \lambda_\Lambda' ) + \ms{r}_1 ( \lambda_\Lambda', \lambda_\Lambda' ) + \ms{r}_2 ( \lambda_\Lambda' ) ] \cdot \rho_\Lambda' - [ \gs ( \lambda_\Lambda', \lambda_\Lambda' ) + \ms{r}_0 ( \lambda_\Lambda', \lambda_\Lambda' ) ] \cdot \rho_\Lambda \text{,}
\end{equation}
where the remainders $\ms{r}_0$, $\ms{r}_1$, $\ms{r}_2$ are vertical tensor fields satisfying
\begin{equation}
\label{eq.geodesic_remainder} \ms{r}_0 \rightarrow^0 0 \text{,} \qquad \ms{r}_1 \rightarrow^0 0 \text{,} \qquad \ms{r}_2 \rightarrow^0 0 \text{.}
\end{equation}
\end{lemma}

\begin{proof}
Taylor's theorem, combined with the limits \eqref{eq.aads_strong_limits} and \eqref{eq.aads_geom_limit_deriv}, yields
\[
\mi{L}_\rho \gv = 2 ( \gs + \ms{r}_0 ) \rho \text{,} \qquad \Dv^2 t = \Dm^2 t + \ms{r}_1 \text{,}
\]
for some $\ms{r}_0$, $\ms{r}_1$ satisfying \eqref{eq.geodesic_remainder}.
The equation \eqref{eq.geodesic_asymp} now follows from \eqref{eq.geodesic_eq} and the above.
\end{proof}

The key idea behind proving Theorem \ref{thm.geodesic} is to apply the classical Sturm comparison to
\begin{equation}
\label{eql.geodesic_leading} 0 = \rho_\Lambda'' + \Dm^2 t ( \lambda_\Lambda', \lambda_\Lambda' ) \cdot \rho_\Lambda' - \gs ( \lambda_\Lambda', \lambda_\Lambda' ) \cdot \rho_\Lambda \text{,}
\end{equation}
which we view as a second-order ODE for $\rho_\Lambda$, and which represents the leading-order terms of \eqref{eq.geodesic_asymp}.
However, the comparison theorem cannot be applied directly, since \eqref{eq.geodesic_asymp} also contains nonlinear terms, hence one must ensure that the solution remains a perturbation of \eqref{eql.geodesic_leading}.
As a result, we take a more direct approach by combining the proof of the comparison theorem in our setting along with a bootstrap argument to handle the nonlinear terms.

For convenience, we abbreviate the coefficients of \eqref{eq.geodesic_asymp} as
\begin{equation}
\label{eql.geodesic_UV} \mc{U} := \Dm^2 t ( \lambda_\Lambda', \lambda_\Lambda' ) + \ms{r}_1 ( \lambda_\Lambda', \lambda_\Lambda' ) + \ms{r}_2 ( \lambda_\Lambda' ) \text{,} \qquad \mc{V} := - \gs ( \lambda_\Lambda', \lambda_\Lambda' ) - \ms{r}_0 ( \lambda_\Lambda', \lambda_\Lambda' ) \text{,}
\end{equation}
both of which we view as functions of $t$ along $\Lambda$.
In addition, we define the values
\begin{align}
\label{eql.geodesic_T} T_2 &:= \min ( \mc{T}_+ ( B, C_+ ) + \eta, \ell_+ ) \text{,} \\
\notag T_1 &:= \inf \left[ \{ \tau \in ( 0, T_2 ) \mid \rho_\Lambda' ( \tau ) = 0 \} \cup \{ T_2 \} \right] \text{.}
\end{align}
Roughly speaking, $T_1$ represents the first vanishing time for $\rho'_\Lambda$, while $T_2$ represents the terminal time for $\Lambda$.
More specifically, the possibilities for $\Lambda$ at $t = T_2$ are as follows:

\begin{lemma} \label{thm.geodesic_terminal}
One of the three possibilities must hold:
\begin{enumerate}
\item $\ell_+ = T_2$, and $\rho_\Lambda ( \tau ) \rightarrow 0$ as $\tau \nearrow \ell_+$.

\item $\ell_+ = T_2$, and $\rho_\Lambda ( \tau ) \rightarrow \rho_0$ as $\tau \nearrow \ell_+$.

\item $\ell_+ > T_2$, and $\rho_\Lambda ( T_2 ) > 0$.
\end{enumerate}
\end{lemma}

\begin{proof}
This is an immediate consequence of Definition \ref{def.geodesic}.
\end{proof}

Next, let $B < b < c_+ < C_+ < C_- < c_-$, with $b - B$ and $| C_\pm - c_\pm |$ small enough so that
\begin{align}
\label{eql.geodesic_BC} 2 \, T_+ &:= \mc{T}_+ ( b, c_+ ) < \mc{T}_+ ( B, C_+ ) + \eta \text{,} \qquad 2 \, T_- := \mc{T}_- ( b, c_- ) > \mc{T}_- ( B, C_- ) - \eta \text{.}
\end{align}
(Note in particular that $T_- < T_+$.)
Moreover, we define the set
\begin{equation}
\label{eql.geodesic_A} \mc{A} := \{ \tau \in ( 0, T_2 ) \mid | \mc{U} ( \sigma ) | \leq 2 b \text{ and } c_-^2 \geq \mc{V} ( \sigma ) \geq c_+^2 \text{ for all } \sigma \in [ 0, \tau ] \} \text{,}
\end{equation}
representing the times for which the comparison principle is applicable, and we split $\mc{A}$ as
\begin{equation}
\label{eql.geodesic_AA} \mc{A}_1 := \mc{A} \cap ( 0, T_1 ] \text{,} \qquad \mc{A}_2 := \mc{A} \cap [ T_1, T_2 ) \text{.}
\end{equation}
We now obtain, via the comparison principle, a priori estimates for $\rho_\Lambda$ on $\mc{A}_1$ and $\mc{A}_2$:

\begin{lemma} \label{thm.geodesic_A1}
$\mc{A}_1 \subseteq ( 0, T_+ ]$, and the following estimate holds:
\begin{equation}
\label{eq.geodesic_A1} \rho_\Lambda ( \tau ) \lesssim_{ B, C_+ } \rho_\Lambda' ( 0 ) \text{,} \qquad \tau \in \mc{A}_1 \text{.}
\end{equation}
Furthermore, if $T_1 \in \mc{A}_1$, then
\begin{equation}
\label{eq.geodesic_A1_end} T_- \leq T_1 \leq T_+ \text{,} \qquad \rho_\Lambda' ( T_1 ) = 0 \text{.}
\end{equation}
\end{lemma}

\begin{proof}
Since $\rho_\Lambda' \geq 0$ on $[ 0, T_1 ]$, it follows from \eqref{eq.geodesic_asymp}, \eqref{eql.geodesic_UV}, \eqref{eql.geodesic_A}, and \eqref{eql.geodesic_AA} that on $\mc{A}_1$,
\begin{align}
\label{eql.geodesic_A1_1} \rho_\Lambda'' - 2 b \rho_\Lambda' + c_+^2 \rho_\Lambda \leq 0 \text{,} &\qquad ( e^{ - 2 b t } \rho_\Lambda' )' + c_+^2 e^{ - 2 b t } \rho_\Lambda \leq 0 \text{,} \\
\notag \rho_\Lambda'' + 2 b \rho_\Lambda' + c_-^2 \rho_\Lambda \geq 0 \text{,} &\qquad ( e^{ + 2 b t } \rho_\Lambda' )' + c_-^2 e^{ + 2 b t } \rho_\Lambda \geq 0 \text{.}
\end{align}
In addition, let $f_\pm$ denote the solutions of the initial value problems
\begin{equation}
\label{eql.geodesic_A1_2} ( e^{ \mp 2 b t } f_\pm' )' + c_\pm^2 e^{ \mp 2 b t } f_\pm = 0 \text{,} \qquad ( f_\pm (0), f_\pm' (0) ) = ( 0, \rho_\Lambda' (0) ) \text{.}
\end{equation}
Note $f_\pm$ and $\rho_\Lambda$ have the same initial data at $t = 0$, and $f_\pm$ has the explicit form
\begin{equation}
\label{eql.geodesic_A1_3} f_\pm ( \tau ) = \frac{ \rho_\Lambda' (0) }{ \sqrt{ c_\pm^2 - b^2 } } \cdot e^{ \pm b \tau } \cdot \sin \left( \sqrt{ c_\pm^2 - b^2 } \cdot \tau \right) \text{.}
\end{equation}
In particular, $f_\pm$ and $f_\pm'$ are positive on $( 0, T_\pm )$, and $f_\pm'$ vanishes at $t = T_\pm$.

Since $f_\pm \geq 0$ on $[ 0, T_\pm ]$, then \eqref{eql.geodesic_A1_1}, \eqref{eql.geodesic_A1_2}, and integrations by parts yield
\begin{align}
\label{eql.geodesic_A1_10} 0 &\geq \int_0^{ \tau_+ } [ ( e^{ - 2 b t } \rho_\Lambda' )' + c_+^2 e^{ - 2 b t } \rho_\Lambda ] f_+ - \int_0^{ \tau_+ } \rho_\Lambda [ ( e^{ - 2 b t } f_+' )' + c_+^2 e^{ - 2 b t } f_+ ] \\
\notag &= e^{ - 2 b \tau_+ } \rho_\Lambda' ( \tau_+ ) f_+ ( \tau_+ ) - e^{ - 2 b \tau_+ } \rho_\Lambda ( \tau_+ ) f_+' ( \tau_+ ) \text{,} \\
\notag 0 &\leq \int_0^{ \tau_- } [ ( e^{ 2 b t } \rho_\Lambda' )' + c_-^2 e^{ 2 b t } \rho_\Lambda ] f_- - \int_0^{ \tau_- } \rho_\Lambda [ ( e^{ 2 b t } f_-' )' + c_-^2 e^{ 2 b t } f_- ] \\
\notag &= e^{ + 2 b \tau_- } \rho_\Lambda' ( \tau_- ) f_- ( \tau_- ) - e^{ + 2 b \tau_- } \rho_\Lambda ( \tau_- ) f_-' ( \tau_- ) \text{,}
\end{align}
for all $\tau_\pm \in \mc{A}_1 \cap [ 0, T_\pm ]$.
Since $\rho_\Lambda, f_\pm > 0$ on $\mc{A}_1 \cap ( 0, T_\pm )$, the above can be rearranged as
\begin{align*}
( \log \rho_\Lambda )' ( \tau_+ ) \leq ( \log f_+ )' ( \tau_+ ) \text{,} &\qquad \tau_+ \in \mc{A}_1 \cap ( 0, T_+ ] \text{,} \\
( \log \rho_\Lambda )' ( \tau_- ) \geq ( \log f_- )' ( \tau_- ) \text{,} &\qquad \tau_- \in \mc{A}_1 \cap ( 0, T_- ] \text{.}
\end{align*}
Integrating the above from $t = 0$ and noting that $\log$ is monotone, we obtain
\begin{align}
\label{eql.geodesic_A1_11} \frac{ \rho_\Lambda ( \tau_+ ) }{ f_+ ( \tau_+ ) } \leq \lim_{ \tau \searrow 0 } \frac{ \rho_\Lambda ( \tau ) }{ f_+ ( \tau ) } = 1 \text{,} &\qquad \tau_+ \in \mc{A}_1 \cap ( 0, T_+ ] \text{,} \\
\notag \frac{ \rho_\Lambda ( \tau_- ) }{ f_- ( \tau_- ) } \geq \lim_{ \tau \searrow 0 } \frac{ \rho_\Lambda ( \tau ) }{ f_- ( \tau ) } = 1 \text{,} &\qquad \tau_- \in \mc{A}_1 \cap ( 0, T_- ] \text{.}
\end{align}
Combining \eqref{eql.geodesic_A1_10} and \eqref{eql.geodesic_A1_11} then yields
\begin{align}
\label{eql.geodesic_A1_12} \rho_\Lambda' ( \tau_+ ) \leq \frac{ \rho_\Lambda ( \tau_+ ) f_+' ( \tau_+ ) }{ f_+ ( \tau_+ ) } \leq f_+' ( \tau_+ ) \text{,} &\qquad \tau_+ \in \mc{A}_1 \cap ( 0, T_+ ] \text{,} \\
\notag \rho_\Lambda' ( \tau_- ) \geq \frac{ \rho_\Lambda ( \tau_- ) f_-' ( \tau_- ) }{ f_- ( \tau_- ) } \geq f_-' ( \tau_- ) \text{,} &\qquad \tau_- \in \mc{A}_1 \cap ( 0, T_- ] \text{.}
\end{align}

Suppose $\mc{A}_1 \supseteq ( 0, T_+ + \delta ]$ for some $\delta > 0$.
Then, $T_1 > T_+$, but \eqref{eql.geodesic_A1_12} implies
\[
\rho_\Lambda' ( T_+ ) \leq f_+' ( T_+ ) = 0 \text{,} \]
contradicting that $T_1$ is the first vanishing time for $\rho_\Lambda'$.
This proves $\mc{A}_1 \subseteq ( 0, T_+ ]$.
Moreover, for any $\tau \in \mc{A}_1$, we use \eqref{eql.geodesic_A1_3}, \eqref{eql.geodesic_A1_11}, and the fact that $\mc{A}_1 \subseteq ( 0, T_+ ]$ to obtain \eqref{eq.geodesic_A1}:
\[
\rho_\Lambda ( \tau ) \leq f_+ ( \tau ) \lesssim_{ B, C_+ } \rho_\Lambda' (0) \text{.}
\]

Finally, suppose $T_1 \in \mc{A}_1$.
Then, the above yields $T_1 \leq T_+$.
Moreover, since $T_2 \not\in \mc{A}$ by definition, \eqref{eql.geodesic_AA} implies $T_1 < T_2$, and hence $\rho_\Lambda' ( T_1 ) = 0$.
The inequality $T_- \leq T_1$ also holds, since the opposite statement $T_- > T_1$ implies $f_-' ( T_1 ) > 0$ and hence contradicts the second part of \eqref{eql.geodesic_A1_12}.
\end{proof}

\begin{lemma} \label{thm.geodesic_A2}
$\mc{A}_2 \subseteq [ T_1, T_1 + T_+ ]$, and the following estimate holds:
\begin{equation}
\label{eq.geodesic_A2} \rho_\Lambda ( \tau ) \lesssim_{ B, C_+ } \rho_\Lambda' (0) \text{,} \qquad \tau \in \mc{A}_2 \text{.}
\end{equation}
Furthermore, if $T_2$ is a limit point of $\mc{A}_2$, then the following hold:
\begin{equation}
\label{eq.geodesic_A2_end} T_2 = \ell_+ \text{,} \qquad T_1 + T_- \leq \ell_+ \leq T_1 + T_+ \text{,} \qquad \lim_{ \tau \nearrow \ell_+ } \rho_\Lambda ( \tau ) = 0 \text{.}
\end{equation}
\end{lemma}

\begin{proof}
We can assume that $T_1 < T_2$ and that $\mc{A}_2 \cap ( T_1, T_2 )$ is nonempty (and hence $\mc{A}_1 = ( 0, T_1 ]$), since otherwise there is nothing to prove.
First, we claim that the following hold:
\begin{equation}
\label{eql.geodesic_A2_0} \rho_\Lambda' ( T_1 ) = 0 \text{,} \qquad \rho_\Lambda' |_{ \mc{A}_2 \cap ( T_1, T_2 ) } < 0 \text{.}
\end{equation}
To see this, we observe that \eqref{eq.geodesic_asymp} can be equivalently expressed as (see also \eqref{eql.geodesic_UV})
\[
\left( e^{ \int_0^t \mc{U} } \rho_\Lambda' \right)' + \left( \mc{V} e^{ \int_0^t \mc{U} } \right) \rho_\Lambda = 0 \text{.}
\]
Since $\rho_\Lambda > 0$ on $( 0, T_2 )$, and since $\mc{V} > 0$ on $\mc{A}$ by \eqref{eql.geodesic_A}, the above yields that $e^{ \int_0^t \mc{U} } \rho_\Lambda'$ is strictly decreasing on $\mc{A}$. 
Since $T_1 \in \mc{A}_1$ by assumption, the claim \eqref{eql.geodesic_A2_0} now follows from \eqref{eq.geodesic_A1_end}.

Next, from \eqref{eq.geodesic_asymp}, \eqref{eql.geodesic_A}, \eqref{eql.geodesic_AA}, and \eqref{eql.geodesic_A2_0}, we obtain that on $\mc{A}_2$,
\begin{align}
\label{eql.geodesic_A2_1} \rho_\Lambda'' + 2 b \rho_\Lambda' + c_+^2 \rho_\Lambda \leq 0 \text{,} &\qquad ( e^{ + 2 b t } \rho_\Lambda' )' + c_+^2 e^{ + 2 b t } \rho_\Lambda \leq 0 \text{,} \\
\notag \rho_\Lambda'' - 2 b \rho_\Lambda' + c_-^2 \rho_\Lambda \geq 0 \text{,} &\qquad ( e^{ - 2 b t } \rho_\Lambda' )' + c_-^2 e^{ - 2 b t } \rho_\Lambda \geq 0 \text{.}
\end{align}
Let $f_\pm$ now denote the solutions of the initial value problems
\begin{equation}
\label{eql.geodesic_A2_2} ( e^{ \pm 2 b t } f_\pm' )' + c_\pm^2 e^{ \pm 2 b t } f_\pm = 0 \text{,} \qquad ( f_\pm ( T_1 ), f_\pm' ( T_1 ) ) = ( \rho_\Lambda ( T_1 ), 0 ) \text{,}
\end{equation}
so that $f_\pm$ and $\rho_\Lambda$ have the same initial data at $t = T_1$.
Observe that $f_\pm$ has explicit form
\begin{equation}
\label{eql.geodesic_A2_3} f_\pm ( \tau ) = \frac{ \rho_\Lambda ( T_1 ) \cdot \exp \left( \mp b \left( \tau + \frac{ \pi }{ \sqrt{ c_\pm^2 - b^2 } } - T_\pm - T_1 \right) \right) \cdot \sin \left( \sqrt{ c_\pm^2 - b^2 } \cdot ( T_1 + T_\pm - \tau ) \right) }{ \exp \left( \mp b \left( \frac{ \pi }{ \sqrt{ c_\pm^2 - b^2 } } - T_\pm \right) \right) \cdot \sin \left( \sqrt{ c_\pm^2 - b^2 } \cdot T_\pm \right) } \text{,}
\end{equation}
that $f_\pm$ is positive on $[ T_1, T_1 + T_\pm )$, and that $f_\pm'$ is negative on $( T_1, T_1 + T_\pm ]$.

Recalling \eqref{eql.geodesic_A2_1} and \eqref{eql.geodesic_A2_2}, we obtain, from an integration by parts, that
\begin{align}
\label{eql.fg_geodesic_A2_10} 0 &\geq \int_{ T_1 }^{ \tau_+ } [ ( e^{ + 2 b t } \rho_\Lambda' )' + c_+^2 e^{ + 2 b t } \rho_\Lambda ] f_+ - \int_{ T_1 }^{ \tau_+ } \rho_\Lambda [ ( e^{ + 2 b t } f_+' )' + c_+^2 e^{ + 2 b t } f_+ ] \\
\notag &= e^{ + 2 b \tau_+ } \rho_\Lambda' ( \tau_+ ) f_+ ( \tau_+ ) - e^{ + 2 b \tau_+ } \rho_\Lambda ( \tau_+ ) f_+' ( \tau_+ ) \text{,} \\
\notag 0 &\leq \int_{ T_1 }^{ \tau_- } [ ( e^{ - 2 b t } \rho_\Lambda' )' + c_-^2 e^{ - 2 b t } \rho_\Lambda ] f_- - \int_{ T_1 }^{ \tau_- } \rho_\Lambda [ ( e^{ - 2 b t } f_-' )' + c_-^2 e^{ - 2 b t } f_- ] \\
\notag &= e^{ - 2 b \tau_- } \rho_\Lambda' ( \tau_- ) f_- ( \tau_- ) - e^{ - 2 b \tau_- } \rho_\Lambda ( \tau_- ) f_-' ( \tau_- ) \text{,}
\end{align}
for all $\tau_\pm \in \mc{A}_2 \cap [ T_1, T_1 + T_\pm ]$.
The above can then be rearranged as
\begin{align*}
( \log \rho_\Lambda )' ( \tau_+ ) \leq ( \log f_+ )' ( \tau_+ ) \text{,} &\qquad \tau_+ \in \mc{A}_2 \cap [ T_1, T_1 + T_+ ] \text{,} \\
( \log \rho_\Lambda )' ( \tau_- ) \geq ( \log f_- )' ( \tau_- ) \text{,} &\qquad \tau_- \in \mc{A}_2 \cap [ T_1, T_1 + T_- ] \text{.}
\end{align*}
As in the proof of Lemma \ref{thm.geodesic_A1}, integrating the above yields
\begin{align}
\label{eql.geodesic_A2_11} \frac{ \rho_\Lambda ( \tau_+ ) }{ f_+ ( \tau_+ ) } \leq \frac{ \rho_\Lambda ( T_1 ) }{ f_+ ( T_1 ) } = 1 \text{,} &\qquad \tau_+ \in \mc{A}_2 \cap ( T_1, T_1 + T_+ ] \text{,} \\
\notag \frac{ \rho_\Lambda ( \tau_- ) }{ f_- ( \tau_- ) } \geq \frac{ \rho_\Lambda ( T_1 ) }{ f_- ( T_1 ) } = 1 \text{,} &\qquad \tau_- \in \mc{A}_2 \cap ( T_1, T_1 + T_- ] \text{.}
\end{align}

Suppose $\mc{A}_2 \supseteq [ T_1, T_1 + T_+ + \delta ]$ for some $\delta > 0$.
Then, $T_2 > T_1 + T_+$, but \eqref{eql.geodesic_A2_11} implies
\[
0 \leq \lim_{ \tau \nearrow T_1 + T_+ } \rho_\Lambda ( \tau ) \leq f_+ ( T_1 + T_+ ) = 0 \text{,}
\]
contradicting that $\Lambda$ is well-defined at $T_1 + T_+$.
As a result, $\mc{A}_2 \subseteq [ T_1, T_1 + T_+ ]$.
In addition, applying \eqref{eq.geodesic_A1}, \eqref{eql.geodesic_A2_3}, \eqref{eql.geodesic_A2_11}, and that $\mc{A}_2 \subseteq [ T_1, T_1 + T_+ ]$, we obtain \eqref{eq.geodesic_A2}:
\[
\rho_\Lambda ( \tau ) \leq f_+ ( \tau ) \lesssim_{ B, C_+ } \rho_\Lambda ( T_1 ) \lesssim \rho_\Lambda' (0) \text{,} \qquad \tau \in \mc{A}_2 \text{.}
\]

Finally, suppose $T_2$ is a limit point of $\mc{A}_2$.
Since $\mc{A}_2 \subseteq [ T_1, T_1 + T_+ ]$, then \eqref{eq.geodesic_A1_end} yields
\[
T_2 \leq T_1 + T_+ \leq 2 T_+ < \mc{T}_+ ( B, C_+ ) + \eta \text{.}
\]
The above, along with the definition \eqref{eql.geodesic_T} of $T_2$, implies $T_2 = \ell_+$ and eliminates option (3) in Lemma \ref{thm.geodesic_terminal}.
Furthermore, if $\varepsilon$ is small enough, then \eqref{eq.geodesic_ncT} and \eqref{eq.geodesic_A2} rule out option (2) in Lemma \ref{thm.geodesic_terminal}.
As a result, we conclude that $\rho_\Lambda ( \tau ) \rightarrow 0$ as $\tau \nearrow \ell_+$.

To prove \eqref{eq.geodesic_A2_end}, it remains to show $\ell_+ \geq T_1 + T_-$.
For this, we note from \eqref{eql.geodesic_A2_11} that the opposite statement $T_1 + T_- > \ell_+$ implies $\rho ( \ell_+ ) \geq f_- ( \ell_+ ) > 0$, which is a contradiction.
\end{proof}

Thus far, the crucial result \eqref{eq.geodesic_A2_end} is conditional upon assuming that $T_2$ is a limit point of $\mc{A}_2$.
The remaining task is to show, via a continuity argument, that this property must hold.

First, note that since $\bar{\mi{D}}$ has compact cross-sections, then
\begin{equation}
\label{eql.geodesic_dc} \mi{D}_c := \bar{\mi{D}} \cap \{ 0 \leq t \leq \mc{T}_+ ( B, C_+ ) + \eta \}
\end{equation}
is compact.
Thus, $\mi{D}_c$ can be covered by a finite family of compact coordinate systems:
\begin{equation}
\label{eql.geodesic_coord} \Xi_c = \{ ( U_1, \varphi_1 ), \dots, ( U_N, \varphi_N ) \} \text{.}
\end{equation}

The key technical estimates for our continuity argument are captured in the subsequent lemma:

\begin{lemma} \label{thm.geodesic_tech}
Fix $\tau \in ( 0, \ell_+ )$, and suppose that
\begin{equation}
\label{eq.geodesic_tech_rho} | \rho_\Lambda ( \tau ) | \leq d \varepsilon \text{,} \qquad | \rho_\Lambda' ( \tau ) | \leq d \varepsilon \text{,} \qquad d > 0 \text{.}
\end{equation}
Then, the following statements hold:
\begin{itemize}
\item If $\lambda_\Lambda ( \tau ) \in U_i$, where $1 \leq i \leq N$, then the $\varphi_i$-coordinate components of $\lambda_\Lambda'$ satisfy
\begin{equation}
\label{eq.geodesic_tech_lambda} | [ \lambda_\Lambda' ( \tau ) ]^a | \lesssim_{ \mi{D}, \gv, t, d } 1 \text{.}
\end{equation}

\item $\lambda_\Lambda' ( \tau )$ is ``almost $\gm$-null":
\begin{equation}
\label{eq.geodesic_tech_null} | \gm ( \lambda_\Lambda' ( \tau ), \lambda_\Lambda' ( \tau ) ) | \lesssim_{ \mi{D}, \gv, t, d } \varepsilon^2 \text{.}
\end{equation}

\item Let $b_0$, $c_{ 0, + }$, $c_{ 0, - }$ denote the midpoints of the intervals $( B, b )$, $( c_+, C_+ )$, $( C_-, c_- )$, respectively.
Then, for sufficiently small $\varepsilon > 0$ (depending on $\mi{D}$, $\gv$, $t$, and $d$), we have that
\begin{equation}
\label{eq.geodesic_tech_boot} c_{ 0, - } > - \gs ( \lambda_\Lambda' ( \tau ), \lambda_\Lambda' ( \tau ) ) > c_{ 0, + } \text{,} \qquad | \Dm^2 t ( \lambda_\Lambda' ( \tau ), \lambda_\Lambda' ( \tau ) ) | < 2 b_0 \text{.}
\end{equation}
\end{itemize}
\end{lemma}

\begin{proof}
First, note that by the assumption \eqref{eq.geodesic_t}, we have
\begin{equation}
\label{eql.geodesic_tech_pre} dt [ \lambda_\Lambda' ( \tau ) ] = 1 \text{.}
\end{equation}
Moreover, using \eqref{eq.aads_metric} and the fact that $\Lambda$ is $g$-null, we deduce
\begin{equation}
\label{eql.geodesic_tech_0} 0 = \rho^2 \cdot g ( \Lambda' ( \tau ), \Lambda' ( \tau ) ) = [ \rho_\Lambda' ( \tau ) ]^2 + \gv ( \lambda_\Lambda' ( \tau ), \lambda_\Lambda' ( \tau ) ) \text{.}
\end{equation}
which, when combined with \eqref{eq.geodesic_tech_rho}, yields
\[
0 < - \gv ( \lambda_\Lambda' ( \tau ), \lambda_\Lambda' ( \tau ) ) \leq d^2 \varepsilon^2 \text{.}
\]
Thus, by \eqref{eq.aads_time}, \eqref{eql.geodesic_tech_pre}, and the above, we have
\begin{equation}
\label{eql.geodesic_tech_1} \ms{h} ( \lambda_\Lambda', \lambda_\Lambda' ) \lesssim_{ \gv, t } 1 \text{,}
\end{equation}
where $\ms{h}$ is the vertical Riemannian metric from Definition \ref{def.aads_riemann}.
Since the $| \ms{h} |_{ 0, \varphi_i }$'s are uniformly bounded by \eqref{eq.aads_strong_limits} and \eqref{eq.aads_time}, the bounds in \eqref{eq.geodesic_tech_lambda} now follow from \eqref{eql.geodesic_tech_1}.

Next, recalling \eqref{eq.aads_strong_limits} and \eqref{eql.geodesic_tech_0}, we obtain
\begin{equation}
\label{eql.geodesic_tech_10} 0 = [ \rho_\Lambda' ( \tau ) ]^2 + \gm ( \lambda_\Lambda' ( \tau ), \lambda_\Lambda' ( \tau ) ) + \rho_\Lambda^2 \cdot \ms{r}_\ast ( \lambda_\Lambda' ( \tau ), \lambda_\Lambda' ( \tau ) ) \text{,}
\end{equation}
where $\ms{r}_\ast$ is a vertical tensor field that has the boundary limit $2 \ms{r}_\ast \rightarrow^1 \gs$.
Since $\lambda_\Lambda$ lies within the compact region $\mi{D}_c$, it follows from \eqref{eq.geodesic_tech_rho}, \eqref{eq.geodesic_tech_lambda}, and \eqref{eql.geodesic_tech_10} that
\[
| \gm ( \lambda_\Lambda' ( \tau ), \lambda_\Lambda' ( \tau ) ) | \leq d^2 \varepsilon^2 \left( 1 + \sup_{ 1 \leq j \leq N } \sup_{ ( 0, \rho_0 ] \times U_j } | \ms{r}_\ast |_{ 0, \varphi_j } \right) \text{,}
\]
from which \eqref{eq.geodesic_tech_null} immediately follows.

Finally, \eqref{eq.geodesic_tech_null} and \eqref{eql.geodesic_tech_pre} together imply that $\lambda_\Lambda' ( \tau )$ is $\varepsilon$-close to a $\gm$-null vector $Z$ that satisfies $Z t = 1$.
Thus, by taking $\varepsilon$ sufficiently small (depending on $\mi{D}$, $\gv$, $t$, and $d$), and by noting that $\gs$ and $\Dm^2 t$ are uniformly continuous on the compact region $\mi{D}_c$, we conclude that the final estimates \eqref{eq.geodesic_tech_boot} follow from the null convexity criterion \eqref{eq.aads_nc} and the assumption \eqref{eq.geodesic_nc_reverse}.
\end{proof}

The next lemma completes our continuity argument:

\begin{lemma} \label{thm.geodesic_boot}
$T_2 = \ell_+$, $\mc{A} = ( 0, \ell_+ )$, and $2 T_- \leq \ell_+ \leq 2 T_+$.
\end{lemma}

\begin{proof}
First, observe that by continuity, $\mc{A}$ is a closed subset of $( 0, T_2 )$.

In addition, notice that by \eqref{eq.geodesic_init} and \eqref{eq.geodesic_ncT}, we have that \eqref{eq.geodesic_tech_rho} holds---say with $d := 2$---in an interval $( 0, \delta )$, for some $\delta > 0$ (depending on $\varepsilon$).
Thus, by taking $\varepsilon$ sufficiently small (depending on $\gv$ and $t$), and by recalling \eqref{eq.geodesic_remainder}, \eqref{eql.geodesic_UV}, \eqref{eql.geodesic_A}, and \eqref{eq.geodesic_tech_boot}, we also obtain that $\mc{A}$ is non-empty.

Suppose next that $T' \in \mc{A}$, so $( 0, T' ] \subseteq \mc{A}$ as well by definition.
Writing \eqref{eq.geodesic_asymp} as
\[
\left( e^{ \int_0^t \mc{U} } \rho_\Lambda' \right)' + \left( \mc{V} e^{ \int_0^t \mc{U} } \right) \rho_\Lambda = 0 \text{.}
\]
and applying \eqref{eql.geodesic_A}, \eqref{eq.geodesic_A1}, and \eqref{eq.geodesic_A2} to the above, we obtain on $\mc{A}$ that
\[
\left| \left( e^{ \int_0^t \mc{U} } \rho_\Lambda' \right)' \right| \lesssim_{ B, C_\pm } \rho_\Lambda \lesssim_{ B, C_+ } \rho_\Lambda' (0) \text{.}
\]
Integrating the above from $0$ and recalling \eqref{eq.geodesic_ncT}, \eqref{eq.geodesic_A1}, and \eqref{eq.geodesic_A2}, we conclude that
\begin{equation}
\label{eql.geodesic_boot_0} \rho_\Lambda ( \tau ) \lesssim_{ B, C_\pm } \varepsilon \text{,} \qquad | \rho_\Lambda' ( \tau ) | \lesssim_{ B, C_\pm } \varepsilon \text{,} \qquad \tau \in ( 0, T' ] \text{,}
\end{equation}
that is, that \eqref{eq.geodesic_tech_rho} holds all $\tau \in ( 0, T' ]$ and for some $d$ that depends on $B$ and $C_\pm$ (but not on $\varepsilon$).

Now, note that the limits in \eqref{eq.geodesic_remainder} are uniform in the compact region $\mi{D}_c$, that is,
\begin{equation}
\label{eql.geodesic_boot_1} \lim_{ \sigma \searrow 0 } \sup_{ 1 \leq i \leq N } \sup_{ \{ \sigma \} \times U_i } ( | \ms{r}_0 |_{ 0, \varphi_i } + | \ms{r}_1 |_{ 0, \varphi_i } + | \ms{r}_2 |_{ 0, \varphi_i } ) = 0 \text{.}
\end{equation}
Taking $\varepsilon$ to be sufficiently small (depending on $\mi{D}$, $\gv$, and $t$), and recalling that $\lambda_\Lambda$ lies within $\mi{D}_c$, we see from \eqref{eq.geodesic_tech_lambda}, \eqref{eql.geodesic_boot_0}, and \eqref{eql.geodesic_boot_1} that the error terms
\[
| \ms{r}_0 ( \lambda_\Lambda' ( \tau ), \lambda_\Lambda' ( \tau ) ) | \text{,} \qquad | \ms{r}_1 ( \lambda_\Lambda' ( \tau ), \lambda_\Lambda' ( \tau ) ) | \text{,} \qquad | \ms{r}_2 ( \lambda_\Lambda' ( \tau ) ) |
\]
can be made arbitrarily and uniformly small for all $\tau \in ( 0, T' ]$.
Therefore, by combining the above with \eqref{eql.geodesic_UV} and \eqref{eq.geodesic_tech_boot}, and by further shrinking $\varepsilon$ if needed, we obtain
\begin{equation}
\label{eql.geodesic_boot_2} \mc{U} ( T' ) < 2 b_0 \text{,} \qquad c_{ 0, - }^2 > \mc{V} ( T' ) > c_{ 0, + }^2 \text{.}
\end{equation}

By \eqref{eql.geodesic_A}, \eqref{eql.geodesic_boot_2}, and the continuity of $\mc{U}$ and $\mc{V}$, we see that $\mc{A}$ must contain some open interval that includes $T'$.
As a result, we have shown that $\mc{A}$ is open in $( 0, T_2 )$, and hence $\mc{A} = ( 0, T_2 )$.
The conclusions of the lemma now follow immediately from \eqref{eq.geodesic_A1_end} and \eqref{eq.geodesic_A2_end}.
\end{proof}

Finally, the proof of Theorem \ref{thm.geodesic} is concluded with the following:
\begin{itemize}
\item Since Lemma \ref{thm.geodesic_boot} implies $T_2$ is a limit point of $A_2$, then \eqref{eq.geodesic_A2_end} immediately yields \eqref{eq.geodesic_return}.

\item Lemma \ref{thm.geodesic_boot} and the estimates \eqref{eq.geodesic_A1} and \eqref{eq.geodesic_A2} result in \eqref{eq.geodesic_bound}.

\item \eqref{eql.geodesic_BC} and Lemma \ref{thm.geodesic_boot} imply \eqref{eq.geodesic_timespan}.
\end{itemize}

\section{The Carleman Estimate} \label{sec.carleman}

In this section, we state and prove our main Carleman estimates, in terms of the setting and the language developed in Section \ref{sec.aads}.
As mentioned in the introduction, these estimates improve upon the corresponding Carleman inequalities in \cite[Theorems 3.7 and C.1]{hol_shao:uc_ads_ns}.

\subsection{The Carleman Weight}

In order to state our main estimates, we must first describe the weight function that we will use.
In fact, this weight will correspond closely to the functions $f_+$ in Lemmas \ref{thm.geodesic_A1} and \ref{thm.geodesic_A2} that dominated the trajectories of null geodesics.

\begin{definition} \label{def.carleman_eta}
Given constants $0 \leq b < c$, we define the function $\eta := \eta [ b, c ]$ by
\footnote{See \eqref{eq.geodesic_nc_W} for the definition of $\mc{W}_+ ( b, c )$.}
\begin{equation}
\label{eq.carleman_eta} \eta: \R \rightarrow \R \text{,} \qquad \eta ( \tau ) = e^{ - b | \tau | } \cdot \sin \left( \mc{W}_+ ( b, c ) - \sqrt{ c^2 - b^2 } \cdot | \tau | \right) \text{.}
\end{equation}
In general, we omit the dependence of $\eta$ on the constants $b$ and $c$ in our notations.
However, in settings where this association might be ambiguous, we write $\eta [ b, c ]$ in the place of $\eta$.
\end{definition}

The subsequent proposition establishes some basic properties of the function $\eta$ from Definition \ref{def.carleman_eta} and connects it to some constructions in the proof of Theorem \ref{thm.geodesic}:
\footnote{In particular, one should compare the ODEs \eqref{eq.carleman_eta_minus} and \eqref{eq.carleman_eta_plus} to \eqref{eql.geodesic_A1_2} and \eqref{eql.geodesic_A2_2}.}

\begin{proposition} \label{thm.carleman_eta}
Let $0 \leq b < c$ and $\eta$ be as in Definition \ref{def.carleman_eta}.
Then:
\begin{itemize}
\item $\eta$ satisfies the following properties on the negative half-line $( -\infty, 0 )$:
\begin{equation}
\label{eq.carleman_eta_minus} \eta'' - 2 b \eta' + c^2 \eta = 0 \text{,} \qquad ( \eta, \eta' ) \left( - \frac{1}{2} \mc{T}_+ ( b, c ) \right) = \left( 0, + \sqrt{ c^2 - b^2 } \right) \text{.}
\end{equation}

\item $\eta$ satisfies the following properties on the positive half-line $( 0, +\infty )$:
\begin{equation}
\label{eq.carleman_eta_plus} \eta'' + 2 b \eta' + c^2 \eta = 0 \text{,} \qquad ( \eta, \eta' ) \left( + \frac{1}{2} \mc{T}_+ ( b, c ) \right) = \left( 0, - \sqrt{ c^2 - b^2 } \right) \text{.}
\end{equation}

\item The following properties hold for $\eta$:
\begin{equation}
\label{eq.carleman_eta_range} 0 < \eta ( \tau ) \leq 1 \text{,} \quad | \tau | < \frac{1}{2} \mc{T}_+ ( b, c ) \text{.}
\end{equation}

\item $\eta \in C^2 ( \R )$, and $\eta$ is smooth on $\R \setminus \{ 0 \}$.
\end{itemize}
\end{proposition}

\begin{proof}
By direct computations, one can show that $\eta$, as defined in \eqref{eq.carleman_eta}, satisfies \eqref{eq.carleman_eta_minus}--\eqref{eq.carleman_eta_range} and is smooth on $\R \setminus \{ 0 \}$.
Finally, differentiating \eqref{eq.carleman_eta} near the origin yields that $\eta' (0) = 0$, while the equations \eqref{eq.carleman_eta_minus} and \eqref{eq.carleman_eta_plus} then imply $\eta'' (0)$ exists; this proves that $\eta \in C^2 ( \R )$.
\end{proof}

\begin{remark}
On the other hand, $\eta$ fails to be $C^3$; in particular, one can show that
\[
\lim_{ \tau \nearrow 0 } \eta''' ( \tau ) < 0 < \lim_{ \tau \searrow 0 } \eta''' ( \tau ) \text{.}
\]
\end{remark}

\begin{remark}
Note that $\eta$ is a direct analogue of the function $\eta$ from \cite[Definition 2.12]{hol_shao:uc_ads_ns}.
\end{remark}

Next, we define our Carleman weight function and its associated domains:

\begin{definition} \label{def.carleman_f}
Let $( \mi{M}, g )$ be a strongly FG-aAdS segment, and let $t$ be a global time.
In addition, fix constants $0 \leq b < c$, let $\eta := \eta [ b, c ]$ be as in Definition \ref{def.carleman_eta}, and let $t_0 \in \R$ satisfy
\begin{equation}
\label{eq.carleman_timespan} t_- < t_0 - \frac{1}{2} \mc{T}_+ ( b, c ) < t_0 + \frac{1}{2} \mc{T}_+ ( b, c ) < t_+ \text{.}
\end{equation}
\begin{itemize}
\item We define the domain $\Omega_{ t_0 } := \Omega_{ t_0 } [ b, c ]$ by
\begin{equation}
\label{eq.carleman_omega} \Omega_{ t_0 } := \left\{ P \in \mi{M} \, \middle| \, | t ( P ) - t_0 | < \frac{1}{2} \mc{T}_+ ( b, c ) \right\} \text{.}
\end{equation}

\item We define the function $f_{ t_0 } := f_{ t_0 } [ b, c ]$ by
\footnote{Note that $f_{ t_0 }$ is well-defined and everywhere positive on $\Omega_{ t_0 }$, by \eqref{eq.carleman_eta_range}.}
\begin{equation}
\label{eq.carleman_f} f_{ t_0 }: \Omega_{ t_0 } \rightarrow \R \text{,} \qquad f_{ t_0 } := \frac{ \rho }{ \eta ( t - t_0 ) } \text{.}
\end{equation}

\item In addition, given $f_\ast > 0$, we define the region
\begin{equation}
\label{eq.carleman_Omega_f} \Omega_{ t_0 } ( f_\ast ) := \Omega_{ t_0 } \cap \{ f_{ t_0 } < f_\ast \} \text{.}
\end{equation}
\end{itemize}
For future notational convenience, we will also abbreviate $f := f_0$.
\end{definition}

\begin{remark}
In particular, Proposition \ref{thm.carleman_eta} and \eqref{eq.carleman_f} imply $f_{ t_0 } \in C^2 ( \Omega_{ t_0 } )$.
Moreover:
\begin{itemize}
\item $f_{ t_0 }$ is smooth on both $\smash{ \Omega_{ t_0, < } } := \Omega_{ t_0 } \cap \{ t < t_0 \}$ and $\smash{ \Omega_{ t_0, > } } := \Omega_{ t_0 } \cap \{ t > t_0 \}$.

\item $f_{ t_0 }$ can be smoothly extended to the closures $\bar{\Omega}_{ t_0, < }$ and $\bar{\Omega}_{ t_0, > }$ in $\mi{M}$.
\end{itemize}
\end{remark}

\begin{figure}[ht]
\includegraphics{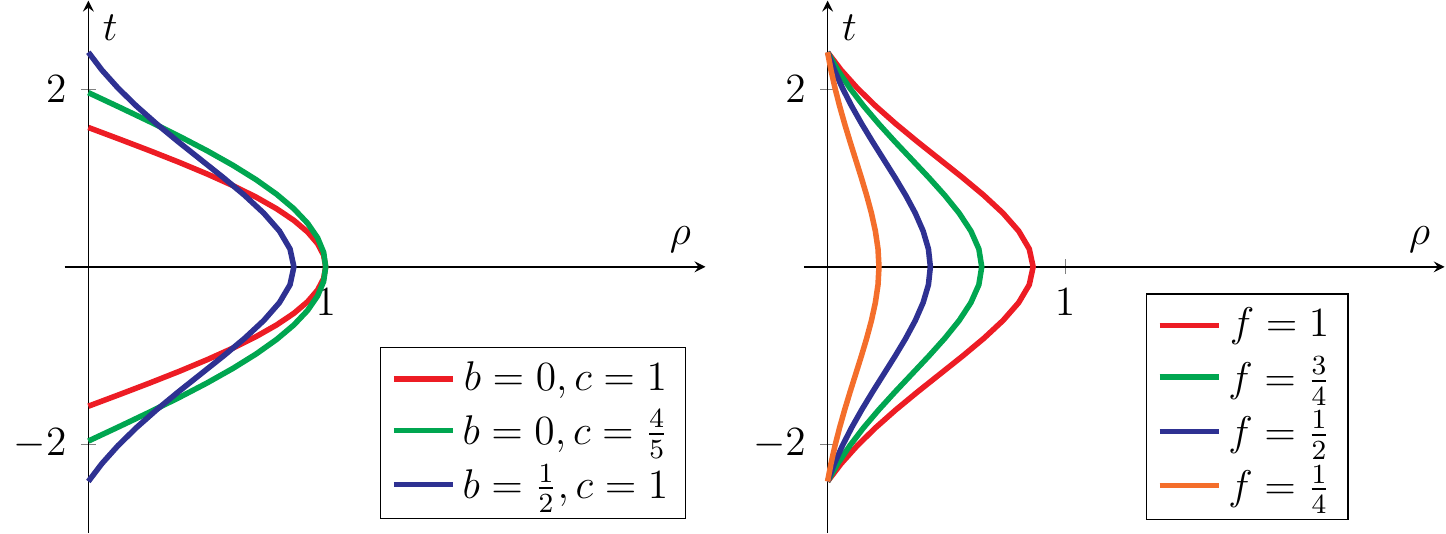}
\caption{The plots depict various level sets of the function $f = f_0$ on $\rho$-$t$ planes.
The left graphic shows the level set $f = 1$, with various parameters $( b, c )$, while the right graphic shows various level sets of $f$, with fixed parameters $( b, c ) = ( \frac{1}{2}, 1 )$.}
\end{figure}

We now collect some basic properties of the $f_{ t_0 }$'s from Definition \ref{def.carleman_f}:
\footnote{In particular, this is the analogue of \cite[Proposition 2.16]{hol_shao:uc_ads_ns} in the present setting and language.}

\begin{proposition} \label{thm.carleman_f}
Assume the setting of Definition \ref{def.carleman_f} (in particular, the constants $b, c, t_0$).
Then:
\begin{itemize}
\item $f_{ t_0 }$ and its derivatives satisfy the following identities:
\begin{align}
\label{eq.carleman_f_grad} \nabla^\sharp f_{ t_0 } &= \rho f_{ t_0 } \partial_\rho - \rho f_{ t_0 }^2 \cdot \eta' ( t - t_0 ) \cdot \Dv^\sharp t \text{,} \\
\notag g ( \nabla^\sharp f_{ t_0 }, \nabla^\sharp f_{ t_0 } ) &= f_{ t_0 }^2 + f_{ t_0 }^4 [ \eta' ( t - t_0 ) ]^2 \cdot \gv ( \Dv^\sharp t, \Dv^\sharp t ) \text{,} \\
\notag \Box f_{ t_0 } &= - ( n - 1 ) f_{ t_0 } + \{ 2 f_{ t_0 }^3 [ \eta' ( t - t_0 ) ]^2 - \rho f_{ t_0 }^2 \cdot \eta'' ( t - t_0 ) \} \cdot \gv ( \Dv^\sharp t, \Dv^\sharp t ) \\
\notag &\qquad + \frac{1}{2} \rho f_{ t_0 } \cdot \trace{\gv} \mi{L}_\rho \gv - \rho f_{ t_0 }^2 \cdot \eta' ( t - t_0 ) \cdot \trace{\gv} \Dv^2 t \text{.}
\end{align}

\item Let $( U, \varphi )$ be an arbitrary coordinate system on $\mi{I}$.
Then, the components of $\nabla^2 f_{ t_0 }$, with respect to $\varphi$- and $\varphi_\rho$-coordinates, satisfy the following identities:
\begin{align}
\label{eq.carleman_f_hessian} \nabla_{ \rho \rho } f_{ t_0 } &= \rho^{-2} f_{ t_0 } \text{,} \\
\notag \nabla_{ \rho a } f_{ t_0 } &= - 2 \rho^{-2} f_{ t_0 }^2 \cdot \eta' ( t - t_0 ) \cdot \Dv_a t + \frac{1}{2} \rho^{-1} f_{ t_0 }^2 \cdot \eta' ( t - t_0 ) \cdot \gv^{ b c } \Dv_b t \mi{L}_\rho \gv_{ a c } \text{,} \\
\notag \nabla_{ a b } f_{ t_0 } &= \{ 2 \rho^{-2} f_{ t_0 }^3 \cdot [ \eta' ( t - t_0 ) ]^2 - \rho^{-1} f_{ t_0 }^2 \cdot \eta'' ( t - t_0 ) \} \Dv_a t \Dv_b t \\
\notag &\qquad + \frac{1}{2} \rho^{-1} f_{ t_0 } \cdot \mi{L}_\rho \gv_{ a b } - \rho^{-2} f_{ t_0 } \cdot \gv_{ a b } - \rho^{-1} f_{ t_0 }^2 \cdot \eta' ( t - t_0 ) \cdot \Dv_{ a b } t \text{.}
\end{align}
\end{itemize}
\end{proposition}

\if\comp1

\begin{proof}
See Appendix \ref{sec.carleman_f}.
\end{proof}

\fi

Finally, for our upcoming main estimates, we will make use of the following local ($g$-)orthonormal frames, which are especially adapted to the level sets of $f_{ t_0 }$:

\begin{proposition} \label{thm.carleman_E}
Let $( \mi{M}, g )$ be an FG-aAdS segment, and let $t$ be a global time.
Given any $p \in \mi{I}$, there exists a neighborhood $U_p$ of $p$ in $\mi{I}$ and vector fields $E_1, \dots, E_{ n - 1 }$ on $( 0, \rho_0 ] \times U_p$ such that:
\begin{itemize}
\item The following identities hold:
\begin{equation}
\label{eq.carleman_E} g ( E_X, E_Y ) = \delta_{ X Y } \text{,} \qquad g ( E_X, \partial_\rho ) = 0 \text{,} \qquad g ( E_X, \Dv^\sharp t ) = 0 \text{,} \qquad 1 \leq X, Y < n \text{.}
\end{equation}

\item $E_1, \dots, E_{ n - 1 }$ are vertical, and there are $\gm$-orthonormal vector fields $\mf{E}_1, \dots, \mf{E}_{ n - 1 }$ on $\mi{I}$ with
\begin{equation}
\label{eq.carleman_E_limit} \rho^{-1} E_X \rightarrow^0 \mf{E}_X \text{,} \qquad 1 \leq X < n \text{.}
\end{equation}
\end{itemize}
\end{proposition}

\if\comp1

\begin{proof}
See Appendix \ref{sec.carleman_E}.
\end{proof}

\fi

\begin{proposition} \label{thm.carleman_frame}
Assume the setting of Definition \ref{def.carleman_f}, and suppose $f_\ast > 0$ is sufficiently small with respect to $t$, $b$, $c$.
Then, the following vector fields are well-defined on $\Omega_{ t_0 } ( f_\ast )$:
\begin{align}
\label{eq.carleman_NV} N_{ t_0 } &:= \rho \{ 1 + f_{ t_0 }^2 [ \eta' ( t - t_0 ) ]^2 \cdot \gv ( \Dv^\sharp t, \Dv^\sharp t ) \}^{ - \frac{1}{2} } [ \partial_\rho - f_{ t_0 } \cdot \eta' ( t - t_0 ) \cdot \Dv^\sharp t ] \text{,} \\
\notag V_{ t_0 } &:= \rho [ - \gv ( \Dv^\sharp t, \Dv^\sharp t ) ]^{ - \frac{1}{2} } \{ 1 + f_{ t_0 }^2 [ \eta' ( t - t_0 ) ]^2 \cdot \gv ( \Dv^\sharp t, \Dv^\sharp t ) \}^{ - \frac{1}{2} } \\
\notag &\qquad \cdot [ f_{ t_0 } \cdot \eta' ( t - t_0 ) \cdot \gv ( \Dv^\sharp t, \Dv^\sharp t ) \cdot \partial_\rho + \Dv^\sharp t ] \text{.}
\end{align}
In addition, if $E_1, \dots, E_{ n - 1 }$ are local vector fields constructed as in Proposition \ref{thm.carleman_E}, then the following properties hold at points where $N_{ t_0 }, V_{ t_0 }, E_1, \dots, E_{ n - 1 }$ are all defined:
\begin{itemize}
\item $( N_{ t_0 }, V_{ t_0 }, E_1, \dots, E_{ n - 1 } )$ defines a local $g$-orthonormal frame.
Moreover, $N_{ t_0 }$ is normal to the level sets of $f_{ t_0 }$, while $V_{ t_0 }, E_1, \dots, E_{ n - 1 }$ are tangent to the level sets of $f_{ t_0 }$.

\item $V_{ t_0 }$ is everywhere $g$-timelike, while $N_{ t_0 }, E_1, \dots, E_{ n - 1 }$ are everywhere $g$-spacelike.
\end{itemize}
\end{proposition}

\if\comp1

\begin{proof}
See Appendix \ref{sec.carleman_frame}.
\end{proof}

\fi

\begin{remark}
Propositions \ref{thm.carleman_E} and \ref{thm.carleman_frame} yield analogues of the frames from \cite[Section 2.4.2]{hol_shao:uc_ads_ns}.
\end{remark}

\subsection{The Main Estimate}

We now give a precise statement of our main Carleman estimate and the associated unique continuation property.
We begin with the Carleman estimate:

\begin{theorem} \label{thm.carleman}
Let $( \mi{M}, g )$ be a strongly FG-aAdS segment, and let $t$ be a global time.
Also:
\begin{itemize}
\item Let $\mi{D} \subseteq \mi{I}$ be open, and suppose $\bar{\mi{D}}$ has compact cross-sections.

\item Assume the null convexity criterion holds on $\mi{D}$, with constants $0 \leq B < C$.

\item Fix constants $B < b < c < C$ and $t_0 \in \R$ such that \eqref{eq.carleman_timespan} holds, and let $\Omega_{ t_0 } := \Omega_{ t_0 } [ b, c ]$ and $f_{ t_0 } := f_{ t_0 } [ b, c ]$ be defined with respect to the above $b, c$.

\item Fix integers $k, l \geq 0$ and fix $\sigma \in \R$, a scalar $\mf{X}^\rho \in C^\infty ( \mi{I} )$, and a vector field $\mf{X}$ on $\mi{I}$.
\end{itemize}
Then, there exist $\mc{C}_b > 0$ and $\mc{C}_o \geq 0$---both depending on $\mi{D}, \gv, t, b, c, B, C, \mf{X}^\rho, \mf{X}, k, l$---such that
\begin{itemize}
\item for any $\kappa \in \R$ satisfying
\begin{equation}
\label{eq.carleman_kappa} 2 \kappa \geq n - 1 + \mc{C}_o \text{,} \qquad \kappa^2 - ( n - 2 ) \kappa + \sigma - ( n - 1 ) - \mc{C}_o \geq 0 \text{,}
\end{equation}

\item for any constant $f_\star$ satisfying
\begin{equation}
\label{eq.carleman_fstar} 0 < f_\star \ll_{ \mi{D}, \gv, t, b, c, B, C, \mf{X}^\rho, \mf{X}, k, l } 1 \text{,}
\end{equation}

\item for any constants $\lambda, p > 0$ satisfying
\begin{equation}
\label{eq.carleman_lambda} \lambda \gg_{ \mi{D}, \gv, t, b, c, B, C, \mf{X}^\rho, \mf{X}, k, l } | \kappa | + | \sigma | \text{,} \qquad 0 < p < \frac{1}{2} \text{,}
\end{equation}

\item and for any rank $( k, l )$ vertical tensor field $\phi$ on $\mi{M}$ such that
\begin{itemize}
\item $\phi$ is supported on $( 0, \rho_0 ] \times \bar{\mi{D}}$, and

\item both $\phi$ and $\nablam \phi$ vanish on $\Omega_{ t_0 } \cap \{ f_{ t_0 } = f_\ast \}$,
\end{itemize}
\end{itemize}
the following Carleman estimate holds:
\begin{align}
\label{eq.carleman} &\int_{ \Omega_{ t_0 } ( f_\ast ) } e^{ - \lambda p^{-1} f_{ t_0 }^p } f_{ t_0 }^{ n - 2 - p - 2 \kappa } | ( \Boxm + \sigma + \rho^2 \nablam_{ \mf{X}^\rho \partial_\rho + \mf{X} } ) \phi |_\hv^2 \, dg \\
\notag &\qquad + \mc{C}_b \lambda^3 \limsup_{ \rho_\ast \searrow 0 } \int_{ \Omega_{ t_0 } ( f_\ast ) \cap \{ \rho = \rho_\ast \} } [ | \Dvm_{ \partial_\rho } ( \rho^{ - \kappa } \phi ) |_\hv^2 + | \Dv_{ \Dv^\sharp t } ( \rho^{ - \kappa } \phi ) |_\hv^2 + | \rho^{ - \kappa - 1 } \phi |_\hv^2 ] \, d \gv |_{ \rho_\ast } \\
\notag &\quad \geq \lambda \int_{ \Omega_{ t_0 } ( f_\ast ) } e^{ - \lambda p^{-1} f_{ t_0 }^p } f_{ t_0 }^{ n - 2 - 2 \kappa } ( f_{ t_0 } \rho^3 | \Dvm_{ \partial_\rho } \phi |_\hv^2 + f_{ t_0 } \rho^3 | \Dv \phi |_\hv^2 + f^{ 2 p } | \phi |_\hv^2 ) \, dg \text{.}
\end{align}
Furthermore, the above holds with $\mc{C}_o = 0$ in \eqref{eq.carleman_kappa} when $( k + l ) \Dm^2 t$, $\mf{X}^\rho$, and $\mf{X}$ all vanish on $\bar{\mi{D}}$.
\end{theorem}

\begin{remark}
The integrals over $\Omega_{ t_0 } ( f_\ast )$ in \eqref{eq.carleman} are with respect to the spacetime metric $g$, while the boundary integral over $\Omega_{ t_0 } ( f_\ast ) \cap \{ \rho = \rho_\ast \}$ is with respect to the vertical metric $\gv$.
\end{remark}

\begin{remark}
The need for a sufficiently large timespan in \eqref{eq.carleman} is implicitly captured by the domain $\Omega_{ t_0 } ( f_\ast )$ of integration, which at $\rho \searrow 0$ covers a timespan of $\mc{T}_+ ( b, c )$; see Definition \ref{def.carleman_f}.
\end{remark}

The proof of Theorem \ref{thm.carleman} is given in Section \ref{sec.carleman_proof}.
In addition, the following corollary is the corresponding unique continuation result that follows from Theorem \ref{thm.carleman}:

\begin{corollary} \label{thm.carleman_uc}
Let $( \mi{M}, g )$ be a strongly FG-aAdS segment, and let $t$ be a global time.
Also:
\begin{itemize}
\item Let $\mi{D} \subseteq \mi{I}$ be open, and suppose $\bar{\mi{D}}$ has compact cross-sections.

\item Assume the null convexity criterion holds on $\mi{D}$, with constants $0 \leq B < C$.

\item Fix constants $B < b < c < C$ and $t_0 \in \R$ such that \eqref{eq.carleman_timespan} holds, and let $\Omega_{ t_0 } := \Omega_{ t_0 } [ b, c ]$ and $f_{ t_0 } := f_{ t_0 } [ b, c ]$ be defined with respect to the above $b, c$.

\item Fix integers $k, l \geq 0$ and fix $\sigma \in \R$, a scalar $\mf{X}^\rho \in C^\infty ( \mi{I} )$, and a vector field $\mf{X}$ on $\mi{I}$.
\end{itemize}
Then, there exists $\mc{C}_o \geq 0$---depending on $\mi{D}, \gv, t, b, c, B, C, \mf{X}^\rho, \mf{X}, k, l$---such that the following unique continuation property holds: if $\phi$ is a rank $( k, l )$ vertical tensor field on $\mi{M}$ such that
\begin{itemize}
\item there exists $p > 0$ and $\mc{C} > 0$ such that $\phi$ satisfies
\begin{equation}
\label{eq.carleman_uc_wave} | ( \Boxm + \sigma + \rho^2 \nablam_{ \mf{X}^\rho \partial_\rho + \mf{X} } ) \phi |_\hv^2 \leq \mc{C} ( \rho^{ 4 + p } | \Dvm_\rho \phi |_\hv^2 + \rho^{ 4 + p } | \Dv \phi |_\hv^2 + \rho^{ 3p } | \phi |_\hv^2 ) \text{,}
\end{equation}

\item $\phi$ is supported on $( 0, \rho_0 ] \times \bar{\mi{D}}$, and

\item there exists $\kappa \in \R$, satisfying \eqref{eq.carleman_kappa}, such that $\phi$ satisfies the vanishing condition
\begin{equation}
\label{eq.carleman_uc_vanish} \lim_{ \rho' \searrow 0 } \int_{ \left\{ \rho = \rho' \text{, } |t| < \frac{1}{2} \mc{T}_+ ( b, c ) \right\} } [ | \Dvm_{ \partial_\rho } ( \rho^{ - \kappa } \phi ) |_\hv^2 + | \Dv ( \rho^{ - \kappa } \phi ) |_\hv^2 + | \rho^{ - \kappa - 1 } \phi |_\hv^2 ] d \gv |_{ \rho' } = 0 \text{,}
\end{equation}
\end{itemize}
then there is some $f_\ast > 0$---depending on $\mi{D}, \gv, t, b, c, B, C, \mf{X}^\rho, \mf{X}, k, l$---such that $\phi \equiv 0$ on $\Omega_{ t_0 } ( f_\ast )$.
Furthermore, the above holds with $\mc{C}_o = 0$ in \eqref{eq.carleman_kappa} when $( k + l ) \Dm^2 t$, $\mf{X}^\rho$, and $\mf{X}$ all vanish on $\bar{\mi{D}}$.
\end{corollary}

We omit the proof of Corollary \ref{thm.carleman_uc}, which applies Theorem \ref{thm.carleman} via a standard process.
For further details, see the proof of \cite[Theorem 3.11]{hol_shao:uc_ads_ns}, which uses an analogous argument.

\begin{remark}
Again, if we assume in Theorem \ref{thm.carleman} and Corollary \ref{thm.carleman_uc} that $\mi{I}$ itself has compact cross-sections, then we can take $\mi{D} := \mi{I}$ and obtain a global version of these results.
\end{remark}

\begin{remark}
In particular, the quantity $\kappa$ in \eqref{eq.carleman_kappa} represents the order of vanishing required of the solution $\phi$ in order for the unique continuation result of Corollary \ref{thm.carleman_uc} to hold.
Notice that the range of admissible $\kappa$ depends on the boundary dimension $n$; the Klein--Gordon mass $\sigma$; the ``leading-order" asymptotics $\mf{X}^\rho$, $\mf{X}$ of the first-order coefficients; the tensor rank $( k, l )$ of the solution $\phi$; and the ``non-stationarity" $\Dm^2 t$ of the boundary metric.
\footnote{If $\phi$ is decomposed into scalar quantities, then the corresponding scalar wave equations would contain additional ``leading-order" first-order terms arising from $\mf{D}^2 t$.
Thus, we can view $\Dm^2 t$ as playing the same role as $\mf{X}^\rho$ and $\mf{X}$.}
While the dependence on $n$ and $\sigma$ are necessary, it is not known whether the dependence on the other quantities can be removed.
\end{remark}

\begin{remark}
Note that if $\mf{X}^\rho$ and $\mf{X}$ vanish, and if either $\phi$ is scalar or $\Dm^2 t$ vanishes (i.e.~$\gm$ is stationary), then Corollary \ref{thm.carleman_uc} holds with the same optimal vanishing rate as in \cite[Section 3.3]{hol_shao:uc_ads_ns}.
In particular, if we take $\mc{C}_o = 0$ in \eqref{eq.carleman_kappa}, then the smallest $\kappa$ satisfying \eqref{eq.carleman_kappa} is
\begin{equation}
\kappa = \begin{cases} \frac{n - 2}{2} + \sqrt{ \frac{ n^2 }{4} - \sigma } & \sigma \leq \frac{ n^2 - 1 }{4} \text{,} \\ \frac{n - 1}{2} & \sigma \geq \frac{ n^2 - 1 }{4} \text{.} \end{cases}
\end{equation}
\end{remark}

\begin{remark}
The vanishing condition \eqref{eq.carleman_uc_vanish} can be relaxed slightly to
\[
\lim_{ \rho' \searrow 0 } \int_{ \left\{ \rho = \rho' \text{, } |t| < \frac{1}{2} \mc{T}_+ ( b, c ) \right\} } [ | \Dvm_{ \partial_\rho } ( \rho^{ - \kappa } \phi ) |_\hv^2 + | \Dv_{ \Dv^\sharp t } ( \rho^{ - \kappa } \phi ) |_\hv^2 + | \rho^{ - \kappa - 1 } \phi |_\hv^2 ] d \gv |_{ \rho' } = 0 \text{,}
\]
and a spacetime integrability assumption for $\Dv \phi$; see \cite[Definition 3.10, Theorem 3.11]{hol_shao:uc_ads_ns} for details.
\end{remark}

\subsection{Proof of Theorem \ref{thm.carleman}} \label{sec.carleman_proof}

Throughout this subsection, we assume the hypotheses of Theorem \ref{thm.carleman}.
Moreover, by replacing $t$ by $t - t_0$, we can assume, without loss of generality, that $t_0 = 0$; in particular, we can replace $f_{ t_0 }$ everywhere by $f := f_0$, as well as $\Omega_{ t_0 }$ by $\Omega_0$.

Furthermore, throughout this proof:
\begin{itemize}
\item We assume that all tensor fields (spacetime, vertical, or boundary) are indexed are with respect to $\varphi$- and $\varphi_\rho$-coordinates, for arbitrary coordinate systems $( U, \varphi )$ on $\mi{I}$.

\item We use the symbols $N := N_0, V := V_0, E_1, \dots, E_{ n - 1 }$ to refer to any $g$-orthonormal frame that is constructed as in Propositions \ref{thm.carleman_E} and \ref{thm.carleman_frame}.

\item To make notations more concise, we define the regions
\footnote{In particular, all the quantities we consider will be smooth on both $\Omega^c_<$ and $\Omega^c_>$.}
\begin{equation}
\label{eql.carleman_Omega} \Omega^c := \Omega_0 ( f_\ast ) \cap ( ( 0, \rho_0 ] \times \bar{\mi{D}} ) \text{,} \qquad \Omega^c_< := \Omega^c \cap \{ t < 0 \} \text{,} \qquad \Omega^c_> := \Omega^c \cap \{ t > 0 \} \text{.}
\end{equation}
\end{itemize}
For convenience, we also define the shorthand
\begin{equation}
\label{eql.carleman_Upsilon} \Upsilon := \{ \mi{D}, \gv, t, b, c, B, C, \mf{X}^\rho, \mf{X}, k, l \} \text{,}
\end{equation}
which contains the objects on which our constants depend.
In addition, to streamline error terms in computations, we write $\mc{O} ( \xi )$ to denote any (spacetime) scalar function $\omega$ satisfying
\begin{equation}
\label{eql.carleman_OOO} | \omega | \lesssim_\Upsilon \xi \text{.}
\end{equation}

\begin{remark} \label{rmkl.carleman_bounded}
Observe that $\Omega^c$ is entirely contained within $( 0, \rho_0 ] \times \mi{D}_c$, where
\begin{equation}
\label{eql.carleman_bounded} \mi{D}_c := \left\{ P \in \bar{\mi{D}} \mid | t ( P ) | \leq \frac{1}{2} \mc{T}_+ ( b, c ) \right\} \text{.}
\end{equation}
Furthermore, since $\bar{\mi{D}}$ has compact cross-sections, then \eqref{eq.carleman_timespan} implies $\mi{D}_c$ is compact.
As a result, by Definition \ref{def.aads_strong}, Proposition \ref{thm.aads_geom_limit}, and the above, we conclude that the geometric quantities $\gv$, $\hv$, and $\Dvm t$ are uniformly bounded up to three derivatives on $\Omega^c$.
\end{remark}

\subsubsection{Pseudoconvexity}

The first key step of the proof is to show, using the null convexity criterion, that the level sets of the function $f$ are pseudoconvex.

Fix $0 < \delta \ll 1$, whose exact value will be determined later, and let $\chi \in C^\infty ( \R )$ satisfy
\[
0 \leq \chi \leq 1 \text{,} \qquad \chi ( \tau ) = \begin{cases} 1 & | \tau | \geq 2 \delta \text{,} \\ 0 & | \tau | \leq \delta \text{.} \end{cases} 
\]
Since the null convexity criterion holds on $\mi{D}$, Theorem \ref{thm.psc_nc} implies there exist $\zeta_0, \zeta_\pm \in C^\infty ( \mi{I} )$ such that \eqref{eq.psc_nc_g} and \eqref{eq.psc_nc_t} hold on $\mi{D}$, with $b$ and $c$ as in the theorem statement.
We now set
\begin{equation}
\label{eql.carleman_zeta} \zeta: \mi{I} \rightarrow \R \text{,} \qquad \zeta = \begin{cases} \eta (t) \zeta_0 + \chi (t) \eta' (t) \zeta_- & t < 0 \text{,} \\ \eta (t) \zeta_0 - \chi (t) \eta' (t) \zeta_+ & t > 0 \text{.} \end{cases}
\end{equation}
From Proposition \ref{thm.carleman_eta}, it follows that $\zeta \in C^2 ( \mi{I} )$.

Furthermore, we define $w_\zeta: \Omega_0 \rightarrow \R$ and the (spacetime, symmetric) tensor field $\pi_\zeta$ by
\begin{equation}
\label{eql.carleman_pi} \qquad w_\zeta := f + f^2 \rho \zeta \text{,} \qquad \pi_\zeta := - [ \nabla ( f^{ n - 3 } \nabla f ) + f^{ n - 3 } w_\zeta \cdot g ] \text{,}
\end{equation}
In particular, the pseudoconvexity of the level sets of $f$ will be captured by the positivity properties of $\pi_\zeta$.
The following lemma provides asymptotics for all the components of $\pi_\zeta$:

\begin{lemma} \label{thm.carleman_pseudoconvex_pi}
Let $\mf{T}$ denote the $\gm$-unit (boundary) vector field
\begin{equation}
\label{eql.carleman_pseudoconvex_T} \mf{T} := | \mf{g} ( \Dm^\sharp t, \Dm^\sharp t ) |^{ - \frac{1}{2} } \Dm^\sharp t \text{.}
\end{equation}
and let $\mf{E}_1, \dots, \mf{E}_{ n - 1 }$ denote the limits as $\rho \searrow 0$ of $\rho^{-1} E_1, \dots, \rho^{-1} E_{ n - 1 }$, respectively.
Then, the following asymptotic relations hold on $\Omega^c_< \cup \Omega^c_>$, for any $1 \leq X, Y < n$:
\begin{align}
\label{eq.carleman_pseudoconvex_pi} \pi_\zeta ( V, V ) &= \rho f^{ n - 1 } [ \eta'' (t) \cdot dt^2 + \eta' (t) \cdot \Dm^2 t - \eta \cdot \gs - \zeta \cdot \gm ] ( \mf{T}, \mf{T} ) + \mc{O} ( \rho f^n ) \text{,} \\
\notag \pi_\zeta ( V, E_X ) &= \rho f^{ n - 1 } [ \eta'' (t) \cdot dt^2 + \eta' (t) \cdot \Dm^2 t - \eta \cdot \gs - \zeta \cdot \gm ] ( \mf{T}, \mf{E}_X ) + \mc{O} ( \rho f^n ) \text{,} \\
\notag \pi_\zeta ( E_X, E_Y ) &= \rho f^{ n - 1 } [ \eta'' (t) \cdot dt^2 + \eta' (t) \cdot \Dm^2 t - \eta \cdot \gs - \zeta \cdot \gm ] ( \mf{E}_X, \mf{E}_Y ) + \mc{O} ( \rho f^n ) \text{,} \\
\notag \pi_\zeta ( N, N ) &= - ( n - 1 ) f^{ n - 2 } + \mc{O} ( f^n ) \text{,} \\
\notag \pi_\zeta ( N, V ) &= \mc{O} ( \rho f^n ) \text{,} \\
\notag \pi_\zeta ( N, E_X ) &= \mc{O} ( \rho f^n ) \text{,}
\end{align}
\end{lemma}

\if\comp1

\begin{proof}
See Appendix \ref{sec.carleman_pseudoconvex_pi}.
\end{proof}

\fi

The main pseudoconvexity property for $f$ is given in the subsequent lemma:

\begin{lemma} \label{thm.carleman_pseudoconvex}
There exists a constant $K \geq 0$, depending on $\Upsilon$, such that the following inequality holds on $\Omega^c_< \cup \Omega^c_>$ for any spacetime vector field $Z$:
\begin{align}
\label{eq.carleman_pseudoconvex} \pi_\zeta ( Z, Z ) &\geq K f^{ n - 1 } \rho \left[ | g ( Z, V ) |^2 + \sum_{ X = 1 }^{ n - 1 } | g ( Z, E_X ) |^2 \right] \\
\notag &\qquad - [ ( n - 1 ) f^{ n - 2 } + \mc{O} ( f^n ) ] \cdot | g ( Z, N ) |^2 \text{.}
\end{align}
\end{lemma}

\begin{proof}
Let $\mf{T}, \mf{E}_1, \dots, \mf{E}_{ n - 1 }$ be as in the statement of Lemma \ref{thm.carleman_pseudoconvex_pi}.
Note from Propositions \ref{thm.carleman_E} and \ref{thm.carleman_frame} that $\mf{T}, \mf{E}_1, \dots, \mf{E}_{ n - 1 }$ is a $\gm$-orthonormal frame.
For conciseness, we also define
\begin{equation}
\label{eql.carleman_pseudoconvex_1} Z_N := g ( Z, N ) \text{,} \qquad Z_V := - g ( Z, V ) \text{,} \qquad Z_X := g ( Z, E_X ) \text{,} \qquad 1 \leq X < n \text{,}
\end{equation}
as well as the following (local) vertical vector field:
\begin{equation}
\label{eql.carleman_pseudoconvex_3} \ms{Z} := Z_V \cdot \mf{T} + \sum_{ X = 1 }^{ n - 1 } Z_X \cdot \mf{E}_X \text{.}
\end{equation}

Now, from \eqref{eq.carleman_pseudoconvex_pi}, \eqref{eql.carleman_pseudoconvex_1}, and \eqref{eql.carleman_pseudoconvex_3}, we obtain, on $\Omega^c_< \cup \Omega^c_>$, the relation
\begin{align}
\label{eql.carleman_pseudoconvex_20} \pi_\zeta ( Z, Z ) &= Z_V^2 \pi_\zeta ( V, V ) + 2 \sum_{ X = 1 }^{ n - 1 } Z_V Z_X \pi_\zeta ( V, E_X ) + \sum_{ X, Y = 1 }^{ n - 1 } Z_X Z_Y \pi_\zeta ( E_X, E_Y ) \\
\notag &\qquad + Z_N^2 \pi_\zeta ( N, N ) + 2 Z_N Z_V \pi_\zeta ( N, V ) + 2 \sum_{ X = 1 }^{ n - 1 } Z_N Z_X \pi_\zeta ( N, E_X ) \\
\notag &= \rho f^{ n - 1 } [ \eta'' (t) \cdot dt^2 + \eta' (t) \cdot \mf{D}^2 t - \eta (t) \cdot \gs - \zeta \cdot \gm ] ( \ms{Z}, \ms{Z} ) \\
\notag &\qquad - ( n - 1 ) f^{ n - 2 } \cdot Z_N^2 + \mc{E} \text{,}
\end{align}
where the error terms $\mc{E}$ satisfy
\begin{align}
\label{eql.carleman_pseudoconvex_21} \mc{E} &= Z_V^2 \cdot \mc{O} ( \rho f^n ) + \sum_{ X = 1 }^{ n - 1 } Z_V Z_X \cdot \mc{O} ( \rho f^n ) + \sum_{ X, Y = 1 }^{ n - 1 } Z_X Z_Y \cdot \mc{O} ( \rho f^n ) \\
\notag &\qquad + Z_N^2 \cdot \mc{O} ( f^n ) + Z_N Z_V \cdot \mc{O} ( \rho f^n ) + \sum_{ X = 1 }^{ n - 1 } Z_N Z_X \cdot \mc{O} ( \rho f^n ) \\
\notag &\geq Z_V^2 \cdot \mc{O} ( \rho f^n ) + \sum_{ X = 1 }^{ n - 1 } Z_X^2 \cdot \mc{O} ( \rho f^n ) + Z_N^2 \cdot \mc{O} ( f^n ) \text{.}
\end{align}

First, whenever $| t | \geq 2 \delta$, we apply Theorem \ref{thm.psc_nc}, \eqref{eq.carleman_eta_minus}, \eqref{eql.carleman_zeta}, and \eqref{eql.carleman_pseudoconvex_20} in order to obtain
\footnote{Note that $\eta' (t) > 0$ whenever $t < 0$, and that $\eta' (t) < 0$ whenever $t > 0$.}
\begin{align}
\label{eql.carleman_pseudoconvex_30} \pi_\zeta ( Z, Z ) &= - ( n - 1 ) f^{ n - 2 } \cdot Z_N^2 + \rho f^{ n - 1 } \eta (t) [ - \gs - c^2 \cdot dt^2 - \zeta_0 \cdot \gm ] ( \ms{Z}, \ms{Z} ) + \mc{E} \\
\notag &\qquad + \begin{cases} \rho f^{ n - 1 } | \eta' (t) | [ \mf{D}^2 t + 2 b \cdot dt^2 - \zeta_- \cdot \gm ] ( \ms{Z}, \ms{Z} ) & t < 0 \\ \rho f^{ n - 1 } | \eta' (t) | [ - \mf{D}^2 t + 2 b \cdot dt^2 - \zeta_+ \cdot \gm ] ( \ms{Z}, \ms{Z} ) & t > 0 \end{cases} \\
\notag &\geq - ( n - 1 ) f^{ n - 2 } \cdot Z_N^2 + K' \rho f^{ n - 1 } [ \eta (t) + | \eta' (t) | ] \hm ( \ms{Z}, \ms{Z} ) + \mc{E} \text{,}
\end{align}
where $K'$ depends on $\Upsilon$.
Next, a similar analysis when $|t| < 2 \delta$ yields
\begin{align}
\label{eql.carleman_pseudoconvex_31} \pi_\zeta ( Z, Z ) &\geq - ( n - 1 ) f^{ n - 2 } \cdot Z_N^2 + \rho f^{ n - 1 } \eta (t) [ - \gs - c^2 \cdot dt^2 - \zeta_0 \cdot \gm ] ( \ms{Z}, \ms{Z} ) + \mc{E} \\
\notag &\qquad + \begin{cases} \rho f^{ n - 1 } | \eta' (t) | [ \mf{D}^2 t + 2 b \cdot dt^2 - \chi (t) \zeta_- \cdot \gm ] ( \ms{Z}, \ms{Z} ) & t < 0 \\ \rho f^{ n - 1 } | \eta' (t) | [ - \mf{D}^2 t + 2 b \cdot dt^2 - \chi (t) \zeta_+ \cdot \gm ] ( \ms{Z}, \ms{Z} ) & t > 0 \end{cases} \\
\notag &\geq - ( n - 1 ) f^{ n - 2 } \cdot Z_N^2 + \rho f^{ n - 1 } [ K'' \eta (t) - C'' | \eta' (t) | ] \hm ( \ms{Z}, \ms{Z} ) + \mc{E} \text{,}
\end{align}
where $K''$ and $C''$ again depend on $\Upsilon$.

Observe \eqref{eq.carleman_eta} implies that everywhere on $\Omega_0$,
\begin{equation}
\label{eql.carleman_pseudoconvex_32} \eta (t) + | \eta' (t) | \simeq 1 \text{.}
\end{equation}
Moreover, if $\delta$ is sufficiently small (with respect to $\Upsilon$), then Definition \ref{def.carleman_eta} yields
\[
\eta (t) \simeq 1 \text{,} \qquad | \eta' (t) | \ll_\Upsilon 1 \text{,} \qquad |t| < 2 \delta \text{.}
\]
Combining \eqref{eql.carleman_pseudoconvex_21}--\eqref{eql.carleman_pseudoconvex_32} and the above, we conclude there exists $K > 0$, depending on $\Upsilon$, with
\begin{align}
\label{eql.carleman_pseudoconvex_40} \pi_\zeta ( Z, Z ) &\geq K \rho f^{ n - 1 } \hm ( \ms{Z}, \ms{Z} ) - ( n - 1 ) f^{ n - 2 } \cdot Z_N^2 + Z_N^2 \cdot \mc{O} ( f^n ) \\
\notag &\qquad + Z_V^2 \cdot \mc{O} ( \rho f^n ) + \sum_{ X = 1 }^{ n - 1 } Z_X^2 \cdot \mc{O} ( \rho f^n ) \text{.}
\end{align}

Finally, note that \eqref{eq.aads_riemann} and \eqref{eql.carleman_pseudoconvex_3} together imply
\[
\hm ( \ms{Z}, \ms{Z} ) = Z_V^2 + \sum_{ X = 1 }^{ n - 1 } Z_X^2 \text{.}
\]
The desired result \eqref{eq.carleman_pseudoconvex} now follows from \eqref{eq.carleman_fstar}, \eqref{eql.carleman_pseudoconvex_1}, \eqref{eql.carleman_pseudoconvex_40}, and the above.
\end{proof}

\subsubsection{Preliminary Estimates}

For convenience, we define the following quantities:
\begin{equation}
\label{eql.carleman_S} S := f^{ n - 3 } \nabla^\sharp f \text{,} \qquad v_\zeta := f^{ n - 3 } w_\zeta + \frac{1}{2} \nabla_\alpha S^\alpha \text{.}
\end{equation}
Moreover, we define the Carleman weight $e^{-F}$ and the auxiliary unknown $\psi$ by
\begin{equation}
\label{eql.carleman_weight} F := \kappa \, \log f + \lambda p^{-1} f^p \text{,} \qquad \psi := e^{-F} \phi \text{.}
\end{equation}
The next step is to collect some asymptotic bounds that will be useful later on.

\begin{lemma} \label{thm.carleman_asymp}
The following asymptotic relations hold on $\Omega^c_< \cup \Omega^c_>$:
\begin{align}
\label{eq.carleman_asymp} g^{ \alpha \beta } \nabla_\alpha f \nabla_\beta f = f^2 + \mc{O} ( f^4 ) \text{,} &\qquad \Box f = - ( n - 1 ) f + \mc{O} ( f^3 ) \text{,} \\
\notag \nabla^\sharp f = [ f + \mc{O} ( f^3 ) ] N \text{,} &\qquad S = [ f^{ n - 2 } + \mc{O} ( f^n ) ] N \text{.}
\end{align}
Moreover, the following expansions---given with respect to a fixed finite family of compact coordinate systems that cover $\mi{D}^c$ (see Remark \ref{rmkl.carleman_bounded})---hold on $\Omega^c_< \cup \Omega^c_>$ for each $1 \leq X < n$:
\begin{equation}
\label{eq.carleman_frame_asymp} N = \mc{O} ( \rho ) \, \partial_\rho + \sum_a \mc{O} ( \rho f ) \, \partial_a \text{,} \qquad V = \mc{O} ( \rho f ) \, \partial_\rho + \sum_a \mc{O} ( \rho ) \, \partial_a \text{,} \qquad E_X = \sum_a \mc{O} ( \rho ) \, \partial_a \text{.}
\end{equation}
\end{lemma}

\if\comp1

\begin{proof}
See Appendix \ref{sec.carleman_asymp}.
\end{proof}

\fi

\begin{lemma} \label{thm.carleman_error_v}
The following estimates hold on $\Omega^c_< \cup \Omega^c_>$:
\begin{align}
\label{eq.carleman_error_v} | v_\zeta | = \mc{O} ( f^n ) \text{,} &\qquad | \nabla_\rho v_\zeta | = \mc{O} ( \rho^{-1} f^n ) \text{,} \\
\notag | \Dv v_\zeta |_\hv = \mc{O} ( \rho^{-1} f^{ n + 1 } ) \text{,} &\qquad | \Box v_\zeta | = \mc{O} ( f^n ) \text{.}
\end{align}
\end{lemma}

\if\comp1

\begin{proof}
See Appendix \ref{sec.carleman_error_v}.
\end{proof}

\fi

We will also need asymptotic estimates for the mixed curvature $\bar{R}$ and for the metric $\hv$:

\begin{lemma} \label{thm.carleman_error_R}
The following hold on $\Omega^c_< \cup \Omega^c_>$ for any rank $( k, l )$ vertical tensor field $\ms{A}$:
\begin{equation}
\label{eq.carleman_error_R} | \bar{R}_{ N V } [ \ms{A} ] |_\hv \leq ( k + l ) \mc{O} ( \rho^2 f ) \, | \ms{A} |_\hv \text{,} \qquad | \bar{R}_{ N E_X } [ \ms{A} ] | \leq ( k + l ) \mc{O} ( \rho^2 f ) \, | \ms{A} |_\hv \text{,} \qquad 1 \leq X < n \text{.}
\end{equation}
\end{lemma}

\if\comp1

\begin{proof}
See Appendix \ref{sec.carleman_error_R}.
\end{proof}

\fi

\begin{lemma} \label{thm.carleman_error_h}
The following hold on $\Omega^c_< \cup \Omega^c_>$ for any $1 \leq X < n$:
\begin{align}
\label{eq.carleman_error_h} | \nablam_\rho \hv |_\hv = \mc{O} ( \rho ) \text{,} &\qquad | \nablam_N \hv |_\hv = \mc{O} ( \rho f ) \\
\notag | \nablam_V \hv |_\hv + | \nablam_{ E_X } \hv |_\hv \leq \mc{O} ( \rho ) \, | \Dm^2 t |_\hv + \mc{O} ( \rho f ) \text{,} &\qquad | \Boxm \hv |_\hv = \mc{O} ( \rho^2 ) \text{.}
\end{align}
\end{lemma}

\if\comp1

\begin{proof}
See Appendix \ref{sec.carleman_error_h}.
\end{proof}

\fi

Let us now define the following operators:
\begin{equation}
\label{eql.carleman_L} \mc{L} := e^{-F} ( \Boxm + \sigma ) e^F \text{,} \qquad \bar{S}_\zeta := \bar{\nabla}_S + v_\zeta \text{.}
\end{equation}
A direct computation using \eqref{eql.carleman_S}, \eqref{eql.carleman_weight}, and \eqref{eql.carleman_L} then yields
\begin{equation}
\label{eql.carleman_LA} \mc{L} = \Boxm + 2 F' f^{ -n + 3 } \nablam_S + \mc{A} \text{,} \qquad \mc{A} := [ ( F' )^2 + F'' ] g^{ \alpha \beta } \nabla_\alpha f \nabla_\beta f + F' \Box f + \sigma \text{.}
\end{equation}

\begin{lemma} \label{thm.carleman_error_A}
The following relations hold on $\Omega^c_< \cup \Omega^c_>$:
\begin{align}
\label{eq.carleman_error_A} \mc{A} &= ( \kappa^2 - n \kappa + \sigma ) + ( 2 \kappa - n + p ) \lambda f^p + \lambda^2 f^{ 2 p } + \lambda^2 \, \mc{O} ( f^2 ) \text{,} \\
\notag - \frac{1}{2} \nabla_\beta ( \mc{A} S^\beta ) &= ( \kappa^2 - n \kappa + \sigma ) f^{ n - 2 } + \left( 1 - \frac{p}{2} \right) ( 2 \kappa - n + p ) \lambda f^{ n - 2 + p } \\
\notag &\qquad + ( 1 - p ) \lambda^2 f^{ n - 2 + 2 p } + \lambda^2 \, \mc{O} ( f^n ) \text{.}
\end{align}
\end{lemma}

\if\comp1

\begin{proof}
See Appendix \ref{sec.carleman_error_A}.
\end{proof}

\fi

Finally, recalling $\mf{X}$ and $\mf{X}^\rho$ from the statement of Theorem \ref{thm.carleman}, we then define
\begin{equation}
\label{eql.carleman_XX} \mc{Y} := \rho^2 ( \mf{X}^\rho \partial_\rho + \mf{X} ) \text{,} \qquad \mc{L}^\dagger := e^{ -F } ( \bar{\Box} + \sigma + \nablam_{ \mc{Y} } ) e^F \text{,}
\end{equation}
where $\mc{Y}$ is viewed as a spacetime vector field.

\begin{lemma} \label{thm.carleman_pointwise_J}
Consider the quantity
\footnote{Recall that the definitions of $\hvd$ and $\langle \cdot, \cdot \rangle$ were given in Definition \ref{def.aads_tensor_norm}.}
\begin{equation}
\label{eq.carleman_J} J := \frac{1}{2} \langle \mc{L} ( \hvd \psi ), \bar{S}_\zeta \psi \rangle + \frac{1}{2} \langle \mc{L} \psi, \bar{S}_\zeta ( \hvd \psi ) \rangle \text{.}
\end{equation}
Then, the following pointwise inequality holds everywhere on $\Omega^c_< \cup \Omega^c_>$,
\begin{align}
\label{eq.carleman_pointwise_J} | J | &\leq \lambda^{-1} f^{ n - 2 - p } | \mc{L}^\dagger \psi |_\hv^2 + \mc{C}_o f^{ n - 2 } | \nablam_N \psi |_\hv^2 + \frac{1}{2} \lambda f^{ n - 2 + p } | \nablam_N \psi |_\hv^2 \\
\notag &\qquad + \frac{1}{4} K \rho f^{ n - 1 } \left( | \nablam_V \psi |_\hv^2 + \sum_{ X = 1 }^{ n - 1 } | \nablam_{ E_X } \psi |_\hv^2 \right) + \lambda \, \mc{O} ( f^{ n - 1 } ) \, ( | \nablam_N \psi |_\hv^2 + | \psi |_\hv^2 ) \text{,}
\end{align}
where $K$ is as in the statement of Lemma \ref{thm.carleman_pseudoconvex}, and where the constant $\mc{C}_o \geq 0$ depends on $\Upsilon$.

Furthermore, if $( k + l ) \mf{D}^2 t$, $\mf{X}^\rho$, and $\mf{X}$ all vanish, then \eqref{eq.carleman_pointwise_J} holds with $\mc{C}_o = 0$.
\end{lemma}

\begin{proof}
First, notice that from \eqref{eql.carleman_L} and \eqref{eql.carleman_XX}, we have
\[
\label{eql.carleman_pointwise_J_0} \mc{L} \psi = \mc{L}^\dagger \psi - \nablam_{ \mc{Y} } \psi - F' \mc{Y} f \cdot \psi \text{.}
\]
Applying \eqref{eql.carleman_LA}, \eqref{eq.carleman_J}, and the above, we can then expand
\begin{align}
\label{eql.carleman_pointwise_J_1} J &= \hvm ( \mc{L}^\dagger \psi, \bar{S}_\zeta \psi ) - \hvm ( \nablam_{ \mc{Y} } \psi, \bar{S}_\zeta \psi ) - F' \mc{Y} f \cdot \hvm ( \psi, \bar{S}_\zeta \psi ) \\
\notag &\qquad + g^{ \alpha \beta } \nablam_\alpha \hvm ( \nablam_\beta \psi, \bar{S}_\zeta \psi ) + \frac{1}{2} \Boxm \hvm ( \psi, \bar{S}_\zeta \psi ) + F' f^{ -n + 3 } \nablam_S \hvm ( \psi, \bar{S}_\zeta \psi ) \\
\notag &\qquad + \frac{1}{2} \nablam_S \hvm ( \mc{L}^\dagger \psi, \psi ) - \frac{1}{2} \nablam_S \hvm ( \nablam_{ \mc{Y} } \psi, \psi ) - \frac{1}{2} F' \mc{Y} f \cdot \nabla_S \hvm ( \psi, \psi ) \\
\notag &= J_1 + \dots + J_9 \text{.}
\end{align}

Recalling \eqref{eq.carleman_lambda}, \eqref{eql.carleman_weight}, \eqref{eq.carleman_asymp}, \eqref{eq.carleman_error_v}, and \eqref{eq.carleman_error_h}, we estimate
\begin{align}
\label{eql.carleman_pointwise_J_2} | J_6 | &\leq \lambda \, \mc{O} (1) \, | \nablam_N \hvm |_\hv | \psi |_\hv ( | \nablam_S \psi |_\hv + | v_\zeta | | \psi |_\hv ) \\
\notag &\leq \lambda \, \mc{O} ( \rho f^{ n - 1 } ) | \nablam_N \psi |_\hv | \psi |_\hv + \lambda \, \mc{O} ( \rho f^{ n + 1 } ) | \psi |_\hv^2 \\
\notag &\leq \lambda \, \mc{O} ( \rho f^{ n - 1 } ) \, ( | \nablam_N \psi |_\hv^2 + | \psi |_\hv^2 ) \text{,}
\end{align}
In addition, \eqref{eq.carleman_asymp}, \eqref{eq.carleman_error_v}, and \eqref{eq.carleman_error_h} yield
\begin{align}
\label{eql.carleman_pointwise_J_3} | J_5 | &\leq | \Boxm \hvm |_\hv | \psi |_\hv [ \mc{O} ( f^{ n - 2 } ) \, | \nablam_N \psi |_\hv + \mc{O} ( f^n ) \, | \psi |_\hv ] \\
\notag &\leq \mc{O} ( \rho^2 f^{ n - 2 } ) \, ( | \nablam_N \psi |_\hv^2 + | \psi |_\hv^2 ) \text{.}
\end{align}
For $J_4$, we expand using orthonormal frames and apply \eqref{eq.carleman_lambda}, \eqref{eq.carleman_asymp}, \eqref{eq.carleman_error_v}, and \eqref{eq.carleman_error_h}:
\footnote{The extra factor $k + l$ in the right-hand side arises from the observation that $\hvm$ contains $k + l$ copies of $\hv$ and $\hv^{-1}$.
In particular, when estimating the derivative of $\hvm$, we apply \eqref{eq.carleman_error_h} a total of $k + l$ times.}
\begin{align}
\label{eql.carleman_pointwise_J_4} | J_4 | &\leq ( k + l ) | \Dm^2 t |_\hv \, \mc{O} ( \rho ) \, \left( | \nablam_V \psi |_\hv + \sum_{ X = 1 }^{ n - 1 } | \nablam_{ E_X } \psi |_\hv \right) | \bar{S}_\zeta \psi |_\hv \\
\notag &\qquad + \mc{O} ( \rho f ) \, \left( | \nablam_N \psi |_\hv + | \nablam_V \psi |_\hv + \sum_{ X = 1 }^{ n - 1 } | \nablam_{ E_X } \psi |_\hv \right) | \bar{S}_\zeta \psi |_\hv \\
\notag &\leq ( k + l ) | \Dm^2 t |_\hv \, \mc{O} ( \rho f^{ n - 2 } ) \, \left( | \nablam_V \psi |_\hv + \sum_{ X = 1 }^{ n - 1 } | \nablam_{ E_X } \psi |_\hv \right) | \nablam_N \psi |_\hv \\
\notag &\qquad + \mc{O} ( \rho f^{ n - 1 } ) \, \left( | \nablam_N \psi |_\hv + | \nablam_V \psi |_\hv + \sum_{ X = 1 }^{ n - 1 } | \nablam_{ E_X } \psi |_\hv \right) ( | \nablam_N \psi |_\hv + | \psi |_\hv ) \\
\notag &\leq ( k + l )^2 | \Dm^2 t |_\hv^2 \, \mc{O} ( \rho f^{ n - 3 } ) \cdot | \nablam_N \psi |_\hv^2 + \frac{1}{8} K \rho f^{ n - 1 } \left( | \nablam_V \psi |_\hv^2 + \sum_{ X = 1 }^{ n - 1 } | \nablam_{ E_X } \psi |_\hv^2 \right) \\
\notag &\qquad + \mc{O} ( \rho f^{ n - 1 } ) \, ( | \nablam_N \psi |_\hv^2 + | \psi |_\hv^2 ) \text{.}
\end{align}

Next, using Propositions \ref{thm.carleman_E} and \ref{thm.carleman_frame}, we can expand $\mc{Y}$ in terms of orthonormal frames:
\begin{align}
\label{eql.carleman_pointwise_J_10} \mc{Y} &= ( | \mf{X}^\rho | + | \mf{X} |_\hv ) \, [ \mc{O} ( \rho ) \cdot N + \mc{O} ( \rho ) \cdot V + \sum_{ X = 1 }^{ n - 1 } \mc{O} ( \rho ) \cdot E_X ] \\
\notag &= \mc{O} ( \rho ) \cdot N + \mc{O} ( \rho ) \cdot V + \sum_{ X = 1 }^{ n - 1 } \mc{O} ( \rho ) \cdot E_X \text{.}
\end{align}
Thus, applying \eqref{eq.carleman_lambda}, \eqref{eql.carleman_weight}, \eqref{eq.carleman_asymp}, \eqref{eq.carleman_error_v}, and \eqref{eql.carleman_pointwise_J_10}, we see that
\begin{align}
\label{eql.carleman_pointwise_J_11} | J_3 | &\leq \lambda \, \mc{O} ( f^{-1} ) \, | \mc{Y} f | | \psi |_\hv [ \mc{O} ( f^{ n - 2 } ) \, | \nablam_N \psi |_\hv + \mc{O} ( f^n ) \, | \psi |_\hv ] \\
\notag &\leq \lambda \, \mc{O} ( f^{ n - 2 } \rho ) \, ( | \nablam_N \psi |_\hv^2 + | \psi |_\hv^2 ) \text{.}
\end{align}
A similar process, along with \eqref{eq.carleman_error_h}, yields
\begin{equation}
\label{eql.carleman_pointwise_J_12} | J_9 | \leq \lambda \, \mc{O} ( \rho f^{ n - 2 } ) \, | \nablam_N \hvm |_\hv | \psi |_\hv^2 \leq \lambda \, \mc{O} ( \rho^2 f^{ n - 1 } ) \, | \psi |_\hv^2 \text{.}
\end{equation}
Furthermore, from \eqref{eq.carleman_asymp}, \eqref{eq.carleman_error_h}, and \eqref{eql.carleman_pointwise_J_10}, we obtain
\begin{align}
\label{eql.carleman_pointwise_J_13} | J_8 | &\leq \mc{O} ( \rho^2 f^{ n - 1 } ) \, | \psi |_\hv \sum_{ Z \in \{ N, V, E_1, \dots, E_{ n - 1 } \} } | \nablam_Z \psi |_\hv \\
\notag &\leq \mc{O} ( \rho^3 f^{ n - 1 } ) \sum_{ Z \in \{ N, V, E_1, \dots, E_{ n - 1 } \} } | \nablam_Z \psi |_\hv^2 + \mc{O} ( \rho f^{ n - 1 } ) \, | \psi |_\hv^2 \text{.}
\end{align}
Now, for $J_2$, we again recall \eqref{eq.carleman_asymp}, \eqref{eq.carleman_error_h}, and \eqref{eql.carleman_pointwise_J_10}:
\begin{align}
\label{eql.carleman_pointwise_J_14} | J_2 | &\leq ( | \mf{X}^\rho | + | \mf{X} |_\hv ) \sum_{ Z \in \{ N, V, E_1, \dots, E_{ n - 1 } \} } | \nablam_Z \psi |_\hv [ \mc{O} ( \rho f^{ n - 2 } ) \, | \nablam_N \psi |_\hv + \mc{O} ( \rho f^n ) \, | \psi |_\hv ] \\
\notag &\leq ( | \mf{X}^\rho |^2 + | \mf{X} |_\hv^2 ) \mc{O} ( \rho f^{ n - 3 } ) \cdot | \nablam_N \psi |_\hv^2 + \frac{1}{8} K \rho f^{ n - 1 } \left( | \nablam_V \psi |_\hv^2 + \sum_{ X = 1 }^{ n - 1 } | \nablam_{ E_X } \psi |_\hv^2 \right) \\
\notag &\qquad + \mc{O} ( \rho f^{ n - 2 } ) \, ( | \nablam_N \psi |_\hv^2 + | \psi |_\hv^2 ) \text{.}
\end{align}

Finally, for the remaining two terms, we apply \eqref{eq.carleman_asymp}, \eqref{eq.carleman_error_v}, and \eqref{eq.carleman_error_h} to estimate
\begin{align}
\label{eql.carleman_pointwise_J_20} | J_7 | &\leq \mc{O} ( f^{ n - 2 } ) \, | \nablam_N \hv |_\hv | \mc{L}^\dagger \psi |_\hv | \psi |_\hv \\
\notag &\leq \mc{O} ( \rho f^{ n - 1 } ) \, ( | \mc{L}^\dagger \psi |_\hv^2 + | \psi |_\hv^2 ) \text{,} \\
\notag | J_1 | &\leq | \mc{L}^\dagger \psi |_\hv [ f^{ n - 2 } | \nablam_N \psi |_\hv + \mc{O} ( f^n ) \cdot | \nablam_N \psi | + \mc{O} ( f^n ) \cdot | \psi |_\hv ] \\
\notag &\leq \lambda^{-1} \left[ \frac{1}{2} f^{ n - 2 - p } + \mc{O} ( f^n ) \right] | \mc{L}^\dagger \psi |_\hv^2 + \frac{1}{2} \lambda f^{ n - 2 + p } | \nablam_N \psi |_\hv^2 + \lambda \, \mc{O} ( f^n ) \, ( | \nablam_N \psi |_\hv^2 + | \psi |_\hv^2 ) \text{.}
\end{align}
Combining \eqref{eq.carleman_f}, \eqref{eq.carleman_fstar}, \eqref{eql.carleman_pointwise_J_1}--\eqref{eql.carleman_pointwise_J_4}, and \eqref{eql.carleman_pointwise_J_11}--\eqref{eql.carleman_pointwise_J_20} yields the desired identity \eqref{eq.carleman_pointwise_J}, with
\[
\mc{C}_o \lesssim_\Upsilon ( k + l )^2 | \Dm^2 t |_\hv^2 + | \mf{X}^\rho |^2 + | \mf{X} |_\hv^2 \text{.}
\]
In particular, $\mc{C}_o$ vanishes when $( k + l ) \mf{D}^2 t$, $\mf{X}^\rho$, and $\mf{X}$ all vanish.
\end{proof}

We now set $\mc{C}_o$ in Theorem \ref{thm.carleman} to be the $\mc{C}_o$ from Lemma \ref{thm.carleman_pointwise_J} (note this satisfies the required conditions in Theorem \ref{thm.carleman}).
In particular, we assume \eqref{eq.carleman_kappa} holds with this particular $\mc{C}_o$.

\subsubsection{The Pointwise Estimate}

The next key step of the proof is the following bound for $\psi$:

\begin{lemma} \label{thm.carleman_pointwise_psi}
There exists $\mc{C} > 0$, depending on $\Upsilon$, such that the following holds on $\Omega^c_< \cup \Omega^c_>$,
\begin{align}
\label{eq.carleman_pointwise_psi} \lambda^{-1} f^{ n - 2 - p } | \mc{L}^\dagger \psi |_\hv^2 &\geq \mc{C} f^{ n - 1 } \rho \left( | \nablam_N \psi |_\hv^2 + | \nablam_V \psi |_\hv^2 + \sum_{ X = 1 }^{ n - 1 } | \nablam_{ E_X } \psi |_\hv^2 \right) \\
\notag &\qquad + \frac{1}{4} \lambda^2 f^{ n - 2 + 2 p } \cdot | \psi |_\hv^2 + g^{ \alpha \beta } \nabla_\alpha P_\beta \text{,}
\end{align}
where $P$ is the (spacetime) $1$-form on $\Omega^c_< \cup \Omega^c_>$ given by
\begin{align}
\label{eq.carleman_pointwise_psi_PP} P_\beta &:= \frac{1}{2} [ \langle \nablam_S ( \hvd \psi ), \nablam_\beta \psi \rangle + \langle \nablam_S \psi, \nablam_\beta ( \hvd \psi ) \rangle - g_{ \alpha \beta } S^\alpha \, g^{ \mu \nu } \langle \nablam_\mu ( \hvd \psi ), \nablam_\nu \psi \rangle ] \\
\notag &\qquad + \frac{1}{2} v_\zeta [ \langle \nablam_\beta ( \hvd \psi ), \psi \rangle + \langle \nablam_\beta \psi, \hvd \psi \rangle ] - \frac{1}{2} \nabla_\beta v_\zeta \, \hvm ( \psi, \psi ) + \frac{1}{2} g_{ \alpha \beta } S^\alpha \mc{A} \cdot | \psi |_\hv^2 \\
\notag &\qquad + ( 2 \kappa - n + 1 - \mc{C}_o ) f^{ n - 3 } \nabla_\beta f \cdot | \psi |_\hv^2 + \left( 1 - \frac{p}{2} \right) \lambda f^{ n - 3 + p } \nabla_\beta f \cdot | \psi |_\hv^2 \text{.}
\end{align}
Furthermore, there exists $\mc{C}_b > 0$, also depending on $\Upsilon$, such that
\begin{align}
\label{eq.carleman_pointwise_psi_P} P ( \rho \partial_\rho ) \leq \mc{C}_b f^{ n - 2 } \rho^2 \cdot ( | \nablam_\rho \psi |_\hv^2 + | \nablam_{ \Dv^\sharp t } \psi |_\hv^2 ) + \mc{C}_b \lambda^2 f^{n-2} \cdot | \psi |_\hv^2 \text{.}
\end{align}
\end{lemma}

\begin{proof}
Let $J$ be as in \eqref{eq.carleman_J}.
Expanding $J$ using \eqref{eql.carleman_LA}, we obtain
\begin{align}
\label{eql.carleman_pointwise_psi_1} J &= \frac{1}{2} [ \langle \Boxm ( \hvd \psi ), \bar{S}_\zeta \psi \rangle + \langle \Boxm \psi, \bar{S}_\zeta ( \hvd \psi ) \rangle ] + F' f^{ -n + 3 } [ \langle \nablam_S ( \hvd \psi ), \bar{S}_\zeta \psi \rangle + \langle \nablam_S \psi, \bar{S}_\zeta ( \hvd \psi ) \rangle ] \\
\notag &\qquad + \frac{1}{2} \mc{A} [ \langle \hvd \psi, \bar{S}_\zeta \psi \rangle + \langle \psi, \bar{S}_\zeta ( \hvd \psi ) \rangle ] \\
\notag &= g^{ \alpha \beta } \nabla_\alpha P^S_\beta + \frac{1}{2} [ \langle \Boxm ( \hvd \psi ), \bar{S}_\zeta \psi \rangle + \langle \Boxm \psi, \bar{S}_\zeta ( \hvd \psi ) \rangle ] + 2 F' f^{ - n + 3 } \cdot | \nablam_S \psi |_\hv^2 \\
\notag &\qquad + \left[ \mc{A} v_\zeta - \frac{1}{2} \nabla_\beta ( \mc{A} S^\beta ) \right] \cdot | \psi |_\hv^2 + I_S + I_\zeta \text{,}
\end{align}
where the $1$-form $P^S$ and the scalars $I_S$, $I_\zeta$ are defined as
\begin{align}
\label{eql.carleman_pointwise_psi_2} P^S_\beta &:= \frac{1}{2} g_{ \alpha \beta } S^\alpha \mc{A} \cdot | \psi |_\hv^2 \text{,} \\
\notag I_S &:= 2 F' f^{ -n + 3 } \cdot \nablam_S \hvm ( \psi, \nablam_S \psi ) \text{,} \\
\notag I_\zeta &:= F' f^{ -n + 3 } v_\zeta \, [ \langle \nablam_S ( \hvd \psi ), \psi \rangle + \langle \nablam_S \psi, \hvd \psi \rangle ] \text{.}
\end{align}

Next, consider the following spacetime tensor fields associated to $\psi$:
\begin{align}
\label{eql.carleman_pointwise_psi_4} Q_{ \alpha \beta } &:= \frac{1}{2} [ \langle \nablam_\alpha ( \hvd \psi ), \nablam_\beta \psi \rangle + \langle \nablam_\alpha \psi, \nablam_\beta ( \hvd \psi ) \rangle - g_{ \alpha \beta } g^{ \mu \nu } \langle \nablam_\mu ( \hvd \psi ), \nablam_\nu \psi \rangle ] \text{,} \\
\notag P^Q_\beta &:= Q_{ \alpha \beta } S^\alpha + \frac{1}{2} v_\zeta [ \langle \nablam_\beta ( \hvd \psi ), \psi \rangle + \langle \nablam_\beta \psi, \hvd \psi \rangle ] - \frac{1}{2} \nabla_\beta v_\zeta \, \hvm ( \psi, \psi ) \text{.}
\end{align}
A direct computation using \eqref{eql.carleman_pi} and \eqref{eql.carleman_pointwise_psi_4} then yields
\begin{align}
\label{eql.carleman_pointwise_psi_6} g^{ \alpha \beta } \nabla_\alpha P^Q_\beta &= \frac{1}{2} [ \langle \Boxm ( \hvd \psi ), \bar{S}_\zeta \psi \rangle + \langle \Boxm \psi, \bar{S}_\zeta ( \hvd \psi ) \rangle ] - g^{ \alpha \mu } g^{ \beta \nu } ( \pi_\zeta )_{ \alpha \beta } \, \hvm ( \nablam_\mu \psi, \nablam_\nu \psi ) \\
\notag &\qquad - \frac{1}{2} \Box v_\zeta \, | \psi |_\hv^2 - I_R - I_\pi \text{,} \\
\notag I_R &:= \frac{1}{2} g^{ \mu \nu } S^\alpha [ \langle \bar{R}_{ \alpha \mu } [ \hvd \psi ], \nablam_\nu \psi \rangle + \langle \bar{R}_{ \alpha \mu } [ \psi ], \nablam_\nu ( \hvd \psi ) \rangle ] \text{,} \\
\notag I_\pi &:= g^{ \alpha \mu } g^{ \beta \nu } ( \pi_\zeta )_{ \alpha \beta } \, \nablam_\mu \hvm ( \psi, \nablam_\nu \psi ) \text{.}
\end{align}
Combining \eqref{eql.carleman_pointwise_psi_1} and \eqref{eql.carleman_pointwise_psi_6} results in the identity
\begin{align}
\label{eql.carleman_pointwise_psi_10} J &= g^{ \alpha \beta } \nabla_\alpha ( P^Q + P^S )_\beta + 2 F' f^{ - n + 3 } \cdot | \nablam_S \psi |_\hv^2 + g^{ \alpha \mu } g^{ \beta \nu } ( \pi_\zeta )_{ \alpha \beta } \, \hvm ( \nablam_\mu \psi, \nablam_\nu \psi ) \\
\notag &\qquad + \left[ \mc{A} v_\zeta - \frac{1}{2} \nabla_\beta ( \mc{A} S^\beta ) + \frac{1}{2} \Box v_\zeta \right] \cdot | \psi |_\hv^2 + I_S + I_\zeta + I_R + I_\pi \text{.}
\end{align}

Recalling \eqref{eql.carleman_weight} and \eqref{eq.carleman_asymp}, we see that
\begin{equation}
\label{eql.carleman_pointwise_psi_11} 2 F' f^{ - n + 3 } \cdot | \nablam_S \psi |_\hv^2 = [ 2 \kappa f^{ n - 2 } + 2 \lambda f^{ n - 2 + p } + \lambda \, \mc{O} ( f^n ) ] \cdot | \nablam_N \psi |_\hv^2 \text{.}
\end{equation}
Moreover, by Lemma \ref{thm.carleman_pseudoconvex}, we have
\begin{align}
\label{eql.carleman_pointwise_psi_12} g^{ \alpha \mu } g^{ \beta \nu } ( \pi_\zeta )_{ \alpha \beta } \, \hvm ( \nablam_\mu \psi, \nablam_\nu \psi ) &\geq K f^{ n - 1 } \rho \left( | \nablam_V \psi |_\hv^2 + \sum_{ X = 1 }^{ n - 1 } | \nablam_{ E_X } \psi |_\hv^2 \right) \\
\notag &\qquad - [ ( n - 1 ) f^{ n - 2 } + \mc{O} ( f^n ) ] \cdot | \nablam_N \psi |_\hv^2 \text{,}
\end{align}
with $K > 0$ as in the statement of Lemma \ref{thm.carleman_pseudoconvex}.
By \eqref{eq.carleman_error_v} and \eqref{eq.carleman_error_A}, we also obtain
\begin{align}
\label{eql.carleman_pointwise_psi_13} \mc{A} v_\zeta - \frac{1}{2} \nabla_\beta ( \mc{A} S^\beta ) + \frac{1}{2} \Box v_\zeta &= ( \kappa^2 - n \kappa + \sigma ) f^{ n - 2 } + \left( 1 - \frac{p}{2} \right) ( 2 \kappa - n + p ) \lambda f^{ n - 2 + p } \\
\notag &\qquad + ( 1 - p ) \lambda^2 f^{ n - 2 + 2 p } + \lambda^2 \, \mc{O} ( f^n ) \text{.}
\end{align}
Thus, \eqref{eql.carleman_pointwise_psi_10}--\eqref{eql.carleman_pointwise_psi_13} together yield
\begin{align}
\label{eql.carleman_pointwise_psi_20} J &\geq [ ( 2 \kappa - n + 1 ) f^{ n - 2 } + 2 \lambda f^{ n - 2 + p } + \lambda \, \mc{O} ( f^n ) ] \cdot | \nablam_N \psi |_\hv^2 \\
\notag &\qquad + K f^{ n - 1 } \rho \left( | \nablam_V \psi |_\hv^2 + \sum_{ X = 1 }^{ n - 1 } | \nablam_{ E_X } \psi |_\hv^2 \right) + ( \kappa^2 - n \kappa + \sigma ) f^{ n - 2 } \cdot | \psi |_\hv^2 \\
\notag &\qquad + \left[ \left( 1 - \frac{p}{2} \right) ( 2 \kappa - n + p ) \lambda f^{ n - 2 + p } + ( 1 - p ) \lambda^2 f^{ n - 2 + 2 p } + \lambda^2 \, \mc{O} ( f^n ) \right] | \psi |_\hv^2 \\
\notag &\qquad + I_S + I_\zeta + I_R + I_\pi + g^{ \alpha \beta } \nabla_\alpha ( P^Q + P^S )_\beta \text{.}
\end{align}

We next handle the error terms.
First, we bound $I_S$ using \eqref{eq.carleman_lambda}, \eqref{eql.carleman_weight}, \eqref{eq.carleman_asymp}, and \eqref{eq.carleman_error_h}:
\begin{align}
\label{eql.carleman_pointwise_psi_21} | I_S | &\leq \lambda \, \mc{O} ( f^{ n - 2 } ) \, | \nablam_N \hvm |_\hv | \nablam_N \psi |_\hv | \psi |_\hv \\
\notag &\leq \lambda \, \mc{O} ( \rho f^{ n - 1 } ) \, ( | \nablam_N \psi |_\hv^2 + | \psi |_\hv^2 ) \text{.}
\end{align}
Similarly, for $I_\zeta$, we apply \eqref{eq.carleman_fstar}, \eqref{eq.carleman_lambda}, \eqref{eql.carleman_weight}, \eqref{eq.carleman_asymp}, and \eqref{eq.carleman_error_v}:
\begin{align}
\label{eql.carleman_pointwise_psi_22} | I_\zeta | &\leq \lambda \, \mc{O} ( f^n ) \, ( | \nablam_N ( \hvd \psi ) |_\hv | \psi |_\hv + | \nablam_N \psi |_\hv | \hvd \psi |_\hv ) \\
\notag &\leq \lambda \, \mc{O} ( f^n ) \, ( | \nablam_N \psi |_\hv | \psi |_\hv + | \nablam_N \hvm | | \psi |_\hv^2 ) \\
\notag &\leq \lambda \, \mc{O} ( f^n ) \, ( | \nablam_N \psi |_\hv^2 + | \psi |_\hv^2 ) \text{.} 
\end{align}
To estimate $I_\pi$, we expand using orthonormal frames and apply \eqref{eq.carleman_pseudoconvex_pi}, and \eqref{eq.carleman_error_h}:
\begin{align}
\label{eql.carleman_pointwise_psi_23} | I_\pi | &\leq \sum_{ Z_1, Z_2 \in \{ N, V, E_1, \dots, E_{ n - 1 } \} } | ( \pi_\zeta )_{ Z_1 Z_2 } | | \nablam_{ Z_1 } \hvm |_\hv \cdot | \nablam_{ Z_2 } \psi |_\hv | \psi |_\hv \\
\notag &\leq \mc{O} ( \rho f^{ n - 1 } ) \, | \nablam_N \psi |_\hv | \psi |_\hv + \mc{O} ( \rho^2 f^{ n - 1 } ) \, | \nablam_V \psi |_\hv | \psi |_\hv \\
\notag &\qquad + \sum_{ X = 1 }^{ n - 1 } \mc{O} ( \rho^2 f^{ n - 1 } ) \, | \nablam_{ E_X } \psi |_\hv | \psi |_\hv \\
\notag &\leq \mc{O} ( \rho^3 f^{ n - 1 } ) \, \left( | \nablam_V \psi |_\hv^2 + \sum_{ X = 1 }^{ n - 1 } | \nablam_{ E_X } \psi |_\hv^2 \right) + \mc{O} ( \rho f^{ n - 1 } ) \, ( | \nablam_N \psi |_\hv^2 + | \psi |_\hv^2 ) \text{.}
\end{align}

Lastly, for $I_R$, we expand and apply \eqref{eq.carleman_asymp}, \eqref{eq.carleman_error_R}, and \eqref{eq.carleman_error_h}:
\begin{align}
\label{eql.carleman_pointwise_psi_24} | I_R | &\leq \mc{O} ( f^{ n - 2 } ) \sum_{ Z \in \{ V, E_1, \dots, E_{ n - 1 } \} } ( | \bar{R}_{ N Z } [ \hvd \psi ] |_\hv + | \bar{R}_{ N Z } [ \psi ] |_\hv ) | \nablam_Z \psi |_\hv \\
\notag &\qquad + \mc{O} ( f^{ n - 2 } ) \sum_{ Z \in \{ V, E_1, \dots, E_{ n - 1 } \} } | \nablam_Z \hvm |_\hv | \bar{R}_{ N Z } [ \psi ] |_\hv | \psi |_\hv ) \\
\notag &\leq \mc{O} ( \rho^3 f^{ n - 1 } ) \, \left( | \nablam_V \psi |_\hv^2 + \sum_{ X = 1 }^{ n - 1 } | \nablam_{ E_X } \psi |_\hv^2 \right) + \mc{O} ( \rho f^{ n - 1 } ) \, ( | \nablam_N \psi |_\hv^2 + | \psi |_\hv^2 ) \text{.}
\end{align}
Thus, combining \eqref{eql.carleman_pointwise_psi_20}--\eqref{eql.carleman_pointwise_psi_24} and recalling \eqref{eq.carleman_f} and \eqref{eq.carleman_lambda}, we conclude that
\begin{align*}
J &\geq g^{ \alpha \beta } \nabla_\alpha ( P^Q + P^S )_\beta + [ ( 2 \kappa - n + 1 ) f^{ n - 2 } + 2 \lambda f^{ n - 2 + p } + \lambda \, \mc{O} ( f^n ) ] \cdot | \nablam_N \psi |_\hv^2 \\
&\qquad + \frac{1}{2} K f^{ n - 1 } \rho \left( | \nablam_V \psi |_\hv^2 + \sum_{ X = 1 }^{ n - 1 } | \nablam_{ E_X } \psi |_\hv^2 \right) + ( \kappa^2 - n \kappa + \sigma ) f^{ n - 2 } \cdot | \psi |_\hv^2 \\
&\qquad + \left[ \left( 1 - \frac{p}{2} \right) ( 2 \kappa - n + p ) \lambda f^{ n - 2 + p } + ( 1 - p ) \lambda^2 f^{ n - 2 + 2 p } + \lambda^2 \, \mc{O} ( f^n ) \right] | \psi |_\hv^2 \text{.}
\end{align*}
Furthermore, combining \eqref{eq.carleman_fstar}, \eqref{eq.carleman_pointwise_J}, and the above, we conclude that
\begin{align}
\label{eql.carleman_pointwise_psi_40} \frac{1}{ \lambda } f^{ n - 2 - p } | \mc{L}^\dagger \psi |_\hv^2 &\geq \left[ ( 2 \kappa - n + 1 - \mc{C}_o ) f^{ n - 2 } + \frac{3}{2} \lambda f^{ n - 2 + p } \right] \cdot | \nablam_N \psi |_\hv^2 \\
\notag &\qquad + \frac{1}{4} K f^{ n - 1 } \rho \left( | \nablam_V \psi |_\hv^2 + \sum_{ X = 1 }^{ n - 1 } | \nablam_{ E_X } \psi |_\hv^2 \right) + ( \kappa^2 - n \kappa + \sigma ) f^{ n - 2 } \cdot | \psi |_\hv^2 \\
\notag &\qquad + \left[ \left( 1 - \frac{p}{2} \right) ( 2 \kappa - n + p ) \lambda f^{ n - 2 + p } + ( 1 - p ) \lambda^2 f^{ n - 2 + 2 p } \right] | \psi |_\hv^2 \\
\notag &\qquad + \lambda \, \mc{O} ( f^{ n - 1 } ) \cdot | \nablam_N \psi |_\hv^2 + \lambda^2 \, \mc{O} ( f^{ n - 1 } ) \cdot | \psi |_\hv^2 + g^{ \alpha \beta } \nabla_\alpha ( P^Q + P^S )_\beta \text{.}
\end{align}

In addition, given any $b, q \in \R$, we have the inequality
\begin{align*}
0 &\leq f^{ q - 2 } | g^{ \alpha \beta } \nabla_\alpha f \nablam_\beta \psi + b f \cdot \psi |_\hv^2 \\
&= f^{ q - 2 } | g^{ \alpha \beta } \nabla_\alpha f \nablam_\beta \psi |_\hv^2 + [ b^2 f^q - b ( q - 1 ) f^{ q - 2 } g^{ \alpha \beta } \nabla_\alpha f \nabla_\beta f - b f^{ q - 1 } \Box f ] \cdot | \psi |_\hv^2 \\
\notag &\qquad - b f^{ q - 1 } g^{ \alpha \beta } \nabla_\alpha f \nablam_\beta \hvm ( \psi, \psi ) + g^{ \alpha \beta } \nabla_\alpha ( b f^{ q - 1 } \nabla_\beta f \cdot | \psi |_\hv^2 ) \text{.}
\end{align*}
Setting $b := \frac{1}{2} ( q - n )$ and applying \eqref{eq.carleman_asymp} to the above, we then obtain
\begin{align}
\label{eql.carleman_pointwise_psi_41} f^q | \nablam_N \psi |_\hv^2 &\geq \frac{1}{4} ( q - n )^2 f^q \cdot | \psi |_\hv^2 + ( 1 + q^2 ) \, \mc{O} ( f^{ q + 2 } ) \cdot ( | \nablam_N \psi |_\hv^2 + | \psi |_\hv^2 ) \\
\notag &\qquad + I_{ H, q } + g^{ \alpha \beta } \nabla_\alpha P^{ H, q }_\beta \text{,} \\
\notag P^{ H, q }_\beta &:= - \frac{1}{2} ( q - n ) f^{ q - 1 } \nabla_\beta f \cdot | \psi |_\hv^2 \text{,} \\
\notag I_{ H, q } &:= \frac{1}{2} ( q - n ) f^{ q - 1 } g^{ \alpha \beta } \nabla_\alpha f \nablam_\beta \hvm ( \psi, \psi ) \text{.}
\end{align}
Moreover, using \eqref{eq.carleman_asymp} and \eqref{eq.carleman_error_h}, the terms $I_{ H, q }$ can be estimated as follows:
\begin{equation}
\label{eql.carleman_pointwise_psi_42} | I_{ H, q } | \leq ( q - n ) \, \mc{O} ( \rho f^{ q + 1 } ) \, | \psi |^2 \text{.}
\end{equation}

Noting that $2 \kappa - ( n - 1 ) - \mc{C}_o \geq 0$ by \eqref{eq.carleman_kappa}, we can now apply \eqref{eql.carleman_pointwise_psi_41} and \eqref{eql.carleman_pointwise_psi_42} twice, with $q := n - 2$ and $q := n - 2 + p$, to \eqref{eql.carleman_pointwise_psi_40}.
This the results in the inequality
\begin{align}
\label{eql.carleman_pointwise_psi_50} \frac{1}{ \lambda } f^{ n - 2 - p } | \mc{L}^\dagger \psi |_\hv^2 &\geq \frac{1}{2} \lambda f^{ n - 2 + p } \cdot | \nablam_N \psi |_\hv^2 + \frac{1}{4} K f^{ n - 1 } \rho \left( | \nablam_V \psi |_\hv^2 + \sum_{ X = 1 }^{ n - 1 } | \nablam_{ E_X } \psi |_\hv^2 \right) \\
\notag &\qquad + [ \kappa^2 - ( n - 2 ) \kappa + \sigma - ( n - 1 ) - \mc{C}_o ] f^{ n - 2 } \cdot | \psi |_\hv^2 \\
\notag &\qquad + \left[ \left( 1 - \frac{p}{2} \right) \left( 2 \kappa - n + 1 + \frac{p}{2} \right) \lambda f^{ n - 2 + p } + ( 1 - p ) \lambda^2 f^{ n - 2 + 2 p } \right] | \psi |_\hv^2 \\
\notag &\qquad + \lambda \, \mc{O} ( f^{ n - 1 } ) \cdot | \nablam_N \psi |_\hv^2 + \lambda^2 \, \mc{O} ( f^{ n - 1 } ) \cdot | \psi |_\hv^2 + g^{ \alpha \beta } \nabla_\alpha ( P^Q + P^S )_\beta \\
\notag &\qquad + g^{ \alpha \beta } \nabla_\alpha [ \lambda P^{ H, n - 2 + p } + ( 2 \kappa - n + 1 - \mc{C}_o ) P^{ H, n - 2 } ]_\beta \text{.}
\end{align}
In addition, we now set
\begin{equation}
\label{eql.carleman_pointwise_psi_51} P := P^Q + P^S + \lambda P^{ H, n - 2 + p } + ( 2 \kappa - n + 1 - \mc{C}_o ) P^{ H, n - 2 } \text{,}
\end{equation}
and we note that this corresponds with the quantity given in the right-hand side \eqref{eq.carleman_pointwise_psi_P}.
Applying \eqref{eql.carleman_pointwise_psi_50} and \eqref{eql.carleman_pointwise_psi_51}, and keeping in mind \eqref{eq.carleman_kappa}--\eqref{eq.carleman_lambda} as well, we obtain the bound
\begin{align*}
\frac{1}{ \lambda } f^{ n - 2 - p } | \mc{L}^\dagger \psi |_\hv^2 &\geq \frac{1}{4} \lambda f^{ n - 2 + p } \cdot | \nablam_N \psi |_\hv^2 + \frac{1}{4} K f^{ n - 1 } \rho \left( | \nablam_V \psi |_\hv^2 + \sum_{ X = 1 }^{ n - 1 } | \nablam_{ E_X } \psi |_\hv^2 \right) \\
&\qquad + \frac{1}{4} \lambda^2 f^{ n - 2 + 2 p } | \psi |_\hv^2 + g^{ \alpha \beta } \nabla_\alpha P_\beta \text{,}
\end{align*}
In particular, \eqref{eq.carleman_fstar}, \eqref{eq.carleman_lambda}, and the above immediately imply the desired estimate \eqref{eq.carleman_pointwise_psi}.

It remains to show \eqref{eq.carleman_pointwise_psi_P}.
First, from \eqref{eq.carleman_f}, \eqref{eql.carleman_S}, \eqref{eq.carleman_error_A}, and \eqref{eql.carleman_pointwise_psi_2}, we have
\begin{equation}
\label{eql.carleman_pointwise_psi_61} | P^S ( \rho \partial_\rho ) | \leq \mc{O} ( f^{ n - 3 } ) \, | \rho \partial_\rho f | \cdot \lambda^2 \, \mc{O} ( 1 ) \cdot | \psi |_\hv^2 = \lambda^2 \, \mc{O} ( f^{ n - 2 } ) \cdot | \psi |_\hv^2 \text{.}
\end{equation}
Similarly, recalling \eqref{eq.carleman_f}, we obtain
\begin{align}
\label{eql.carleman_pointwise_psi_62} | \lambda P^{ H, n - 2 + p } ( \rho \partial_\rho ) + ( 2 \kappa - n + 1 - \mc{C}_o ) P^{ H, n - 2 } ( \rho \partial_\rho ) | &\leq \lambda \, \mc{O} ( f^{ n - 3 } ) \, | \rho \partial_\rho f | \cdot | \psi |_\hv^2 \\
\notag &= \lambda \, \mc{O} ( f^{ n - 2 } ) \, | \psi |_\hv^2 \text{.}
\end{align}
The remaining term $P^Q$ is more involved; we begin by expanding \eqref{eql.carleman_pointwise_psi_4}:
\begin{align}
\label{eql.carleman_pointwise_psi_63} P^Q_\beta ( \rho \partial_\rho ) &:= \frac{1}{2} \rho [ \langle \nablam_\rho ( \hvd \psi ), \nablam_S \psi \rangle + \langle \nablam_\rho \psi, \nablam_S ( \hvd \psi ) \rangle ] - \frac{1}{2} \, g ( \rho \partial_\rho, S ) \, g^{ \mu \nu } \langle \nablam_\mu ( \hvd \psi ), \nablam_\nu \psi \rangle \\
\notag &\qquad + \frac{1}{2} \rho v_\zeta [ \langle \nablam_\rho ( \hvd \psi ), \psi \rangle + \langle \nablam_\rho \psi, \hvd \psi \rangle ] - \frac{1}{2} \rho \nabla_\rho v_\zeta \cdot \hvm ( \psi, \psi ) \\
\notag &:= B_1 + B_2 + B_3 + B_4 \text{.}
\end{align}

First, $B_4$ can be controlled using \eqref{eq.carleman_error_v}:
\begin{equation}
\label{eql.carleman_pointwise_psi_64} | B_4 | \leq \mc{O} ( f^n ) \cdot | \psi |_\hv^2 \text{.}
\end{equation}
Similarly, by \eqref{eq.carleman_error_v} and \eqref{eq.carleman_error_h},
\begin{align}
\label{eql.carleman_pointwise_psi_65} | B_3 | &\leq \mc{O} ( \rho f^n ) \, ( | \nablam_\rho \psi |_\hv + | \nablam_\rho \hvm |_\hv | \psi |_\hv ) | \psi |_\hv \\
\notag &= \mc{O} ( \rho^2 f^n ) \, | \nablam_\rho \psi |_\hv^2 + \mc{O} ( f^n ) \, | \psi |_\hv^2 \text{.}
\end{align}
Next, by \eqref{eq.carleman_NV}, \eqref{eq.carleman_asymp}, and \eqref{eq.carleman_error_h}, we can bound
\begin{align}
\label{eql.carleman_pointwise_psi_66} | B_1 | &\leq \mc{O} ( \rho f^{ n - 2 } ) \, ( | \nablam_\rho \psi |_\hv | \nablam_N \psi |_\hv + | \nablam_N \hvm |_\hv | \nablam_\rho \psi |_\hv | \psi |_\hv + | \nablam_\rho \hvm |_\hv | \nablam_N \psi |_\hv | \psi |_\hv ) \\
\notag &\leq \mc{O} ( \rho f^{ n - 2 } ) \, [ \mc{O} ( \rho ) \, | \nablam_\rho \psi |_\hv | \nablam_\rho \psi |_\hv + \mc{O} ( \rho f ) \, | \nablam_\rho \psi |_\hv | \nablam_{ \Dv^\sharp t } \psi |_\hv ] \\
\notag &\qquad + \mc{O} ( \rho f^{ n - 2 } ) \, [ \mc{O} ( \rho f ) \, | \nablam_\rho \psi |_\hv | \psi |_\hv + \mc{O} ( \rho^2 ) \, | \nablam_\rho \psi |_\hv | \psi |_\hv + \mc{O} ( \rho^2 f ) \, | \nablam_{ \Dv^\sharp t } \psi |_\hv | \psi |_\hv ) \\
\notag &\leq \mc{O} ( \rho^2 f^{ n - 2 } ) \, | \nablam_\rho \psi |_\hv^2 + \mc{O} ( \rho^2 f^n ) \, | \nablam_{ \Dv^\sharp t } \psi |_\hv^2 + \mc{O} ( \rho^2 f^n ) \, | \psi |_\hv^2 \text{.}
\end{align}

For $B_2$, we must be a bit more careful about the signs.
Since \eqref{eq.carleman_f} and \eqref{eql.carleman_S} imply
\[
g ( \rho \partial, S ) = f^{ n - 3 } \rho \partial_\rho f = f^{ n - 2 } \text{,}
\]
which is strictly positive, then \eqref{eq.carleman_fstar}, \eqref{eq.carleman_error_h}, and the above yield
\begin{align}
\label{eql.carleman_pointwise_psi_68} B_2 &\leq - \frac{1}{2} f^{ n - 2 } \left[ | \nablam_N \psi |_\hv^2 - | \nablam_V \psi |_\hv^2 + \sum_{ X = 1 }^{ n - 1 } | \nablam_{ E_X } \psi |_\hv^2 \right] \\
\notag &\qquad + \mc{O} ( f^{ n - 2 } ) \, \left[ \mc{O} ( \rho f ) \, | \nablam_N \psi |_\hv + \mc{O} ( \rho ) \, | \nablam_V \psi |_\hv + \sum_{ X = 1 }^{ n - 1 } \mc{O} ( \rho ) \, | \nablam_{ E_X } \psi |_\hv \right] | \psi |_\hv \\
\notag &\leq f^{ n - 2 } | \nablam_V \psi |^2 + \mc{O} ( \rho^2 f^{ n - 2 } ) \cdot | \psi |_\hv^2 \\
\notag &\leq \mc{O} ( \rho^2 f^n ) \cdot | \nablam_\rho \psi |^2 + \mc{O} ( \rho^2 f^{ n - 2 } ) \cdot | \nablam_{ \Dv^\sharp t } \psi |^2 + \mc{O} ( \rho^2 f^{ n - 2 } ) \cdot | \psi |_\hv^2 \text{.}
\end{align}
Finally, combining \eqref{eql.carleman_pointwise_psi_51}, \eqref{eql.carleman_pointwise_psi_61}--\eqref{eql.carleman_pointwise_psi_68} results in \eqref{eq.carleman_pointwise_psi_P} and completes the proof.
\end{proof}

\subsubsection{Completion of the Proof}

We now have all the key ingredients to complete the proof of Theorem \ref{thm.carleman}.
Next, we use Lemma \ref{thm.carleman_pointwise_psi} to derive a corresponding bound for $\phi$:

\begin{lemma} \label{thm.carleman_pointwise}
There exist $\mc{C}, \mc{C}_b > 0$, depending on $\Upsilon$, such that the following holds on $\Omega^c_< \cup \Omega^c_>$,
\begin{equation}
\label{eq.carleman_pointwise} \lambda^{-1} \mc{E} f^{ -p } | ( \Boxm + \sigma + \nablam_{ \mc{Y} } ) \phi |_\hv^2 \geq \mc{C} \mc{E} f \rho^3 ( | \nablam_\rho \phi |_\hv^2 + | \Dv \phi |_\hv^2 ) + \frac{1}{8} \lambda^2 \mc{E} f^{ 2 p } | \phi |_\hv^2 + g^{ \alpha \beta } \nabla_\alpha P_\beta \text{,}
\end{equation}
where $\mc{E}$ denotes the function
\begin{equation}
\label{eq.carleman_pointwise_weight} \mc{E} := e^{ -2 F } f^{ n - 2 } = e^{ -2 \lambda p^{-1} f^p } f^{ n - 2 - 2 \kappa } \text{,}
\end{equation}
and where the $1$-form $P$ is as in Lemma \ref{thm.carleman_pointwise_psi} and satisfies
\begin{align}
\label{eq.carleman_pointwise_P} \rho^{-n} \, P ( \rho \partial_\rho ) \leq \mc{C}_b [ | \nablam_\rho ( \rho^{ -\kappa } \phi ) |_\hv^2 + | \nablam_{ \Dv^\sharp t } ( \rho^{ -\kappa } \phi ) |_\hv^2 ] + \mc{C}_b \lambda^2 | \rho^{ -\kappa - 1 } \phi |_\hv^2 \text{.}
\end{align}
\end{lemma}

\begin{proof}
First, note that the second equality in \eqref{eq.carleman_pointwise_weight} is an immediate consequence of \eqref{eql.carleman_weight}.
We now apply \eqref{eq.carleman_pointwise_psi} and recall \eqref{eql.carleman_weight}, \eqref{eql.carleman_XX}, and \eqref{eq.carleman_pointwise_weight} to obtain, for some $\mc{C} > 0$ depending on $\Upsilon$,
\begin{align}
\label{eql.carleman_pointwise_1} \lambda^{-1} \mc{E} f^{ - p } | ( \Boxm + \sigma + \nablam_{ \mc{Y} } ) \phi |_\hv^2 &\geq \mc{C} f^{ n - 1 } \rho \left[ | \nablam_N \psi |_\hv^2 + | \nablam_V \psi |_\hv^2 + \sum_{ X = 1 }^{ n - 1 } | \nablam_{ E_X } \psi |_\hv^2 \right] \\
\notag &\qquad + \frac{1}{4} \lambda^2 \mc{E} f^{ 2 p } | \phi |_\hv^2 + g^{ \alpha \beta } \nabla_\alpha P_\beta \text{.}
\end{align}

Next, note that by \eqref{eq.aads_strong_limits}, Proposition \ref{thm.aads_geom_limit}, Proposition \ref{thm.carleman_E}, and \eqref{eq.carleman_NV},
\begin{align}
\label{eql.carleman_pointwise_3} \rho^4 | \nablam_\rho \phi |_\hv^2 &\leq \mc{O} ( \rho^2 ) \, | \nablam_N \phi |_\hv^2 + \mc{O} ( f^2 \rho^2 ) \, | \nablam_V \phi |_\hv^2 \text{,} \\
\notag \rho^4 | \Dv \phi |_\hv^2 &\leq \mc{O} ( f^2 \rho^2 ) \, | \nablam_N \phi |_\hv^2 + \mc{O} ( \rho^2 ) \, | \nablam_V \phi |_\hv^2 + \sum_{ X = 1 }^{ n - 1 } \mc{O} ( \rho^2 ) \, | \nablam_{ E_X } \phi |_\hv^2 \text{.}
\end{align}
Moreover, applying \eqref{eq.carleman_NV}, \eqref{eq.carleman_lambda}, and \eqref{eql.carleman_weight} yields
\[
e^{ -2 F } | \nablam_N \phi |_\hv^2 + e^{ -2 F } | \nablam_V \phi |_\hv^2 \leq | \nablam_N \psi |_\hv^2 + | \nablam_V \psi |_\hv^2 + \lambda^2 \, \mc{O} (1) \, | \psi |_\hv^2 \text{,} \qquad e^{ -2 F } | \nablam_{ E_X } \phi |_\hv^2 = | \nablam_{ E_X } \psi |_\hv^2 \text{,}
\]
for each $1 \leq X < n$.
Thus, combining \eqref{eq.carleman_fstar}, \eqref{eql.carleman_pointwise_1}, \eqref{eql.carleman_pointwise_3} and the above results in \eqref{eq.carleman_pointwise}.

For \eqref{eq.carleman_pointwise_P}, we begin by noticing that
\[
e^{ - \frac{ \lambda f^p }{p} } \leq 1 \text{,} \qquad \left| \nabla_\rho \left( e^{ - \frac{ \lambda f^p }{p} } \right) \right| \leq \lambda \, \mc{O} ( f^p \rho^{-1} ) \text{,} \qquad \left| \Dv \left( e^{ - \frac{ \lambda f^p }{p} } \right) \right|_\hv \leq \lambda \, \mc{O} ( f^{ p + 1 } \rho^{-1} ) \text{,}
\]
where we also applied \eqref{eq.carleman_f} in the second and third inequalities.
Next, since $f^{-1} \leq \rho^{-1}$ by \eqref{eq.carleman_f}, and since $2 \kappa \geq n - 1$ by \eqref{eq.carleman_kappa}, then \eqref{eql.carleman_weight} and the above imply
\begin{align*}
f^{ n - 2 } \rho^{-n} | \psi |_\hv^2 &\leq | \rho^{ - \kappa - 1 } \phi |_\hv^2 \text{,} \\
f^{ n - 2 } \rho^{ -n + 2 } | \nablam_\rho \psi |_\hv^2 &\leq | \nablam_\rho ( \rho^{ - \kappa } \phi ) |_\hv^2 + \lambda^2 \, \mc{O} (1) \, | \rho^{ - \kappa - 1} \phi |_\hv^2 \text{,} \\
f^{ n - 2 } \rho^{ -n + 2 } | \nablam_{ \Dv^\sharp t } \psi |_\hv^2 &\leq | \nablam_{ \Dv^\sharp t } ( \rho^{ - \kappa } \phi ) |_\hv^2 + \lambda^2 \, \mc{O} (1) \, | \rho^{ - \kappa - 1} \phi |_\hv^2 \text{.}
\end{align*}
The bound \eqref{eq.carleman_pointwise_P} now follows from multiplying \eqref{eq.carleman_pointwise_psi_P} by $\rho^{-n}$ and applying the above.
\end{proof}

The final step is to integrate \eqref{eq.carleman_pointwise} over $\Omega^c_< \cup \Omega^c_>$.
For convenience, we define
\footnote{The extra cutoff in $\rho$ is needed since $\Omega^c$ and $\Omega_0 ( f_\ast )$ have infinite volume.}
\begin{equation}
\label{eql.carleman_0} 0 < \rho_\ast < f_\ast \text{,} \qquad \Omega_0 ( f_\ast, \rho_\ast ) := \Omega_0 ( f_\ast ) \cap \{ \rho > \rho_\ast \} \text{.}
\end{equation}
Note the left-hand side and the first two terms on the right-hand side of \eqref{eq.carleman_pointwise} are continuous on $\Omega_0 ( f_\ast ) \cap \{ t = 0 \}$, and note also $\phi$ is supported on $( 0, \rho_0 ] \times \bar{\mi{D}}$.
As a result, when we integrate \eqref{eq.carleman_pointwise} (with respect to $g$), we can enlarge the domains of the corresponding integrals to $\Omega_0 ( f_\ast, \rho_\ast )$:
\begin{align}
\label{eql.carleman_1} &\int_{ \Omega_0 ( f_\ast, \rho_\ast ) } \mc{E} f^{ -p } | ( \Boxm + \sigma + \nablam_{ \mc{Y} } ) \phi |_\hv^2 - \lambda \int_{ ( \Omega^c_< \cup \Omega^c_> ) \cap \{ \rho > \rho_\ast \} } g^{ \alpha \beta } \nabla_\alpha P_\beta \\
\notag &\quad \geq \lambda \mc{C} \int_{ \Omega_0 ( f_\ast, \rho_\ast ) } \mc{E} f \rho^3 ( | \nablam_\rho \phi |_\hv^2 + | \Dv \phi |_\hv^2 ) + \frac{1}{8} \lambda^3 \int_{ \Omega_0 ( f_\ast, \rho_\ast ) } \mc{E} f^{ 2 p } | \phi |_\hv^2 \text{.}
\end{align}

Using the divergence theorem (with respect to $g$), we further expand
\begin{align}
\label{eql.carleman_2} - \lambda \int_{ ( \Omega^c_< \cup \Omega^c_> ) \cap \{ \rho > \rho_\ast \} } g^{ \alpha \beta } \nabla_\alpha P_\beta &= \lambda \int_{ \Omega_0 ( f_\ast ) \cap \{ \rho = \rho_\ast \} } P ( \rho \partial_\rho ) - \lambda \lim_{ \tau \searrow 0+ } \int_{ \Omega_0 ( f_\ast, \rho_\ast ) \cap \{ t = \tau \} } P ( \ms{T} ) \\
\notag &\qquad + \lambda \lim_{ \tau \nearrow 0- } \int_{ \Omega_0 ( f_\ast, \rho_\ast ) \cap \{ t = \tau \} } P ( \ms{T} ) \\
\notag &= D_1 + D_2 + D_3 \text{,}
\end{align}
where each of the integrals are with respect to the metric induced by $g$, and where
\begin{equation}
\label{eql.carleman_3} \ms{T} := - \frac{ \rho \cdot \Dv^\sharp t }{ \sqrt{ \gv ( \Dv^\sharp t, \Dv^\sharp t ) } } \text{.}
\end{equation}
In particular, $\rho \partial_\rho$ and $\ms{T}$ are the inward and future-pointing $g$-unit normals to $\Omega_0 ( f_\ast ) \cap \{ \rho = \rho_\ast \}$ and $\Omega_0 ( f_\ast, \rho_\ast ) \cap \{ t = 0 \}$, respectively.
Observe also that the terms $D_3$ and $D_2$ arise as integrals over the ``top" and ``bottom" boundaries of $\Omega^c_<$ and $\Omega^c_>$, respectively.
Moreover, note we do not obtain an integral over $\Omega_0 \cap \{ f = f_\ast \}$, since $\phi$ and $\nablam \phi$ are assumed to vanish there.

For $D_1$, we first observe from \eqref{eq.aads_metric} that the metric induced by $g$ on $\Omega_0 ( f_\ast ) \cap \{ \rho = \rho_\ast \}$ is precisely given by $\rho^{-n} \, d \gv |_{ \rho_\ast }$.
As a result, applying \eqref{eq.carleman_pointwise_P}, we obtain the inequality
\begin{align}
\label{eql.carleman_10} D_1 &= \lambda \int_{ \Omega_0 ( f_\ast ) \cap \{ \rho = \rho_\ast \} } \rho^{-n} P ( \rho \partial_\rho ) \, d \gv |_{ \rho_\ast } \\
\notag &\leq \mc{C}_b \lambda^3 \int_{ \Omega_0 ( f_\ast ) \cap \{ \rho = \rho_\ast \} } [ | \nablam_\rho ( \rho^{ -\kappa } \phi ) |_\hv^2 + | \nablam_{ \Dv^\sharp t } ( \rho^{ -\kappa } \phi ) |_\hv^2 + | \rho^{ -\kappa - 1 } \phi |_\hv^2 ] \, d \gv |_{ \rho_\ast } \text{.}
\end{align}
For the remaining boundary terms, we consider the definition \eqref{eq.carleman_pointwise_psi_PP} of $P$, and we notice that all the terms on the right-hand side of \eqref{eq.carleman_pointwise_psi_PP} are continuous at $\Omega ( f_\ast, \rho_\ast ) \cap \{ t = 0 \}$, except for the third term (due to the factor $\ms{T} v_\zeta$).
\footnote{This is a consequence of the formula \eqref{eql.carleman_S} for $v_\zeta$ and the fact that $f$ is only $C^2$ on $\Omega^c$; see Proposition \ref{thm.carleman_eta}.}
As a result of the above and of \eqref{eq.carleman_error_v}, we conclude that
\footnote{Again, the integrals on the right-hand side are with respect to the metric induced by $g$.}
\begin{align*}
D_2 + D_3 &\leq \lambda \lim_{ \tau \searrow 0+ } \int_{ \Omega_0 ( f_\ast, \rho_\ast ) \cap \{ t = \tau \} } | \ms{T} v_\zeta | | \psi |_\hv^2 + \lambda \lim_{ \tau \nearrow 0- } \int_{ \Omega_0 ( f_\ast, \rho_\ast ) \cap \{ t = \tau \} } | \ms{T} v_\zeta | | \psi |_\hv^2 \\
&\leq \mc{C}_v \lambda \int_{ \Omega_0 ( f_\ast, \rho_\ast ) \cap \{ t = 0 \} } \mc{E} \rho^3 | \phi |_\hv^2 \text{,}
\end{align*}
for some $\mc{C}_v > 0$ that depends on $\Upsilon$, where in the last step, we also used that $f \simeq \rho$ whenever $t = 0$.

Using a cutoff function $\chi (t)$ that is localized near $t = 0$ and applying the divergence theorem with respect to $g$, along with the assumption \eqref{eq.aads_time}, we have that
\begin{align*}
D_2 + D_3 &\leq \mc{C}_v \lambda \int_{ \Omega_0 ( f_\ast, \rho_\ast ) \cap \{ t = 0 \} } \mc{E} \rho^4 | \phi |_\hv^2 \, \cdot g ( \Dv^\sharp t, \ms{T} ) \\
&= \mc{C}_v \lambda \int_{ \Omega_0 ( f_\ast, \rho_\ast ) } \nabla_\alpha ( \mc{E} \rho^2 | \phi |_\hv^2 \, \cdot g^{ \alpha \beta } \nabla_\beta t ) \text{,}
\end{align*}
where the value of $\mc{C}_v$ may have changed, but still depends only on $\Upsilon$.
Expanding the right-hand side of the above and recalling \eqref{eq.aads_metric}, \eqref{eq.aads_strong_limits}, Proposition \ref{thm.aads_geom_limit}, and \eqref{eq.carleman_pointwise_weight}, we obtain
\begin{align}
\label{eql.carleman_13} D_2 + D_3 &\leq \lambda \int_{ \Omega_0 ( f_\ast, \rho_\ast ) } [ \mc{E} \, \mc{O} ( \rho^4 ) \, | \phi |_\hv | \nablam_{ \Dv^\sharp t } \phi |_\hv + \mc{E} | \Box t | \, \mc{O} ( \rho^2 ) \, | \phi |_\hv^2 + | \nabla_{ \Dv^\sharp t } \mc{E} | \, \mc{O} ( \rho^4 ) \, | \phi |_\hv^2 ] \\
\notag &\leq \lambda \int_{ \Omega_0 ( f_\ast, \rho_\ast ) } \mc{E} [ \, \mc{O} ( \rho^4 ) \, | \phi |_\hv | \nablam_{ \Dv^\sharp t } \phi |_\hv + \lambda | \, \mc{O} ( f \rho^3 ) \, | \phi |_\hv^2 ] \\
\notag &\leq \lambda \int_{ \Omega_0 ( f_\ast, \rho_\ast ) } \mc{E} \, \mc{O} ( \rho^6 ) \, | \Dvm \phi |_\hv^2 + \lambda^2 \int_{ \Omega_0 ( f_\ast, \rho_\ast ) } \mc{E} \, \mc{O} ( \rho^2 ) \, | \phi |_\hv^2 \text{.}
\end{align}

Combining the inequalities \eqref{eql.carleman_1}, \eqref{eql.carleman_2}, \eqref{eql.carleman_10}, \eqref{eql.carleman_13}; recalling the assumptions \eqref{eq.carleman_fstar} and \eqref{eq.carleman_lambda}; and recalling the definition \eqref{eq.carleman_pointwise_weight}; we arrive at the inequality
\begin{align*}
&\int_{ \Omega_{ t_0 } ( f_\ast, \rho_\ast ) } e^{ - \lambda p^{-1} f_{ t_0 }^p } f_{ t_0 }^{ n - 2 - p - 2 \kappa } | ( \Boxm + \sigma + \rho^2 \nablam_{ \mf{X}^\rho \partial_\rho + \mf{X} } ) \phi |_\hv^2 \\
&\qquad + \mc{C}_b \lambda^3 \int_{ \Omega_{ t_0 } ( f_\ast ) \cap \{ \rho = \rho_\ast \} } [ | \Dvm_{ \partial_\rho } ( \rho^{ - \kappa } \phi ) |_\hv^2 + | \Dv_{ \Dv^\sharp t } ( \rho^{ - \kappa } \phi ) |_\hv^2 + | \rho^{ - \kappa - 1 } \phi |_\hv^2 ] \, d \gv |_{ \rho_\ast } \\
&\quad \geq \lambda \int_{ \Omega_{ t_0 } ( f_\ast, \rho_\ast ) } e^{ - \lambda p^{-1} f_{ t_0 }^p } f_{ t_0 }^{ n - 2 - 2 \kappa } ( f_{ t_0 } \rho^3 | \Dvm_{ \partial_\rho } \phi |_\hv^2 + f_{ t_0 } \rho^3 | \Dv \phi |_\hv^2 + f^{ 2 p } | \phi |_\hv^2 ) \text{.}
\end{align*}
Finally, the desired Carleman estimate \eqref{eq.carleman} now follows from the above by taking the limit $\rho_\ast \searrow 0$ and then applying the monotone convergence theorem.

\appendix

\if\comp1

\section{Additional Details and Computations} \label{sec.comp}

This appendix contains additional proofs and computational details for readers' convenience.
Let us begin by collecting some preliminary computations that will be useful later:

\begin{lemma} \label{thm.comp_prelim_christoffel}
Let $( \mi{M}, g )$ be an FG-aAdS segment, and let $( U, \varphi )$ be a coordinate system on $\mi{I}$.
Then, the Christoffel symbols for $g$, with respect to $\varphi_\rho$-coordinates, are given by
\begin{align}
\label{eq.comp_prelim_christoffel} \Gamma^\rho_{ \rho \rho } = - \rho^{-1} \text{,} &\qquad \Gamma^\rho_{ \rho a } = 0 \text{,} \\
\notag \Gamma^a_{ \rho \rho } = 0 \text{,} &\qquad \Gamma^\rho_{ a b } = - \frac{1}{2} \mi{L}_\rho \gv_{ab} + \rho^{-1} \gv_{ab} \text{,} \\
\notag \Gamma^a_{ \rho b } = \frac{1}{2} \gv^{ac} \mi{L}_\rho \gv_{bc} - \rho^{-1} \delta^a_b \text{,} &\qquad \Gamma^c_{ a b } = \bar{\Gamma}^c_{ a b } \text{,}
\end{align}
where the symbols $\bar{\Gamma}^c_{ a b }$ are defined as in \eqref{eq.aads_vertical_christoffel}.
\end{lemma}

\begin{proof}
These follow from direct computations using \eqref{eq.aads_metric}.
\end{proof}

\subsection{Proof of Proposition \ref{thm.aads_vertical_connection}} \label{sec.aads_vertical_connection}

First, notice that the right-hand sides of \eqref{eq.aads_vertical_connection} define tensorial quantities, with respect to the indices $a_1 \dots a_k$ and $b_1 \dots b_l$.
In particular, by \eqref{eq.aads_vertical_christoffel}, the first formula of \eqref{eq.aads_vertical_connection} corresponds precisely to the $\Dv$-covariant derivative of $\ms{A}$.

Thus, we can define tensorial operators $\bar{\Dv}$ by setting, with respect to $\varphi$ and $\varphi_\rho$-coordinates,
\begin{equation}
\label{eql.aads_vertical_connection_1} \bar{\Dv}_X \ms{A}^{ a_1 \dots a_k }_{ b_1 \dots b_l } = X^\rho \, \bar{\Dv}_\rho \ms{A}^{ a_1 \dots a_k }_{ b_1 \dots b_l } + X^c \, \bar{\Dv}_c \ms{A}^{ a_1 \dots a_k }_{ b_1 \dots b_l } \text{,}
\end{equation}
for any vector field $X$ on $\mi{M}$.
In particular, \eqref{eql.aads_vertical_connection_1} defines bundle connections, since they are clearly $C^\infty ( \mi{M} )$-linear with respect to the derivative component, and since the formulas in \eqref{eq.aads_vertical_connection} imply
\[
\bar{\Dv}_X ( a \ms{A} ) = a \cdot \bar{\Dv}_X \ms{A} + X a \cdot \ms{A} \text{,} \qquad a \in C^\infty ( \mi{M} ) \text{.}
\]
Moreover, since the first part of \eqref{eq.aads_vertical_connection} is just the $\Dv$-derivative of $\ms{A}$, then \eqref{eq.aads_vertical_connection_extend} also follows.

Next, when the vector field $X$ on $\mi{M}$ is vertical, then $\bar{\Dv}_X$ coincides with $\Dv_X$ by \eqref{eq.aads_vertical_connection_extend}, and hence \eqref{eq.aads_vertical_connection_scalar}--\eqref{eq.aads_vertical_connection_compat} follow from the standard properties of Levi-Civita connections (see \cite{on:srg}).
Thus, to complete the proof, it remains to establish \eqref{eq.aads_vertical_connection_scalar}--\eqref{eq.aads_vertical_connection_compat} in the special case $X := \partial_\rho$.
\begin{itemize}
\item The identities \eqref{eq.aads_vertical_connection_scalar}--\eqref{eq.aads_vertical_connection_contract} follow from \eqref{eq.aads_vertical_connection} and the standard facts that the Lie derivative $\mi{L}_\rho$ coincides with the $\partial_\rho$-derivative for scalars, satisfies the Leibniz rule with respect to tensor products, and commutes with all (non-metric) contraction operations.

\item For the remaining property \eqref{eq.aads_vertical_connection_compat}, we apply \eqref{eq.aads_vertical_christoffel} and \eqref{eq.aads_vertical_connection} to obtain
\begin{align*}
\bar{\Dv}_\rho \gv_{ a b } &= \mi{L}_\rho \gv_{ a b } - \frac{1}{2} \gv^{ c d } \mi{L}_\rho \gv_{ a c } \cdot \gv_{ d b } - \frac{1}{2} \gv^{ c d } \mi{L}_\rho \gv_{ c b } \cdot \gv_{ a d } \\
&= \mi{L}_\rho \gv_{ a b } - \mi{L}_\rho \gv_{ a b } \text{,} \\
\bar{\Dv}_\rho \gv^{ a b } &= \mi{L}_\rho \gv^{ a b } + \frac{1}{2} \gv^{ a d } \mi{L}_\rho \gv_{ c d } \cdot \gv^{ c b } + \frac{1}{2} \gv^{ b d } \mi{L}_\rho \gv_{ c d } \cdot \gv^{ a c } \\
&= \mi{L}_\rho \gv^{ a b } + \gv^{ a c } \gv^{ b d } \mi{L}_\rho \gv_{ c d } \text{,}
\end{align*}
both of which clearly vanish identically.
\end{itemize}

\subsection{Proof of Proposition \ref{thm.aads_mixed_connection}} \label{sec.aads_mixed_connection}

The identities \eqref{eq.aads_mixed_connection_compat} follow immediately from \eqref{eq.aads_mixed_connection}---for instance, writing $g$ as $g \otimes 1$, with the ``$1$" viewed as a vertical scalar, we obtain
\[
\bar{\nabla}_X g = \bar{\nabla}_X ( g \otimes 1 ) = \nabla_X g \otimes 1 = 0 \text{.}
\]
The remaining parts of \eqref{eq.aads_mixed_connection_compat} are proved similarly.

For the product rule \eqref{eq.aads_mixed_connection_leibniz}, we first recall that $\mathbf{A}$ and $\mathbf{B}$ can be locally expressed as
\begin{equation}
\label{eql.aads_mixed_connection_1} \mathbf{A} = \sum_{ i = 1 }^m ( A_i \otimes \ms{A}_i ) \text{,} \qquad \mathbf{B} = \sum_{ j = 1 }^q ( B_j \otimes \ms{B}_j ) \text{,}
\end{equation}
where $A_1, \dots A_m, B_1, \dots, B_q$ are tensor fields on $\mi{M}$, and where $\ms{A}_1, \dots, \ms{A}_m, \ms{B}_1, \dots, \ms{B}_q$ are vertical tensor fields.
From \eqref{eql.aads_mixed_connection_1}, we can then (locally) expand
\begin{align}
\label{eql.aads_mixed_connection_2} \bar{\nabla}_X ( \mathbf{A} \otimes \mathbf{B} ) &= \sum_{ i = 1 }^m \sum_{ j = 1 }^q \bar{\nabla}_X [ ( A_i \otimes B_j ) \otimes ( \ms{A}_i \otimes \ms{B}_j ) ] \\
\notag &= \sum_{ i = 1 }^m \sum_{ j = 1 }^q [ \nabla_X ( A_i \otimes B_j ) \otimes ( \ms{A}_i \otimes \ms{B}_j ) + ( A_i \otimes B_j ) \otimes \bar{\Dv}_X ( \ms{A}_i \otimes \ms{B}_j ) ] \text{.}
\end{align}
The identity \eqref{eq.aads_mixed_connection_leibniz} now follows directly from \eqref{eq.aads_vertical_connection_leibniz}, \eqref{eql.aads_mixed_connection_1}, and \eqref{eql.aads_mixed_connection_2}.

\subsection{Proof of Proposition \ref{thm.aads_vertical_curv}} \label{sec.aads_vertical_curv}

For conciseness, we give the proof only for vertical vector fields and $1$-forms, as the derivations are similar for higher-order fields.
(Moreover, it is clear that both expressions in \eqref{eq.aads_vertical_curv} vanish identically when $\ms{A}$ is a scalar.)
Throughout, we let $\ms{X}$ denote a vertical vector field, and we let $\ms{H}$ be a vertical $1$-form.
In addition, we let $\smash{ \Gamma^\alpha_{ \gamma \beta } }$ denote the Christoffel symbols for $g$ with respect to $\varphi_\rho$-coordinates, and we let $\bar{\Gamma}^a_{ c b }$ and $\bar{\Gamma}^a_{ \rho b }$ be as defined in \eqref{eq.aads_vertical_christoffel}.

First, from Proposition \ref{thm.aads_mixed_connection}, the definition \eqref{eq.aads_mixed_curv}, and \eqref{eq.comp_prelim_christoffel}, we obtain
\begin{align}
\label{eql.aads_vertical_curv_1} \bar{R}_{ a b } [ \ms{H} ]_c &= \bar{\nabla}_{ a b } \ms{H}_c - \bar{\nabla}_{ b a } \ms{H}_c \\
\notag &= [ \partial_a ( \bar{\nabla}_b \ms{H}_c ) - \Gamma^\delta_{ a b } \bar{\nabla}_\delta \ms{H}_c - \bar{\Gamma}^d_{ a c } \bar{\nabla}_b \ms{H}_d ] - [ \partial_b ( \bar{\nabla}_a \ms{H}_c ) - \Gamma^\delta_{ b a } \bar{\nabla}_\delta \ms{H}_c - \bar{\Gamma}^d_{ b c } \bar{\nabla}_a \ms{H}_d ] \\
\notag &= - \partial_a ( \bar{\Gamma}^d_{ b c } \ms{H}_d ) - \bar{\Gamma}^d_{ a c } ( \partial_b \ms{H}_d - \bar{\Gamma}^e_{ b d } \ms{H}_e ) + \partial_b ( \bar{\Gamma}^d_{ a c } \ms{H}_d ) + \bar{\Gamma}^d_{ b c } ( \partial_a \ms{H}_d - \bar{\Gamma}^e_{ a d } \ms{H}_e ) \\
\notag &= ( \partial_b \bar{\Gamma}^e_{ a c } - \partial_a \bar{\Gamma}^e_{ b c } + \bar{\Gamma}^d_{ a c } \bar{\Gamma}^e_{ b d } - \bar{\Gamma}^d_{ b c } \bar{\Gamma}^e_{ a d } ) \ms{H}_e \text{.}
\end{align}
Moreover, we know from \eqref{eq.aads_vertical_christoffel} and standard differential geometric formulas that
\begin{equation}
\label{eql.aads_vertical_curv_2} \partial_b \bar{\Gamma}^e_{ a c } - \partial_a \bar{\Gamma}^e_{ b c } + \bar{\Gamma}^d_{ a c } \bar{\Gamma}^e_{ b d } - \bar{\Gamma}^d_{ b c } \bar{\Gamma}^e_{ a d } = - \ms{R}^e{}_{ c a b } \text{.}
\end{equation}
Thus, the first part of \eqref{eq.aads_vertical_curv}---when $\ms{A}$ is a $1$-form---follows immediately from \eqref{eql.aads_vertical_curv_1} and \eqref{eql.aads_vertical_curv_2}.

Moreover, an analogous calculation for $\ms{X}$ yields
\begin{align}
\label{eql.aads_vertical_curv_3} \bar{R}_{ a b } [ \ms{X} ]^c &= [ \partial_a ( \bar{\nabla}_b \ms{X}^c ) - \Gamma^\delta_{ a b } \bar{\nabla}_\delta \ms{X}^c + \bar{\Gamma}^c_{ a d } \bar{\nabla}_b \ms{X}^d ] - [ \partial_b ( \bar{\nabla}_a \ms{X}^c ) - \Gamma^\delta_{ b a } \bar{\nabla}_\delta \ms{X}^c + \bar{\Gamma}^c_{ b d } \bar{\nabla}_a \ms{X}^d ] \\
\notag &= ( \partial_a \bar{\Gamma}^c_{ b e } - \partial_b \bar{\Gamma}^c_{ a e } + \bar{\Gamma}^c_{ a d } \bar{\Gamma}^d_{ b e } - \bar{\Gamma}^c_{ b d } \bar{\Gamma}^d_{ a e } ) \ms{X}^e \\
\notag &= \ms{R}^c{}_{ e a b } \ms{X}^e \text{,}
\end{align}
which is again the first formula in \eqref{eq.aads_vertical_curv}.

The second part of \eqref{eq.aads_vertical_curv} is proved similarly.
By Proposition \ref{thm.aads_mixed_connection}, \eqref{eq.aads_mixed_curv}, and \eqref{eq.comp_prelim_christoffel}, we have
\begin{align}
\label{eql.aads_vertical_curv_10} \bar{R}_{ \rho a } [ \ms{H} ]_c &= [ \partial_\rho ( \bar{\nabla}_a \ms{H}_c ) - \Gamma^\delta_{ \rho a } \bar{\nabla}_\delta \ms{H}_c - \bar{\Gamma}^d_{ \rho c } \bar{\nabla}_a \ms{H}_d ] - [ \partial_a ( \bar{\nabla}_\rho \ms{H}_c ) - \Gamma^\delta_{ a \rho } \bar{\nabla}_\delta \ms{H}_c - \bar{\Gamma}^d_{ a c } \bar{\nabla}_\rho \ms{H}_d ] \\
\notag &= - \partial_\rho ( \bar{\Gamma}^d_{ a c } \ms{H}_d ) - \bar{\Gamma}^d_{ \rho c } ( \partial_a \ms{H}_d - \bar{\Gamma}^e_{ a d } \ms{H}_e ) + \partial_a ( \bar{\Gamma}^d_{ \rho c } \ms{H}_d ) + \bar{\Gamma}^d_{ a c } ( \partial_\rho \ms{H}_d - \bar{\Gamma}^e_{ \rho d } \ms{H}_e ) \\
\notag &= ( \partial_a \bar{\Gamma}^e_{ \rho c } - \partial_\rho \bar{\Gamma}^e_{ a c } + \bar{\Gamma}^d_{ \rho c } \bar{\Gamma}^e_{ a d } - \bar{\Gamma}^d_{ a c } \bar{\Gamma}^e_{ \rho d } ) \ms{H}_e \text{.}
\end{align}
In addition, a direct, though tedious, calculation yields the identity
\begin{equation}
\label{eql.aads_vertical_curv_11} ( \partial_a \bar{\Gamma}^e_{ \rho c } - \partial_\rho \bar{\Gamma}^e_{ a c } + \bar{\Gamma}^d_{ \rho c } \bar{\Gamma}^e_{ a d } - \bar{\Gamma}^d_{ a c } \bar{\Gamma}^e_{ \rho d } ) = - \frac{1}{2} \gv^{ b e } ( \Dv_c \mi{L}_\rho \gv_{ a b } - \Dv_b \mi{L}_\rho \gv_{ a c } ) \text{.}
\end{equation}
The equations \eqref{eql.aads_vertical_curv_10} and \eqref{eql.aads_vertical_curv_11} together yield the second part of \eqref{eq.aads_vertical_curv}, at least when $\ms{A}$ is a $1$-form.
The vector field case can also be established using an analogous computation:
\begin{align}
\label{eql.aads_vertical_curv_12} \bar{R}_{ \rho a } [ \ms{X} ]^c &= [ \partial_\rho ( \bar{\nabla}_a \ms{X}^c ) + \bar{\Gamma}^c_{ \rho d } \bar{\nabla}_a \ms{X}^d ] - [ \partial_a ( \bar{\nabla}_\rho \ms{X}^c ) + \bar{\Gamma}^c_{ a d } \bar{\nabla}_\rho \ms{X}^d ] \\
\notag &= ( \partial_\rho \bar{\Gamma}^c_{ a e } - \partial_a \bar{\Gamma}^c_{ \rho e } + \bar{\Gamma}^c_{ \rho d } \bar{\Gamma}^d_{ a e } - \bar{\Gamma}^c_{ a d } \bar{\Gamma}^d_{ \rho e } ) \ms{X}^e \\
\notag &= \frac{1}{2} \gv^{ b c } ( \Dv_e \mi{L}_\rho \gv_{ a b } - \Dv_b \mi{L}_\rho \gv_{ a e } ) \ms{X}^e \text{.}
\end{align}

\subsection{Proof of Proposition \ref{thm.carleman_f}} \label{sec.carleman_f}

It suffices to prove the proposition with $t_0 := 0$ (that is, with $f_{ t_0 }$ replaced by $f$), since we can simply replace our global time function $t$ by $t - t_0$.
In addition, throughout the proof, we let $( U, \varphi )$ denote an arbitrary coordinate system on $\mi{I}$.

First, we apply \eqref{eq.aads_metric} and \eqref{eq.carleman_f} to expand $\nabla^\sharp f$ with respect to $\varphi_\rho$-coordinates:
\begin{align}
\label{eql.carleman_f_1} \nabla^\sharp f &= g^{ \rho \rho } \partial_\rho f \partial_\rho + \rho^2 \gv^{ab} \partial_b f \partial_a \\
\notag &= \rho^2 \cdot [ \eta (t) ]^{-1} \partial_\rho - \rho^3 \cdot [ \eta (t) ]^{-2} \eta' (t) \cdot \gv^{ab} \partial_b t \partial_a \\
\notag &= \rho f \partial_\rho - \rho f^2 \cdot \eta' (t) \cdot \gv^{ab} \partial_b t \partial_a \text{.}
\end{align}
The first part of \eqref{eq.carleman_f_grad} now follows from \eqref{eql.carleman_f_1}.
Also, continuing from \eqref{eql.carleman_f_1}, we see that
\begin{align*}
g ( \nabla^\sharp f, \nabla^\sharp f ) &= \rho f \partial_\rho f - \rho f^2 \cdot \eta' (t) \cdot \gv^{ab} \partial_b t \partial_a f \\
&= f^2 + f^4 \cdot [ \eta' (t) ]^2 \cdot \gv^{ab} \partial_a t \partial_b t \text{,}
\end{align*}
which immediately implies the second identity in \eqref{eq.carleman_f_grad}.

Next, we move to the identities in \eqref{eq.carleman_f_hessian}.
Letting $\smash{ \Gamma^\alpha_{ \gamma \beta } }$ denote the Christoffel symbols for $g$ with respect to $\varphi_\rho$-coordinates, and applying \eqref{eq.carleman_f} and \eqref{eq.comp_prelim_christoffel}, we compute
\begin{align*}
\nabla_{ \rho \rho } f &= \partial_{ \rho \rho } f - \Gamma^\rho_{ \rho \rho } \partial_\rho f \\
&= \rho^{-2} f \text{,} \\
\nabla_{ \rho a } f &= \partial_{ \rho a } f - \Gamma^b_{ \rho a } \partial_b f \\
&= - \rho^{-2} f^2 \cdot \eta' (t) \cdot \Dv_a t + \rho^{-1} f^2 \cdot \eta' (t) \cdot \Dv_b t \left( \frac{1}{2} \gv^{bc} \mi{L}_\rho \gv_{ac} - \rho^{-1} \delta_a^b \right) \\
&= - 2 \rho^{-2} f^2 \cdot \eta' (t) \cdot \Dv_a t + \frac{1}{2} \rho^{-1} f^2 \cdot \eta' (t) \cdot \gv^{ a b } \Dv_b t \mi{L}_\rho \gv_{ a c } \\
\nabla_{ a b } f &= \partial_{ a b } f - ( \Gamma^\rho_{ a b } \partial_\rho f + \Gamma^c_{ab} \partial_c f) \\
&= \Dv_{ a b } f + \rho^{-1} f \left( \frac{1}{2} \mi{L}_\rho \gv_{ab} - \rho^{-1} \gv_{ab} \right) \\
&= 2 \rho^{-2} f^3 \cdot [ \eta' (t) ]^2 \cdot \Dv_a t \Dv_b t - \rho^{-1} f^2 \cdot \eta'' (t) \cdot \Dv_a t \Dv_b t - \rho^{-1} f^2 \cdot \eta' (t) \cdot \Dv_{ a b } t \\
&\qquad + \frac{1}{2} \rho^{-1} f \cdot \mi{L}_\rho \gv_{ a b } - \rho^{-2} f \cdot \gv_{ a b } \text{.}
\end{align*}
The above are precisely the desired identities \eqref{eq.carleman_f_hessian}.
Finally, for $\Box f$, we see from \eqref{eq.aads_metric} that
\[
\Box f = \rho^2 \nabla_{ \rho \rho } f + \rho^2 \gv^{ a b } \nabla_{ a b } f \text{.}
\]
Combining the above with \eqref{eq.carleman_f_hessian} yields the last equation in \eqref{eq.carleman_f_grad}.

\subsection{Proof of Proposition \ref{thm.carleman_E}} \label{sec.carleman_E}

First, standard results imply that on a small enough neighborhood $U_p$ of $p$, there exist vector fields $\mf{E}_1, \dots, \mf{E}_{ n - 1 }$ on $U_p$ satisfying
\begin{equation}
\label{eql.carleman_E_0} \gm ( \mf{E}_X, \mf{E}_Y ) = \delta_{ X Y } \text{,} \qquad \mf{E}_X t = 0 \text{,} \qquad 1 \leq X, Y < n \text{.}
\end{equation}
By \eqref{eq.aads_metric_limit} and standard differential geometric procedures, the vector fields $\mf{E}_1, \dots, \mf{E}_{ n - 1 }$ can be extended to vertical vector fields $\ms{E}_1, \dots, \ms{E}_{ n - 1 }$ on $( 0, \rho_0 ] \times U_p$ satisfying
\begin{equation}
\label{eql.carleman_E_1} \gv ( \ms{E}_X, \ms{E}_Y ) = \delta_{ X Y } \text{,} \qquad \ms{E}_X t = 0 \text{,} \qquad \ms{E}_X \rightarrow^0 \mf{E}_X \text{,} \qquad 1 \leq X, Y < n \text{.}
\end{equation}
The vector fields $E_X := \rho \ms{E}_X$, for all $1 \leq X < n$, satisfy the desired properties.

\subsection{Proof of Proposition \ref{thm.carleman_frame}} \label{sec.carleman_frame}

From \eqref{eq.aads_time}, \eqref{eq.carleman_eta}, and \eqref{eq.carleman_omega}, we obtain that
\begin{equation}
\label{eql.carleman_frame_0} 1 + f_{ t_0 }^2 [ \eta' ( t - t_0 ) ]^2 \cdot \gv ( \Dv^\sharp t, \Dv^\sharp t ) > 0 \text{,}
\end{equation}
as long as $f_{ t_0 }$ is sufficiently small with respect to $t$, $b$, and $c$.
In particular, this immediately implies that the vector fields \eqref{eq.carleman_NV} are well-defined for these values of $f_{ t_0 }$.

That the frames $( N_{ t_0 }, V_{ t_0 }, E_1, \dots, E_{ n-1 } )$ mutually orthonormal follows from direct computations using \eqref{eq.aads_metric} and the definitions \eqref{eq.carleman_E}, \eqref{eq.carleman_NV}.
Furthermore, \eqref{eq.carleman_f_grad} and \eqref{eq.carleman_NV} imply that $N_{ t_0 }$ and $\nabla^\sharp f_{ t_0 }$ point in the same direction.
As a result, $N_{ t_0 }$ is normal to the level sets of $f_{ t_0 }$, and it follows that $V_{ t_0 }, E_1, \dots, E_{n-1}$ must be tangent to the level sets of $f_{ t_0 }$.

The second identity of \eqref{eq.carleman_f_grad}, along with \eqref{eql.carleman_frame_0}, implies that $\nabla^\sharp f_{ t_0 }$---and hence $N_{ t_0 }$---is spacelike.
Moreover, from \eqref{eq.carleman_E}, we conclude that $E_1, \dots, E_{n-1}$ must also be spacelike.
Finally, by orthogonality, we obtain that the remaining element $V_{ t_0 }$ must be timelike.

\subsection{Proof of Lemma \ref{thm.carleman_pseudoconvex_pi}} \label{sec.carleman_pseudoconvex_pi}

We begin by defining the quantities
\begin{equation}
\label{eql.carleman_pseudoconvex_pi_0} \ms{T} := | \gv ( \Dv^\sharp t, \Dv^\sharp t ) |^{ - \frac{1}{2} } \Dv^\sharp t \text{,} \qquad \ms{E}_X := \rho^{-1} E_X \text{,} \qquad 1 \leq X < n \text{.}
\end{equation}
For brevity, we also adopt the abbreviation
\begin{equation}
\label{eql.carleman_pseudoconvex_pi_1} \gv^{tt} := \gv ( \Dv^\sharp t, \Dv^\sharp t ) = - dt^2 ( \ms{T}, \ms{T} ) \text{,}
\end{equation}
Recalling \eqref{eq.aads_strong_limits} and Proposition \ref{thm.aads_geom_limit}, we obtain
\begin{equation}
\label{eql.carleman_pseudoconvex_pi_2} \gv \rightarrow^3 \gm \text{,} \qquad \Dv^\sharp t \rightarrow^0 \Dm^\sharp t \text{,} \qquad \rho^{-1} \mi{L}_\rho \gv \rightarrow^1 2 \gs \text{,} \qquad \Dv^2 t \rightarrow^0 \Dm^2 t \text{.}
\end{equation}
Similarly, from \eqref{eql.carleman_pseudoconvex_pi_0} and the definitions of $\mf{E}_1, \dots, \mf{E}_{ n - 1 }$, we have
\begin{equation}
\label{eql.carleman_pseudoconvex_pi_3} \ms{T} \rightarrow^0 \mf{T} \text{,} \qquad \ms{E}_X \rightarrow^0 \mf{E}_X \text{,} \qquad 1 \leq X < n \text{.}
\end{equation}

Now, from direct (but long) computations using \eqref{eq.carleman_f}, \eqref{eq.carleman_f_hessian}, \eqref{eq.carleman_NV}, and \eqref{eql.carleman_pseudoconvex_pi_0}, we obtain
\begin{align}
\label{eql.carleman_pseudoconvex_pi_10} \nabla_{ V V } f &= \rho^2 ( - \gv^{tt} )^{-1} \{ 1 + f^2 [ \eta' (t) ]^2 \gv^{tt} \}^{-1} \\
\notag &\qquad \cdot \{ f^2 [ \eta' (t) ]^2 ( \gv^{tt} )^2 \cdot \nabla^2 f ( \partial_\rho, \partial_\rho ) + 2 f \eta' (t) \gv^{tt} \cdot \nabla^2 f ( \partial_\rho, \Dv^\sharp t ) + \nabla^2 f ( \Dv^\sharp t, \Dv^\sharp t ) \} \\
\notag &= f + \left[ \frac{1}{2} f^2 \eta (t) \cdot \mi{L}_\rho \gv - \rho f^2 \eta' (t) \cdot \Dv^2 t - \rho f^2 \eta'' (t) \cdot dt^2 \right] ( \ms{T}, \ms{T} ) + \mc{O} ( \rho f^3 ) \text{,} \\
\notag \nabla_{ V E_X } f &= \rho^2 ( - \gv^{tt} )^{ - \frac{1}{2} } \{ 1 + f^2 [ \eta' (t) ]^2 \gv^{tt} \}^{ - \frac{1}{2} } \cdot \{ f \eta' (t) \gv^{tt} \cdot \nabla^2 f ( \partial_\rho, \ms{E}_X ) + \nabla^2 f ( \Dv^\sharp t, \ms{E}_X ) \} \\
\notag &= \left[ \frac{1}{2} f^2 \eta (t) \cdot \mi{L}_\rho \gv - \rho f^2 \eta' (t) \cdot \Dv^2 t \right] ( \ms{T}, \ms{E}_X ) + \mc{O} ( \rho f^3 ) \text{,} \\
\notag \nabla_{ E_X E_Y } f &= \rho^2 \cdot \nabla^2 f ( \ms{E}_X, \ms{E}_Y ) \\
\notag &= - f \cdot \delta_{ X Y } + \left[ \frac{1}{2} f^2 \eta (t) \cdot \mi{L}_\rho \gv - \rho f^2 \eta' \cdot \Dv^2 t \right] ( \ms{E}_X, \ms{E}_Y ) \text{.}
\end{align}
for any $1 \leq X, Y < n$.
(For the error terms $\mc{O} ( \rho f^3 )$, we also recall Remark \ref{rmkl.carleman_bounded}.)
From \eqref{eql.carleman_pseudoconvex_pi_2}--\eqref{eql.carleman_pseudoconvex_pi_10}, along with observation that the $\mi{L}_\rho$-derivative of the quantities in \eqref{eql.carleman_pseudoconvex_pi_2} and \eqref{eql.carleman_pseudoconvex_pi_3} exist and are uniformly bounded on $\Omega^c_< \cup \Omega^c_>$, we then obtain
\begin{align}
\label{eql.carleman_pseudoconvex_pi_11} \nabla_{ V V } f &= f + \rho f^2 [ - \eta'' (t) \cdot dt^2 - \eta' (t) \cdot \Dm^2 t + \eta (t) \cdot \gs ] ( \mf{T}, \mf{T} ) + \mc{O} ( \rho f^3 ) \text{,} \\
\notag \nabla_{ V E_X } f &= \rho f^2 [ - \eta' (t) \cdot \Dm^2 t + \eta (t) \cdot \gs ] ( \mf{T}, \mf{E}_X ) + \mc{O} ( \rho f^3 ) \text{,} \\
\notag \nabla_{ E_X E_Y } f &= - f \cdot \delta_{ X Y } + \rho f^2 [ - \eta' \cdot \Dm^2 t + \eta (t) \cdot \gs ] ( \mf{E}_X, \mf{E}_Y ) + \mc{O} ( \rho f^3 ) \text{.} 
\end{align}

The remaining components of $\nabla^2 f$ are treated similarly, but one can be less precise.
More specifically, using \eqref{eq.carleman_f}, \eqref{eq.carleman_f_hessian}, \eqref{eq.carleman_NV}, and \eqref{eql.carleman_pseudoconvex_pi_0}, we find that
\begin{align}
\label{eql.carleman_pseudoconvex_pi_20} \nabla_{ N N } f &= \rho^2 \{ 1 + f^2 [ \eta' (t) ]^2 \gv^{tt} \}^{ - \frac{1}{2} } \\
\notag &\qquad \cdot \{ \nabla^2 f ( \partial_\rho, \partial_\rho ) - 2 f \eta' (t) \cdot \nabla^2 f ( \partial_\rho, \Dv^\sharp t ) + f^2 [ \eta' (t) ]^2 \nabla^2 f ( \Dv^\sharp t, \Dv^\sharp t ) \} \\
\notag &= f + \mc{O} ( f^3 ) \text{,} \\
\notag \nabla_{ N V } f &= \rho^2 ( \gv^{tt} )^{ - \frac{1}{2} } \{ 1 + f^2 [ \eta' (t) ]^2 \gv^{tt} \}^{-1} [ f \eta' (t) \gv^{tt} \cdot \nabla^2 f ( \partial_\rho, \partial_\rho ) - f \eta' (t) \cdot \nabla^2 f ( \Dv^\sharp t, \Dv^\sharp t ) ] \\
\notag &\qquad + \rho^2 ( \gv^{tt} )^{ - \frac{1}{2} } \{ 1 + f^2 [ \eta' (t) ]^2 \gv^{tt} \}^{-1} \{ 1 - f^2 [ \eta' (t) ]^2 \gv^{tt} \} \cdot \nabla^2 f ( \partial_\rho, \Dv^\sharp t ) \\
\notag &= \mc{O} ( \rho f^3 ) \text{,} \\
\notag \nabla_{ N E_X } f &= \rho^2 \{ 1 + f^2 [ \eta' (t) ]^2 \gv^{tt} \}^{ - \frac{1}{2} } \cdot [ \nabla^2 f ( \partial_\rho, \ms{E}_X ) - f \eta' (t) \cdot \nabla^2 f ( \Dv^\sharp t, \ms{E}_X ) ] \\
\notag &= \mc{O} ( \rho f^3 ) \text{,}
\end{align}
for any $1 \leq X < n$.
Finally, the identities \eqref{eq.carleman_pseudoconvex_pi} follow from computations using \eqref{eql.carleman_pi}, \eqref{eql.carleman_pseudoconvex_pi_11}, and \eqref{eql.carleman_pseudoconvex_pi_20}.
(Here, we also used that $N, V, E_1, \dots, E_{ n - 1 }$ are $g$-orthonormal, that $\mf{T}, \mf{E}_1, \dots, \mf{E}_{ n - 1 }$ are $\gm$-orthonormal, and that $dt^2 ( \mf{T}, \mf{E}_X )$, $dt^2 ( \mf{E}_X, \mf{E}_Y )$ vanish for $1 \leq X, Y < n$.)

\subsection{Proof of Lemma \ref{thm.carleman_asymp}} \label{sec.carleman_asymp}

The formulas \eqref{eq.carleman_asymp} follow from applying \eqref{eq.aads_strong_limits} and Proposition \ref{thm.aads_geom_limit} to \eqref{eq.carleman_f}, \eqref{eq.carleman_f_grad} and \eqref{eq.carleman_NV}.
(Also, recall the assumption \eqref{eq.carleman_fstar} and Remark \ref{rmkl.carleman_bounded}.)
Similarly, \eqref{eq.carleman_frame_asymp} follows from applying \eqref{eq.aads_strong_limits} and Proposition \ref{thm.aads_geom_limit} to Propositions \ref{thm.carleman_E} and \ref{thm.carleman_frame}.

\subsection{Proof of Lemma \ref{thm.carleman_error_v}} \label{sec.carleman_error_v}

First, we use \eqref{eq.carleman_f_grad}, \eqref{eql.carleman_pi}, and \eqref{eql.carleman_S} to expand
\begin{align}
\label{eql.carleman_error_v_1} v_\zeta &= f^{ n - 3 } ( f + f^2 \rho \zeta ) + \frac{ n - 3 }{2} f^{ n - 4 } g^{ \alpha \beta } \nabla_\alpha f \nabla_\beta f + \frac{1}{2} f^{ n - 3 } \Box f \\
\notag &= \frac{ n - 1 }{2} [ \eta' (t) ]^2 \, \gv ( \Dv^\sharp t, \Dv^\sharp t ) \cdot f^n + \frac{1}{4} \trace{\gv} \mi{L}_\rho \gv \cdot f^{ n - 2 } \rho \\
\notag &\qquad + \left( \zeta - \frac{1}{2} \eta'' (t) \, \gv ( \Dv^\sharp t, \Dv^\sharp t ) - \frac{1}{2} \eta' (t) \, \trace{\gv} \Dv^2 t \right) f^{ n - 1 } \rho \text{.}
\end{align}
Applying \eqref{eq.aads_strong_limits}, Proposition \ref{thm.aads_geom_limit}, and \eqref{eq.carleman_f} to \eqref{eql.carleman_error_v_1} yields the first bound in \eqref{eq.carleman_error_v}.
(In particular, note $| \trace{\gv} \mi{L}_\rho \gv | = \mc{O} ( \rho )$ in the second term on the right-hand side of \eqref{eql.carleman_error_v_1}.)

For the remaining inequalities in \eqref{eq.carleman_error_v}, let us first fix a compact coordinate system $( U, \varphi )$ of $\mi{I}$.
We now index with respect to $\varphi$- and $\varphi_\rho$-coordinates, and we let $\smash{ \Gamma^\alpha_{ \gamma \beta } }$ denote the Christoffel symbols for $g$ with respect to $\varphi_\rho$-coordinates.
From \eqref{eq.carleman_f}, we see that
\begin{equation}
\label{eql.carleman_error_v_2} | \partial_\rho f | \lesssim \rho^{-1} f \text{,} \qquad | \partial_a f | \lesssim_\varphi \rho^{-1} f^2 \text{.}
\end{equation}
Furthermore, by differentiating \eqref{eql.carleman_error_v_1} and then recalling Definition \ref{def.aads_strong}, Proposition \ref{thm.aads_geom_limit}, \eqref{eq.carleman_f}, and \eqref{eql.carleman_error_v_2}, we obtain (after a tedious process) the estimates
\begin{align}
\label{eql.carleman_error_v_3} | \partial_a v_\zeta | \lesssim_\varphi \mc{O} ( \rho^{-1} f^{ n + 1 } ) \text{,} &\qquad | \partial_\rho v_\zeta | \lesssim \mc{O} ( \rho^{-1} f^n ) \text{,} \\
\notag | \partial_{ a b } v_\zeta | \lesssim_\varphi \mc{O} ( \rho^{-2} f^{ n + 2 } ) \text{,} &\qquad | \partial_{ \rho \rho } v_\zeta | \lesssim \mc{O} ( \rho^{-2} f^n ) \text{.}
\end{align}
In particular, this proves the second and third estimates of \eqref{eq.carleman_error_v}.

Finally, we apply \eqref{eq.aads_metric}, \eqref{eq.aads_vertical_christoffel}, and \eqref{eq.comp_prelim_christoffel} in order to expand
\begin{align}
\label{eql.carleman_error_v_4} | \Box v_\zeta | &\leq \rho^2 | \partial_{ \rho \rho } v_\zeta - \Gamma^\rho_{ \rho \rho } \partial_\rho v_\zeta | + \rho^2 | \gv^{ a b } ( \partial_{ a b } v_\zeta - \Gamma^c_{ a b } \partial_c v_\zeta - \Gamma^\rho_{ a b } \partial_\rho v_\zeta ) | \\
\notag &\lesssim_\varphi \mc{O} ( \rho^2 ) \, | \partial_{ \rho \rho } v_\zeta | + \mc{O} ( \rho ) \, | \partial_\rho v_\zeta | + \mc{O} ( \rho^2 ) \, \sup_{ a, b } | \partial_{ a b } v_\zeta | + \mc{O} ( \rho^2 ) \, \sup_a | \partial_a v_\zeta | \text{.}
\end{align}
The fourth part of \eqref{eq.carleman_error_v} now follows from \eqref{eql.carleman_error_v_3}, \eqref{eql.carleman_error_v_4}, and the fact that $\mi{D}^c$ (see Remark \ref{rmkl.carleman_bounded}) can be covered by a finite number of compact coordinate systems.

\subsection{Proof of Lemma \ref{thm.carleman_error_R}} \label{sec.carleman_error_R}

Fix a compact coordinate system $( U, \varphi )$ on $\mi{I}$, and assume all indices are with respect to $\varphi$- and $\varphi_\rho$-coordinates.
By \eqref{eq.aads_strong_limits} and Proposition \ref{thm.aads_geom_limit}, we have
\begin{equation}
\label{eql.carleman_error_R_1} | \gv^{-1} |_{ \varphi, 0 } \lesssim_\varphi \mc{O} ( 1 ) \text{,} \qquad | \ms{R} |_{ \varphi, 0 } \lesssim_\varphi \mc{O} ( 1 )  \text{,} \qquad | \Dv \mi{L}_\rho \gv |_{ \varphi, 0 } \lesssim_\varphi \mc{O} ( \rho ) \text{.}
\end{equation}
From Proposition \ref{thm.aads_vertical_curv} and \eqref{eql.carleman_error_R_1}, we then conclude that
\begin{equation}
\label{eql.carleman_error_R_2} | \bar{R}_{ab} [ \ms{A} ] |_{ \varphi, 0 } \lesssim_\varphi ( k + l ) \mc{O} ( 1 ) \, | \ms{A} |_{ \varphi, 0 } \text{,} \qquad | \bar{R}_{ \rho a } [ \ms{A} ] |_{ \varphi, 0 } \lesssim_\varphi ( k + l ) \mc{O} ( \rho ) \, | \ms{A} |_{ \varphi, 0 } \text{.}
\end{equation}
The desired estimates \eqref{eq.carleman_error_R} now follow immediately from \eqref{eq.carleman_frame_asymp}, \eqref{eql.carleman_error_R_2}, and the observation that $\mi{D}^c$ can be covered by a finite number of compact coordinate systems.

\subsection{Proof of Lemma \ref{thm.carleman_error_h}} \label{sec.carleman_error_h}

Let $( U, \varphi )$ be a compact coordinate system on $\mi{I}$, and assume all indices are with respect to $\varphi$- and $\varphi_\rho$-coordinates.
Furthermore, let $\smash{ \Gamma^\alpha_{ \gamma \beta } }$ denote the Christoffel symbols for $g$ with respect to $\varphi_\rho$-coordinates, and let $\bar{\Gamma}^a_{ c b }$ and $\bar{\Gamma}^a_{ \rho b }$ be as defined in \eqref{eq.aads_vertical_christoffel}.

\begin{lemma} \label{thm.carleman_error_t}
The following estimates hold for the first and second derivatives of $t$:
\begin{equation}
\label{eq.carleman_error_t2} | \Dv t |_{ \varphi, 0 } \lesssim_\varphi \mc{O} ( 1 ) \text{,} \qquad | \Dv^2 t |_{ \varphi, 0 } \lesssim_\varphi \mc{O} ( 1 ) \text{,} \qquad | \Dvm_\rho \Dv t |_{ \varphi, 0 } \lesssim_\varphi \mc{O} ( \rho ) \text{.}
\end{equation}
In addition, the third derivatives of $t$ satisfy
\begin{equation}
\label{eq.carleman_error_t3} | \Dv^3 t |_{ \varphi, 0 } \lesssim_\varphi \mc{O} ( 1 ) \text{,} \qquad | ( \Dvm_\rho )^2 \Dv t |_{ \varphi, 0 } \lesssim_\varphi \mc{O} ( 1 ) \text{.} 
\end{equation}
\end{lemma}

\begin{proof}
The first two parts of \eqref{eq.carleman_error_t2} and the first part of \eqref{eq.carleman_error_t3} follow from \eqref{eq.aads_geom_limit_deriv}, which implies
\[
\Dv t \rightarrow^0 \Dm t \text{,} \qquad \Dv^2 t \rightarrow^0 \Dm^2 t \text{,} \qquad \Dv^3 t \rightarrow^0 \Dm^3 t \text{.}
\]
Also, recalling Definition \ref{def.aads_vertical_connection} and the fact that $\Dv t$ is $\rho$-independent, we see that
\footnote{To clarify, by $\Dvm_\rho \Dv_a t$, we mean the $a$-component of the vertical $1$-form $\Dvm_\rho \Dv t$.}
\begin{equation}
\label{eql.carleman_error_t_1} \Dvm_\rho \Dv_a t = \mi{L}_\rho \Dv_a t - \frac{1}{2} \gv^{ b c } \mi{L}_\rho \gv_{ c a } \Dv_b t = - \frac{1}{2} \gv^{ b c } \mi{L}_\rho \gv_{ a c } \Dv_b t \text{.}
\end{equation}
The third part of \eqref{eq.carleman_error_t2} now follows from \eqref{eq.aads_strong_limits}, \eqref{eq.aads_geom_limit}, the first part of \eqref{eq.carleman_error_t2}, and \eqref{eql.carleman_error_t_1}.

Next, using Definition \ref{def.aads_vertical_connection}, \eqref{eql.carleman_error_t_1}, and the observation that $\Dv t$ is $\rho$-independent, we expand
\footnote{By $( \Dvm_\rho )^2 \Dv_a t$, we mean the $a$-component of the vertical $1$-form $( \Dvm_\rho )^2 \Dv t$.}
\begin{align*}
( \Dvm_\rho )^2 \Dv_a t &= \mi{L}_\rho ( \Dvm_\rho \Dv_a t ) - \bar{\Gamma}_{ \rho a }^b \Dvm_\rho \Dv_b t \\
&= - \frac{1}{2} \mi{L}_\rho ( \gv^{ b c } \mi{L}_\rho \gv_{ a c } ) \Dv_b t - \frac{1}{2} \gv^{ b c } \mi{L}_\rho \gv_{ a c } \Dvm_\rho \Dv_b t \text{.}
\end{align*}
Thus, applying \eqref{eq.aads_strong_limits}, \eqref{eq.aads_geom_limit}, \eqref{eql.carleman_error_v_1}, and \eqref{eq.carleman_error_t2} to the above yields
\begin{align*}
| ( \Dvm_\rho )^2 \Dv t |_{ \varphi, 0 } &\lesssim_\varphi \mc{O} (1) \, ( | \mi{L}_\rho^2 \gv |_{ \varphi, 0 } + | \mi{L}_\rho \gv |_{ \varphi, 0 }^2 ) | \Dv t |_{ \varphi, 0 } + \mc{O} (1) \, | \mi{L}_\rho \gv |_{ \varphi, 0 } | \Dvm_\rho \Dv t |_{ \varphi, 0 } \\
&\lesssim_\varphi \mc{O} ( 1 ) \text{,}
\end{align*}
which completes the proof of \eqref{eq.carleman_error_t3}.
\end{proof}

\begin{lemma} \label{thm.carleman_error_hv}
The following estimates hold for the first derivatives of $\hv$:
\begin{equation}
\label{eq.carleman_error_hv1} | \Dv \hv |_{ \varphi, 0 } \lesssim_\varphi \mc{O} (1) \, | \Dm^2 t |_{ \varphi, 0 } + \mc{O} ( \rho ) \text{,} \qquad | \Dvm_\rho \hv |_{ \varphi, 0 } \lesssim_\varphi \mc{O} ( \rho ) \text{.}
\end{equation}
In addition, the second derivatives of $\hv$ satisfy
\begin{equation}
\label{eq.carleman_error_hv2} | \Dv^2 \hv |_{ \varphi, 0 } \lesssim_\varphi \mc{O} (1) \text{,} \qquad | ( \Dvm_\rho )^2 \hv |_{ \varphi, 0 } \lesssim_\varphi \mc{O} (1) \text{.} 
\end{equation}
\end{lemma}

\begin{proof}
First, we apply $\Dv$ and $\Dvm_\rho$ to the right-hand side of \eqref{eq.aads_riemann_vertical} and recall that $\Dv \gv$ and $\Dvm_\rho \gv$ vanish (see Proposition \ref{thm.aads_vertical_connection}).
From the above, \eqref{eq.aads_time}, \eqref{eq.carleman_error_t2}, and the fact that $\Dv^2 t \rightarrow^0 \Dm^2 t$, we obtain
\begin{align}
\label{eql.carleman_error_hv_1} | \Dv \hv |_{ \varphi, 0 } &\lesssim_\varphi \mc{O} (1) \, | \Dv^2 t |_{ \varphi, 0 } \lesssim_\varphi \mc{O} (1) \, | \Dm^2 t |_{ \varphi, 0 } + \mc{O} ( \rho ) \text{,} \\
\notag | \Dvm_\rho \hv |_{ \varphi, 0 } &\lesssim_\varphi \mc{O} (1) \, | \Dvm_\rho \Dv t |_{ \varphi, 0 } \lesssim_\varphi \mc{O} ( \rho ) \text{,}
\end{align}
as well as the second derivative bounds
\begin{align}
\label{eql.carleman_error_hv_2} | \Dv^2 \hv |_{ \varphi, 0 } &\lesssim_\varphi \mc{O} (1) \, | \Dv^3 t |_{ \varphi, 0 } + \mc{O} (1) \, | \Dv^2 t |_{ \varphi, 0 }^2 \text{,} \\
\notag | ( \Dvm_\rho )^2 \hv |_{ \varphi, 0 } &\lesssim_\varphi \mc{O} (1) \, | ( \Dvm_\rho )^2 \Dv t |_{ \varphi, 0 } + \mc{O} (1) \, | \Dvm_\rho \Dv t |_{ \varphi, 0 }^2 \text{.}
\end{align}
Both \eqref{eq.carleman_error_hv1} and \eqref{eq.carleman_error_hv2} now follow from Lemma \ref{thm.carleman_error_t}, \eqref{eql.carleman_error_hv_1}, and \eqref{eql.carleman_error_hv_2}.
\end{proof}

Since $\nablam_a \hv = \Dv_a \hv$ and $\nablam_\rho \hv = \Dvm_\rho \hv$ by Definition \ref{def.aads_mixed_connection}, then Lemma \ref{thm.carleman_error_hv} immediately implies
\begin{equation}
\label{eql.carleman_error_h_1} | \nablam_a \hv |_{ \varphi, 0 } \lesssim_\varphi \mc{O} (1) \, | \Dm^2 t |_{ \varphi, 0 } + \mc{O} ( \rho ) \text{,} \qquad | \nablam_\rho \hv |_{ \varphi, 0 } \lesssim_\varphi \mc{O} ( \rho ) \text{.}
\end{equation}
The first three parts of \eqref{eq.carleman_error_h} now follow from combining \eqref{eq.carleman_frame_asymp} with \eqref{eql.carleman_error_h_1}.

Next, from Definitions \ref{def.aads_vertical_connection} and \ref{def.aads_mixed_connection}, we deduce that
\[
\nablam_{ a b } \hv - \Dvm_{ a b } \hv = - \Gamma^\gamma_{ a b } \Dvm_\gamma \hv + \bar{\Gamma}^c_{ a b } \Dvm_c \hv \text{.}
\]
Combining the above with \eqref{eq.aads_strong_limits}, \eqref{eq.comp_prelim_christoffel}, and Lemma \ref{thm.carleman_error_hv}, we conclude that
\begin{equation}
\label{eql.carleman_error_h_2} | \nablam_{ a b } \hv |_{ \varphi, 0 } \leq | \Dv_{ a b } \hv |_{ \varphi, 0 } + | \Gamma^\rho_{ a b } | | \Dv_\rho \hv |_{ \varphi, 0 } \lesssim_\varphi \mc{O} (1) \text{.}
\end{equation}
Moreover, using again Definitions \ref{def.aads_vertical_connection} and \ref{def.aads_mixed_connection}, as well as \eqref{eq.comp_prelim_christoffel}, we expand
\[
\nablam_{ \rho \rho } \hv - ( \Dvm_\rho )^2 \hv = - \Gamma^\rho_{ \rho \rho } \Dvm_\rho \hv = \rho^{-1} \Dvm_\rho \hv \text{.}
\]
As a result of Lemma \ref{thm.carleman_error_hv} and the above, we obtain
\begin{equation}
\label{eql.carleman_error_h_3} | \nablam_{ \rho \rho } \hv |_{ \varphi, 0 } \leq | ( \Dvm_\rho )^2 \hv |_{ \varphi, 0 } + \rho^{-1} | \Dvm_\rho \hv |_{ \varphi, 0 } \lesssim_\varphi \mc{O} (1) \text{.}
\end{equation}
Finally, the last inequality in \eqref{eq.carleman_error_h} follows from \eqref{eq.aads_metric}, \eqref{eq.aads_strong_limits}, \eqref{eql.carleman_error_h_2}, and \eqref{eql.carleman_error_h_3}:
\[
| \Boxm \hv | \lesssim_\varphi \mc{O} ( \rho^2 ) \, | \nablam_{ \rho \rho } \hv | + \mc{O} ( \rho^2 ) \, \sup_{ a, b } | \nablam_{ a b } \hv | \lesssim_\varphi \mc{O} ( \rho^2 ) \text{.}
\]

\subsection{Proof of Lemma \ref{thm.carleman_error_A}} \label{sec.carleman_error_A}

The first part of \eqref{eq.carleman_error_A} follows from \eqref{eq.carleman_lambda}, \eqref{eql.carleman_weight}, \eqref{eq.carleman_asymp}, and \eqref{eql.carleman_LA}.
For the second part of \eqref{eq.carleman_error_A}, we apply \eqref{eql.carleman_S}, \eqref{eql.carleman_weight}, \eqref{eq.carleman_asymp}, and the first part of \eqref{eq.carleman_error_A}:
\begin{align}
\label{eql.carleman_error_A_2} - \frac{1}{2} \nabla_\beta ( \mc{A} S^\beta ) &= - \frac{1}{2} \nabla_\beta S^\beta \cdot \mc{A} - \frac{1}{2} \nabla_S \mc{A} \\
\notag - \frac{1}{2} \nabla_\beta S^\beta \cdot \mc{A} &= ( \kappa^2 - n \kappa + \sigma ) f^{ n - 2 } + ( 2 \kappa - n + p ) \lambda f^{ n - 2 + p } + \lambda^2 f^{ n - 2 + 2 p } + \lambda^2 \, \mc{O} ( f^n ) \text{,}
\end{align}
Furthermore, differentiating the equation for $\mc{A}$ in \eqref{eql.carleman_LA}, we obtain
\begin{align}
\label{eql.carleman_error_A_3} - \frac{1}{2} \nabla_S \mc{A} &= - \frac{1}{2} ( 2 F' F'' + F''' ) S f ( g^{ \alpha \beta } \nabla_\alpha f \nabla_\beta f ) - [ ( F' )^2 + F'' ] g^{ \beta \mu } S^\alpha \nabla_{ \alpha \beta } f \nabla_\mu f \\
\notag &\qquad - \frac{1}{2} F'' S f \Box f - \frac{1}{2} F' \, S ( \Box f ) \\
\notag &= \mc{A}_1 + \mc{A}_2 + \mc{A}_3 + \mc{A}_4 \text{.}
\end{align}

From \eqref{eq.carleman_f_grad}, \eqref{eq.carleman_fstar}, \eqref{eq.carleman_lambda}, \eqref{eql.carleman_S}, \eqref{eql.carleman_weight}, and \eqref{eq.carleman_asymp}, we obtain
\begin{align}
\label{eql.carleman_error_A_4} \mc{A}_1 &= ( \kappa^2 - \kappa ) f^{ n - 2 } + \left[ ( 2 - p ) \kappa - \frac{1}{2} ( 1 - p ) ( 2 - p ) \right] \lambda f^{ n - 2 + p } \\
\notag &\qquad + ( 1 - p ) \lambda^2 f^{ n - 2 + 2 p } + \lambda^2 \, \mc{O} ( f^n ) \text{,} \\
\notag \mc{A}_3 &= - \frac{ n - 1 }{2} \kappa f^{ n - 2 } - \frac{ ( n - 1 ) ( 1 - p ) }{2} \lambda f^{ n - 2 + p } + \lambda \, \mc{O} ( f^n ) \text{.}
\end{align}
For $\mc{A}_2$, we also take into account the expansion for $\nabla_{ N N } f$ in \eqref{eql.carleman_pseudoconvex_pi_20}:
\begin{equation}
\label{eql.carleman_error_A_6} \mc{A}_2 = - ( \kappa^2 - \kappa ) f^{ n - 2 } - ( 2 \kappa - 1 + p ) \lambda f^{ n - 2 + p } - \lambda^2 f^{ n - 2 + 2 p } + \lambda^2 \, \mc{O} ( f^n ) \text{.}
\end{equation}
The remaining term $\mc{A}_4$ requires more work; for this, we differentiate the formula \eqref{eq.carleman_f_grad} for $\Box f$, and we recall Definition \ref{def.aads_strong}, Proposition \ref{thm.aads_geom_limit}, and Remark \ref{rmkl.carleman_bounded} in order to obtain
\begin{equation}
\label{eql.carleman_error_A_7} \mc{A}_4 = \frac{ n - 1 }{2} \kappa f^{ n - 2 } + \frac{ n - 1 }{2} \lambda f^{ n - 2 + p } + \lambda \, \mc{O} ( f^n ) \text{.}
\end{equation}
Finally, combining \eqref{eql.carleman_error_A_2}--\eqref{eql.carleman_error_A_7} results in the second equation in \eqref{eq.carleman_error_A}.

\fi

\raggedright
\bibliographystyle{amsplain}
\bibliography{articles,books,misc}

\end{document}